\documentclass[a4paper,onecolumn,11pt,accepted=2024-12-09]{quantumarticle}
\pdfoutput=1
\usepackage[standard]{ntheorem}
\usepackage[usenames,dvipsnames]{xcolor}
\usepackage[bookmarks=true,colorlinks=true,urlcolor=Mahogany,linkcolor=Mahogany,citecolor=Mahogany,plainpages=false,pdfpagelabels]{hyperref}

\usepackage[most]{tcolorbox}
\definecolor{block-gray}{gray}{0.89}
\newtcolorbox{blockquote}{colback=block-gray,grow to right by=-1mm,grow to left by=-1mm,boxrule=0pt,boxsep=0pt,breakable}

\newtheorem{claim}{Claim}

\usepackage{mathrsfs}
\usepackage{utfsym}
\usepackage{mathtools}
\usepackage{tikz}
\usetikzlibrary{fit,shapes.arrows,shapes.callouts,shapes.geometric,petri}
\pgfdeclarelayer{background}
\pgfdeclarelayer{foreground}
\pgfsetlayers{background,main,foreground}

\usepackage{halloweenmath}
\usepackage[misc,weather]{ifsym}

\usepackage[
	n,
	operators,
	advantage, 
	sets,
	adversary,
	landau,
	probability,
	notions,
	logic,
	ff,
	mm,
        primitives,
        events,
        complexity,
        asymptotics, 
        keys]{cryptocode}
\newcommand{\qpt}{\ensuremath{\mathsf{QPT}}}

\usetikzlibrary{shapes.callouts,fit}

\newcommand{\qed}{\hfill$\blacksquare$}

\newcommand{\lxy}[2]{\ensuremath{L(#1,#2)}}

\usepackage[short]{optidef}     % SDP: maxi/mini environments
\usepackage{qip}

\usepackage{trace}

\usepackage{graphicx}
\usepackage{amsmath}
\usepackage{epstopdf}
\usepackage{mathtools}

\usepackage[maxbibnames=99,backend=biber,giveninits=true,style=alphabetic,url=false,isbn=false]{biblatex}
\AtEveryBibitem{% Cacher certains champs
  \clearfield{note}%
  \clearname{editor}%
  \clearlist{location}
}

\DeclareFieldFormat{eprint:iacr}{%
  IACR eprint\addcolon\space
  \ifhyperref
    {\href{http://eprint.iacr.org/#1}{\texttt{#1}}}
    {\texttt{#1}}}

\newenvironment{defi}{\begin{definition}}
{\renewcommand{\qed}{\hfill\ensuremath{\square}}\qed\end{definition}}

\DeclareMathAlphabet{\mathbbold}{U}{bbold}{m}{n}
\renewcommand{\hilbert}{\ensuremath{\mathscr{H}}}
\newcommand{\idt}{\ensuremath{}\mathbbold{1}}
\renewcommand{\id}{\idt}

\newcommand{\spr}{\ensuremath{\Sigma\text{--protocol}}}
\newcommand{\sprs}{\ensuremath{\Sigma\text{--protocols}}}
\newcommand{\sprnm}{\ensuremath{\Sigma_{n,m}\text{--protocol}}}
\newcommand{\sprsnm}{\sprnm s}
\newcommand{\siguni}{\ensuremath{\Sigma}\text{--universal}}
\newcommand{\siguninm}{\ensuremath{\Sigma_{n,m}}\text{--universal}}
\newcommand{\fsl}{\ensuremath{\text{\sf FS}}}

\newcommand{\FSH}[2]{\ensuremath{\Pi^{\fsl}_{#2}[#1]}}
\newcommand{\QFS}[1]{\ensuremath{\Pi^{\sf QFS}}[#1]}
\newcommand{\QFSs}{\ensuremath{\Pi^{\sf QFS}}}

\newcommand{\lossykeys}[1]{\ensuremath{S_{\text{lossy}}}}
\newcommand{\injectivekeys}[1]{\ensuremath{S_{\text{inj}}}}
\newcommand{\wotro}{\textsf{WOTRO}}
\newcommand{\pwotro}[1]{\ensuremath{\Pi_{\textsf{WRO}}^{#1}}}
\newcommand{\rom}{\ensuremath{\text{ROM}}}
\newcommand{\crs}{\ensuremath{\text{CRS}}}
\newcommand{\crrs}{\ensuremath{\text{CR\$}}}
\newcommand{\crqs}{\ensuremath{\text{CRQ\$}}}
\newcommand{\crqss}{\ensuremath{\text{CRQS}}}

\newcommand{\sff}[1]{\ensuremath{\mathsf{#1}}}

\newcommand{\dprover}{\ensuremath{\tilde{\prover}}}

\newcommand{\negli}[1]{\ensuremath{\mathsf{negl}(#1)}}

\newcommand{\oracle}{\ensuremath{\mathcal{O}}}
\newcommand*{\EE}{\mathrm{e}}

\newcommand{\epr}[2]{\ensuremath{\textsf{EPR}^{#1}_{#2}}}

\usepackage{tabularx}

\newcommand{\N}{\ensuremath{\mathbb{N}}}
\newcommand{\fq}{\mathbb{F}_q}

\newcommand{\Bg}{\ensuremath{\mathcal{B}}}

\newcommand{\dens}[1]{\ensuremath{\text{D}(#1)}}

\newcommand{\thunder}{\ensuremath{\text{\RainCloud}}}

\newcommand{\lightn}{\ensuremath{\text{\Lightning}}}
\newcommand{\qlgen}{\ensuremath{\mathsf{QLGen}}}
\newcommand{\qlsetup}{\ensuremath{\mathsf{QLSetup}}}
\newcommand{\qlver}{\ensuremath{\mathsf{QLVer}}}
\newcommand{\tqlgen}{\ensuremath{\mathsf{tQLGen}}}
\newcommand{\tqlsetup}{\ensuremath{\mathsf{tQLSetup}}}
\newcommand{\tqlver}{\ensuremath{\mathsf{tQLVer}}}

\newcommand{\verf}{\ensuremath{\mathsf{Ver}}}
\newcommand{\ql}{\ensuremath{\mathsf{QL}}}
\newcommand{\tql}{\ensuremath{\mathsf{tQL}}}

\newcommand{\mach}[1]{\ensuremath{\mathsf{#1}}}
\newcommand{\macho}[2]{\ensuremath{\mach{#1}^{#2}}}

\newcommand{\collshn}{\ensuremath{\mathbb{G}^{n,m}_\Gamma}}

\newcommand{\fbb}{$f$-BB}

\newcommand{\mytag}[2]{%
  \text{#1}%
  \@bsphack
  \begingroup
    \@onelevel@sanitize\@currentlabelname
    \edef\@currentlabelname{%
      \expandafter\strip@period\@currentlabelname\relax.\relax\@@@%
    }%
    \protected@write\@auxout{}{%
      \string\newlabel{#2}{%
        {#1}%
        {\thepage}%
        {\@currentlabelname}%
        {\@currentHref}{}%
      }%
    }%
  \endgroup
  \@esphack
}

\definecolor{mygreen}{rgb}{0,0.6,0}
\definecolor{mygray}{rgb}{0.1,0.1,0.1}
\definecolor{mymauve}{rgb}{0.58,0,0.82}

\tikzstyle{porte} = [fill=blue!25, draw]

\allowdisplaybreaks

\title{Fiat-Shamir for Proofs Lacks a Proof Even in the Presence of Shared Entanglement}
\author{Frédéric Dupuis}
\affiliation{Université de Montréal (DIRO), Montréal, Canada}
\author{Philippe Lamontagne}
\affiliation{National Research Council Canada, Ottawa, Canada}
\affiliation{Université de Montréal (DIRO), Montréal, Canada}
\author{Louis Salvail}
\affiliation{Université de Montréal (DIRO), Montréal, Canada}
\begin{document}
\maketitle

\begin{abstract}
We explore the cryptographic power of arbitrary shared physical resources. The most general such resource is access to a fresh entangled quantum state
at the outset of each protocol execution. We call this the \emph{Common Reference Quantum State (\crqss)} model, in analogy to the well-known \emph{Common Reference String (\crs)}. The \crqss{} model  is a natural generalization of the \crs{} model but appears to be more powerful: in the two-party setting, a \crqss{} can sometimes exhibit properties associated with a Random Oracle queried once by measuring a maximally entangled state in one of many mutually unbiased bases.
We formalize this notion as a \emph{Weak One-Time Random Oracle} (\wotro{}), where we only ask of the $m$--bit output to have \emph{some} randomness when conditioned on the $n$--bit input.

We show that when $n-m\in\omega(\lg n)$, any protocol for \wotro{} in the
\crqss{} model can be attacked by an (inefficient) adversary. Moreover, our
adversary is efficiently simulatable, which rules out the possibility of proving
the computational security of a scheme by a fully black-box reduction to a cryptographic game assumption.
% We show that \wotro{} with $n-m\in\omega(\lg n)$ is \emph{black-box impossible} in the \crqss{} model, meaning that no protocol  can have its security black-box reduced to a cryptographic  game. 
% %between a challenger and an adversary.
% We define a (inefficient) quantum adversary against any \wotro{} protocol that can be efficiently simulated in polynomial time, ruling out any reduction to a secure game that only makes black-box queries to the adversary.
On the other hand, we introduce a non-game quantum assumption for hash functions that implies \wotro{} in the \crqs{} model (where the \crqss\ consists only of EPR pairs). We first build a statistically secure \wotro{} protocol where $m=n$, then hash the output.

The  impossibility of \wotro{} has the following consequences. First, we show
the fully-black-box impossibility of a \emph{quantum} Fiat-Shamir transform,
extending the impossibility result of Bitansky \emph{et al.} (TCC 2013) to the
\crqss{} model. Second, we show a fully-black-box impossibility result for a
strenghtened version of quantum lightning (Zhandry, Eurocrypt 2019) where quantum
\emph{bolts} have an additional parameter that cannot be changed without
generating new bolts. Our results also apply to $2$--message protocols in the plain model.
\end{abstract}

\tableofcontents

\section{Introduction}
\label{sec:intro}
Cryptographic protocols can sometimes only be proven secure if some of their components are assumed to be ideal. For example, some protocols that make use of cryptographic hash functions can be proven secure if they are modelled as ideal random functions provided as a black box; this is called the \emph{random oracle model (ROM)}.  Another, but weaker, idealized resource is the \emph{common random string model (CRS)}, in which the participants get a freshly generated random string at the outset of each protocol execution.  Several cryptographic applications have  their most efficient protocols proven secure when provided access to such extra  resources, as all known protocols in the plain model are either  inefficient, or do not satisfy all security requirements. 
%The ROM  can be much more powerful than  the CRS model as it 
%providess random access to an exponential number of independent CRS. 

\paragraph{The Random Oracle Model (ROM).} 
Introduced by Bellare and Rogaway\,\cite{BR93} as a way to idealize cryptographic hash functions, 
the model has  been shown  to provide formal security proofs for a wide variety 
of   cryptographic protocols  that are not known to be  secure
under standard assumptions in the plain model. A random oracle
models a hash function
as one whose value for every input is chosen uniformly and independently at random
and afresh before each protocol execution. 
This is meant to model the assumption that a hash function is random, and that looking at its source code yields nothing useful beyond its input-output behaviour. Rigorous security proofs for practical and efficient applications like \emph{Full Domain Hash signatures (FDH-Signatures)}, \emph{Optimal Asymmetric Encryption Padding (OAEP)}, Schnorr's signatures\,\cite{Schnorr89,Seurin12}, and Fischlin's NIZK-PoK\,\cite{Fischlin05} are easy to obtain in the ROM but are still missing in the
plain model. The random oracle is a powerful primitive that provides all  the main  properties  of a cryptographic hash function at once: collision resistance, preimage resistance, and pseudorandomness. It also has properties that can never be satisfied by any hash function: programmability,  (query) extractability (also known as observability), and freshness.   

\paragraph{Common Reference String Model.}
A CRS  is nothing more than a fresh random string that
materializes upon each protocol execution (freshness) and to which all players
have access. This model was originally proposed by Blum, Feldman, and
Micali\,\cite{BFM88} to help remove interaction in zero-knowledge proof systems. 
 In \cite{BSMP91},
the model was shown to allow for non-interactive zero-knowledge 
for all NP languages.
The works of~\cite{Can01,CF01,DN02} extend  its  use as 
a resource enabling universally composable cryptographic
primitives. The common reference string model comes in
two main flavours. The weakest  consists of a random and uniform string of
polynomial length (in the security parameter) while the strongest 
consists of a string of polynomial length picked from some efficiently sampleable distribution. The
first flavour will be denoted by the \emph{\crrs\ model (i.e. the Common Random
  String Model)} while the second flavour will be denoted by the \emph{\crs\
  model (i.e. the Common Reference String Model)}.  

A customary application of both  the CRS model and the ROM is  the 
removal of interaction in  interactive proof systems. As mentioned above,
the CRS model was originally designed  for that purpose~\cite{BFM88}. 
Notice that 
 a random oracle  is a much more powerful  resource than  a CRS, since it 
provides random access to an exponential number of them. 
However, a random oracle is an immaterial resource as its properties
could never be satisfied by any efficient local process. This is in sharp
contrast to a CRS, which can be implemented in practice: we only need a way to publish fresh and public random strings of polynomial length.
%This cannot
%be done for a random oracle as its description is of exponential
%size. 
%There is no way to simulate  a random oracle by an efficient and 
%pre-computation. 
Unfortunately, 
some basic and useful cryptographic primitives  
are only known to be securely realizable  in the ROM.

\paragraph{When Entanglement Behaves {Like} a Random Oracle.}
In order to see why entanglement could outperform a $\crs$ in some settings,
consider the following scenario where it seems to provide as much
randomness as the random oracle. Suppose Alice prepares $n$ EPR pairs of qubits and sends half of each pair to Bob. Each can then view their $n$ qubits as an access
to a weak random oracle implementing a random function
$f:\{0,1\}^n\rightarrow\{0,1\}^n$. The value $f(a)$ can be obtained the
following way. To each possible value $a\in\{0,1\}^n$, we associate a publicly
known orthonormal basis $\theta_a$ for $n$ qubits. The value of $f(a)$ is simply
defined as the outcome of the measurement of the $n$ qubits owned by each party
in basis $\theta_a$. Notice that this weak random oracle can be queried in only
one place by each party, as after the measurement is performed, the entangled
pairs have collapsed to a classical state. However, when both parties measure in
the same basis $\theta_a$ they obtain the same uniformly distributed outcome.
Moreover, when the bases $\{\theta_a\}_{a\in\{0,1\}^n}$ are chosen to be
mutually unbiased\,\cite{Schw60,WF89}\footnote{ $\{\theta_a\}_{a\in\{0,1\}^n}$
  is a set of mutually unbiased bases for $n$ qubits if for all $\ket{u}\in
  \theta_a$ and $\ket{v}\in \theta_{a'}$ with $a\neq a'$, we have
  $|\bracket{u}{v}|^2=2^{-n}$. There are $2^n+1$ mutually unbiased bases for $n$
  qubits.}, the value $f(a)$ does not provide any information about $f(a')$ for
any $a\neq a'$.

In this particular setting, $n$ EPR pairs seem to contain as much randomness as
a random oracle. It is therefore tempting to believe that an entangled state of
polynomial size could in certain cases provide a cryptographic resource
tantamount to the random oracle when only one query (or just a few) has to be
made by each player. Such a resource, which we call a \emph{Weak One-Time Random
  Oracle} (\wotro{}), would be a powerful primitive for removing interaction in
procotols, even if it only provides \emph{some} randomness: that the value
$f(a)$ is not a deterministic function of $a$. The above scheme can be made
non-interactive if Alice and Bob share EPR pairs ahead of time. This motivates
our study of a model in which parties have access to a pre-shared entangled
state.

\paragraph{The \crqss\ and \crqs\ Models.}
In this paper, we % consider a quantum version of the CRS model, 
% called the \crqss\ model, and we ask whether it could  go beyond 
% what the CRS model can provide for the safe  removal of interaction 
% in  cryptographic protocols.
% We 
consider models where a quantum state plays the role of a common random
string in a situation involving two parties. In the \crqss\ (\emph{Common
  Reference Quantum State}) model, each party receives one half of a fixed pure
quantum state at the beginning of each protocol execution. The shared quantum
state is of polynomial size and can be generated by some polynomial size quantum
circuit. In the \emph{CRQ\$} model, each player is given halves of polynomially
many (in the security parameter) maximally entangled pairs of qubits (or qudits
in general). Although we could allow a \crqss\ or a \crqs\ to be shared between
more than two parties, in this work we only consider the two-party case. Notice
that the meaning of \emph{common} in \crqss\ and \crqs\ is narrower than for a
\crs\ and \crrs: even though a \crqss\ is common to both parties involved in a
protocol, it is completely unknown to anybody else, as both players share a pure
state. Even though a \crqss\ is obviously more difficult to deploy in practice
than a \crs, it remains a physical resource, unlike the random oracle.
Establishing limits on what a \crqss\ can provide would therefore contribute to
a better understanding of the cryptographic power provided by the sharing of a
\emph{physical} resource between the parties involved in a protocol.

\paragraph{\wotro{} in the \crqss{} model?}
We investigate the question of whether or not \wotro{} has a secure
instantiation in the \crqss{} model. % Why would it be possible given the strong
% impossbility results in the CRS model?
% For one thing,
Like the CRS and ROM, quantum entanglement is
known to allow the reduction of interaction, but it also enables
tasks that would be classically impossible using only a CRS. Watrous~\cite{watrous_pspace_2003} showed that every language in PSPACE
has $3$-message proof systems. Another example would be nonlocal games such as
the magic square game~\cite{a02,a03,bbt04}, where a pair of entangled
non-interacting provers can win a game that would classically require them to
communicate.

The \crqss{}\ model provides quantum non-local correlations\footnote{
In quantum mechanics, a \emph{non-local correlation} is the name
given to the statistics of local measurements applied to distinct
parts of a quantum states when they cannot be explained by a local 
realistic theory. Non-local correlations here (quantum or not) means
also that they do not allow for any form of communication as they must
be compatible with special relativity. 
}
between 
the prover and the verifier.
%  As we mentioned above,  
% $n$ shared EPR pairs measured
% in one of $2^n$ mutually unbiased bases seems to provide a functionality
% reminiscent of a random oracle queried at a single point, exactly as needed
% for Fiat-Shamir. The correlated randomness available  
% to the prover and the verifier seems to contain  as much  randomness as the 
% random oracle.
 Non-local correlations 
are often idealized by non-local (PR) boxes~\cite{PR97}. One PR box 
takes the first party's input $a\in\{0,1\}$ and the second party's input $b\in\{0,1\}$ 
to provide  $u\in\{0,1\}$ and $v\in\{0,1\}$  such that  $u\oplus v = a\wedge b$ 
to the first and second party respectively. EPR pairs achieve this functionality 
with probability of success $\cos^2{(\frac{\pi}{8})}$ while any \crs\ 
would not be able to provide the correct answer with probability better than $\frac{3}{4}$.
It is not too hard to see that access to sufficiently many PR boxes allows for a
secure implementation for \wotro{} (see details in Section~\ref{sec:wotro-from-non}).
While PR boxes are not physical objects, % can we obtain a similar outcome using quantum non-local correlations provided by a \crqss{}?
the question we are addressing here is
whether non-local \emph{quantum} correlations can be harnessed
to provide a functionality akin to the use of a random oracle queried once through the use of a \crqss{}.

One might argue that the \crqss{} model is not currently realistic given the technological
difficulties associated with distributing and coherently storing quantum
entanglement (although this is rapidly improving). However, we ask a more fundamental question on the power of setup
assumptions.
Are there physically realizable setup assumption
that allows to solve problems that, in the classical model, appear to require
a random oracle.

\paragraph{The Fiat-Shamir Transform.} One very useful 
primitive that needs an idealized cryptographic
resource like \wotro{} for establishing its security is the Fiat-Shamir transform, also
known as the Fiat-Shamir \emph{heuristic}, introduced in the pioneering
work of Fiat and Shamir in \cite{fs86} as a  way to transform 
identification schemes of a given form into practical  digital signature schemes.
More generally,
the FS-transform is a  simple and efficient primitive allowing to convert
sound interactive proof systems of a particular form into  non-interactive
arguments for the same language.
%The FS-transform is a cryptographic hash function
%with extra properties that the random oracle can provide.
Its primary use 
 is to remove interaction
 in \sprs.

 \sprs\,\cite{CramerPhD,IvanSig10} are public-coin 3-message proof systems
where, from public input $x\in\{0,1\}^*$,  
the prover sends a \emph{commitment} $a\in \{0,1\}^n$
to the verifier as the first message. The verifier then 
replies with a random \emph{challenge} $c\in_R\{0,1\}^m$ (called \emph{public coins})
before the prover sends the answer $z(x,a,c)$
that the verifier can check for consistency.  Henceforth, \sprs\ with commitments of size $n$ and
public coins of size $m$ will be denoted by \emph{\sprsnm}. 
These proof systems can be proofs of knowledge, like their    use in
identification schemes, or proofs
of language membership. In this paper, 
 \sprs\ are always considered perfectly 
correct and special sound.
Special soundness\footnote{Special soundness is called \emph{optimal soundness}
in \cite{barak_lower_2006}.} for proofs of knowledge means that 
from any two successful conversations with the same commitment
$(a,c,z(x,a,c))$ and
$(a,c',z(x,a,c'))$ with $c\neq c'$, one can efficiently 
extract a witness $\pi$ for $x\in L$. For proofs of language
membership, special soundness means that when $x\notin L$ and
for each commitment $a$,
there exists at most one challenge $\sff{c}(a)$ for which 
a third message $\tilde{z}$  can ever be found such that
$(a,\sff{c}(a), \tilde{z})$ is accepted by the verifier.

The Fiat-Shamir transform applied to a \spr\ 
is implemented using  hash function $h_r:\{0,1\}^{*}\rightarrow\{0,1\}^m$
picked according to CR\$ $r$.
The prover then sends $(a,h_r(a), z(x,a,h_r(a)))$ 
to the verifier.  In other words, the verifier's challenge or public coin $c$
in the \spr\ is replaced 
by $c=h_r(a)$\footnote{Some works include the public instance $x$ as input to $h_r$,
  our results remain untouched if we include it.  We leave it out
  for simplicity.}. It is straightforward to see that when
$h_r$ is modelled by a random oracle, the
transform applied to a \spr\ produces a  sound  argument
The family of hash functions $\mathcal{H}=\{h_r\}_{r\in_D \{0,1\}^{\ell(n)}}$, 
for $D$ an efficiently sampleable distribution over $\{0,1\}^{\ell(n)}$, 
is a \emph{sound \siguninm\ instantiation} of
the Fiat-Shamir transform if $h_r$
  converts the special soundness of any \sprnm\ (as a proof of language membership)  
into a non-interactive argument.
Notice that when the hash function is modelled by  a random oracle, 
the prover and the verifier only have to query the oracle once at the same point.
Replacing the random oracle with a secure instantiation of \wotro{} would thus provide
a sound universal Fiat-Shamir  transform.

\paragraph{The Fiat-Shamir Transform in the ROM and QROM.}
%\subsection{Comparison with earlier results}
As mentioned above,  the Fiat-Shamir transform was 
shown secure in the ROM by Pointcheval and Stern~\cite{pointcheval_security_1996} 
in 1996. The soundness of the  Fiat-Shamir transform is straightforward in the ROM.
The challenging part was to show that it also provides 
non-interactive proofs of knowledge. The same was shown to hold in the
quantum random oracle (QROM)
independently and differently by Don, Fehr, Majenz, and Schaffner in \cite{don_security_2019}  
and  by Liu and Zhandry in \cite{liu_revisiting_2019}.
\
\paragraph{Known Impossibility Results for the Fiat-Shamir Transform.}
The Fiat-Shamir 
transform does not guarantee computational soundness  
for all \sprs\ in the \crs\ model. 
 In particular, 
Goldwasser and Kalai have shown that the Fiat-Shamir
transform applied to some (contrived)  \sprs\  is not
sound for any  instantiation of the hash
function (i.e. instantiated using a \crs)\,\cite{GoldwKalai2003}. 
However, 
this impossibility result requires the \spr\
to be a proof of knowledge.

\citeauthor{aru14}~\cite{aru14} have shown that the Fiat-Shamir transform
cannot preserve the soundness of every \spr\ against quantum
adversaries, even when it is instantiated with a random oracle. More precisely,
they construct a proof system, which can be either a proof of knowledge or an
argument of language membership, which is sound classically
but unsound against quantum adversaries. The same holds true when the Fiat-Shamir transform is applied to these proof systems. In effect, their attack is against the
underlying \spr{} rather than against a physical instantiation of a random
oracle. Their results do not contradict the positive results of~\cite{don_security_2019,liu_revisiting_2019} since they show that the Fiat-Shamir transform preserves soundness in the QROM when the 
underlying \spr{} is sound against quantum adversaries.

Impossibility results for \sprs\ 
used as proofs of language membership 
are not known to be as strong as for proofs of knowledge.
One reason being that for language membership, 
the Fiat-Shamir transform is only asked to provide computational
soundness to a  \spr\  with statistical soundness whereas for a proof of knowledge the 
\spr\ is an argument.
Remember that
a \emph{cryptographic game}~\cite{HH09} is a standard way to define 
computational assumptions by requiring
that no adversary can win an interactive game against a  \emph{challenger}
with probability that is not overwhelmingly close to some constant value\,\cite{HH09}.
%In a cryptographic game, the adversary is interacting with 
%a  \emph{challenger} asking the adversary to carry out some 
%computation on random instances and winning means that the computation
%has been completed by the adversary and verified with success by the challenger. 
An assumption 
that can be formulated as  a cryptographic game with 
an efficient challenger  is
called a \emph{falsifiable assumption}\,\cite{GW11, Naor03}.  
Known impossibility results for the Fiat-Shamir transform applied to \sprs\
for proofs of language membership are about the impossibility 
of \emph{black-box} reducing its computational soundness to a cryptographic game.

In\,\cite{bitansky_why_2013}, Bitansky \emph{et al.} 
provide two  results on the impossibility of establishing the computational
soundness of  the Fiat-Shamir transform  in the \crs\ model.
First,  
if a language $L\notin BPP$ has an honest-verifier zero-knowledge (HVZK) \spr\
%\spr\footnote{In fact, the impossibility
%result also holds for honest verifier 3-message public coin protocols.
%That is,   \sprs\ as we define them in this paper minus the special soundness condition.} 
(with small
enough challenges) 
then the soundness of the Fiat-Shamir transform applied to it 
cannot be established by a black-box reduction\footnote{The security of protocol $\Pi$
is black-box reduced to an assumption expressed as a game if there exists
an oracle polynomial-time machine $\mathcal{R}^{P^*}$ 
that, with oracle access to any successful adversary $P^*$ for protocol $\Pi$, 
wins the game.}  to a falsifiable assumption\footnote{The reason why the result applies in the \crs\ model
is  because \cite{bitansky_why_2013,djkl12}  show how to get, 
from such a Fiat-Shamir transform, a
2-message zero-knowledge proof system for $L$ where the verifier
simply sends the identity of the hash function to the prover as first message. This is equivalent
to non-interactive schemes in the CRS model. These proofs systems are shown impossible
by an extension of the impossibility result for 2-round zero-knowledge for non-trivial languages 
by Goldreich and Oren~\cite{GoldOren1994}.}.  This impossibility result applies even
to Fiat-Shamir transforms tailor-made for specific \sprs.
%This impossibility result holds  
%in the CRS model.
Second,  they show the impossibility of black-box reducing the computational soundness of any 
universal instantiation 
of the Fiat-Shamir transform  to a 
cryptographic game, even a non falsifiable one where the challenger is not required to run in polynomial time.
%While this  second impossibility result
%applies only to $\siguninm$ instantiation of the Fiat-Shamir transform,
% the first one applies even to a Fiat-Shamir transform tailor made 
%for a particular  \spr.
%%%
% Our main contribution consists in showing that the second impossibility result
% of \cite{bitansky_why_2013} also
% holds in the $\crqss$ model  
% %In other words, does a \crqss\ used as a cryptographic resource 
% %allow for more  than a \crs?
% %We provide an answer to  this question with respect to  the 
% %well known Fiat-Shamir transform. 
% even though sharing an entangled quantum state   
% seems to provide enough randomness to mimic a (classical)  random oracle.
% In other words, the computational soundness of any universal Fiat-Shamir implementation in the \crqss\ model cannot be  black-box reduced 
% to any cryptographic game, just like in the \crs\ model.
%%% Mis en commentaire parce qu'on répète plus bas et qu'on a repositionné notre
%%% résultat principal depuis. 

\paragraph{Positive results \& related work.}

A series of results have been focusing 
on achieving soundness of
the Fiat-Shamir transform  from a cryptographic
assumptions that cannot be black-box reduced to 
cryptographic games.
%using non-falsifiable assumptions to
%circumvent impossibility results of~\cite{bitansky_why_2013}. 
Barak, Lindell
and Vadhan\,\cite{barak_lower_2006} introduce the notion of \emph{entropy
  preserving} hash functions (such function families are rather said to \emph{ensure
  conditional entropy} therein) 
  and show that assuming their existence, 
  there is no  
  constant-round auxiliary-input zero-knowledge proof system for non-trivial languages.
  The proof of this result implies the computational soundness 
  of Fiat-Shamir using entropy preserving hash functions. 
 Later, 
Dodis, Ristenpart and Vadhan\,\cite{drv12} gave a construction for entropy preserving
hash functions assuming the existence of robust randomness condensers with some extra properties, but without
providing any candidate construction.  Canetti, Goldreich, and Halevi\,\cite{CGH04} introduce
\emph{correlation intractable} families of hash functions.
Correlation
intractability is related to entropy preservation as the latter  implies the former.
% with respect to relations
%that are not necessarily efficiently computable.
Therefore, a consequence of~\cite{bitansky_why_2013} is  that  correlation intractability  cannot be proven by  black-box reduction to a
game.  
%
%and Canetti, Chen, and Reyzin\,\cite{canetti_correlation_2016} provide a secure
%construction  with respect to sparse relations computable in fixed polynomial time 
%based on input-hiding obfuscators and puncturable pseudorandom functions.
%Correlation intractable families of hash functions allow for universal instantiations
%of the Fiat-Shamir transform. 
In~\cite{KRR17}, Kalai, Rothblum, and Rothblum provide a
construction for  correlation intractable family of hash functions
from 
 a subexponentially secure indistinguishability obfuscator,
an exponentially secure input-hiding obfuscator for the class
of multi-bit point functions, and the existence of a subexponentially secure 
puncturable PRF\footnote{Notice that the result of \cite{KRR17} is very general
as it allows to apply securely the Fiat-Shamir transform to any public-coin 3-message proof systems,
not only to $\sprs$ as we define them.  Some of their assumptions can be relaxed a little
when the Fiat-Shamir transform is applied to \sprs.
%When restricted to $\sprs$, the existence
%of sub-exponentially secure 
%puncturable PRF does not seem to be required anymore.
}. 
The subexponential indistinguishable  security of the 
IO-obfuscator and the exponential security of the multi-bit point functions obfuscator 
allow to evade the impossibility result of  \cite{bitansky_why_2013}.
In \cite{CCHLRR18},  Canetti, Chen, Holmgren, Lombardi, Rothblum, and Rothblum  
show how to construct a universal instance of the Fiat-Shamir transform
using correlation intractable hash functions built from a strong version 
of KDM-encryption. The resulting Fiat-Shamir transform also has security
black-box reducible to a cryptographic game with subexponential security.
%Then, 
%Canetti, Chen, Reyzin, and Rothblum\,\cite{canetti_fiat-shamir_2018} show how 
%to construct them from
%KDM-secure encryption schemes. 
%Notice that limited to fixed polynomial time relations, correlation intractable
%families of hash functions do not allow for a computationally sound
%\siguni\
%Fiat-Shamir transform.  

The concept of shared entanglement as a setup was considered in previous works. 
In \cite{CVZ20}, Coladangelo, Vidick, and Zhang  have shown how to 
design non-interactive zero-knowledge arguments for QMA (i.e.\ quantum NP), with preprocessing.
\citeauthor{morimae_classically_2022}~\cite{morimae_classically_2022} use a
similar setup for classical verifiability of NIZK arguments for QMA. 
The preprocessing is essentially what we call here a CRQ\$. 
Non-interactivity  is obtained from pre-shared EPR pairs used as a teleportation 
channel. This can be viewed as a quantum version of the work of
Peikert and Shiehian~\cite{peikert_noninteractive_2019}
and, as such, is not  a \siguninm\  instantiation of the Fiat-Shamir transform.
The ability of a  \crqss\   to provide zero-knowledge against 
quantum dishonest verifiers has been 
investigated in~\cite{DFS04}. It  was shown that a \crqss\
allows  quantum zero-knowledge implementations of a \sprs\
against a relaxed form of honest verifiers, called \emph{non-oblivious}.

A model called \crqss\ was recently\footnote{The authors
  of~\cite{morimae_unconditionally_2023} were aware of our work but decided to
  use the same naming scheme, arguing that our model is mode akin to the quantum
  version of correlated randomness rather than CRS. } introduced
in~\cite{morimae_unconditionally_2023} as a trusted setup for provably
computationally secure quantum bit commitment (i.e.\ without relying on
complexity assumptions). In the model of~\cite{morimae_unconditionally_2023}, a
setup algorithm samples a classical key $k$ and distributes copies of a quantum
state $\ket {\psi_k}$ to each party. A similar model, called ``unclonable common
random state'', is independently introduced in~\cite{qian_unconditionally_2023}
for the same task of unconditionally secure quantum bit commitment. In the
context of a negative result, the more general the model is, the stronger is the
impossibility result. The models of
\cite{morimae_unconditionally_2023,qian_unconditionally_2023} are a special case
of ours by considering a \crqss{} of the form $\ket{\Psi}=\sum_k\sqrt{p(k)}
\ket{\psi_k}\ket{\psi_k}$ where $p(k)$ denotes the probability to  pick key $k$.

\subsection{Our Contributions}
\label{sec:results}

We introduce a cryptographic primitive called a
\emph{Weak One-Time Random Oracle}, denoted $\wotro^{n,m}$ and defined by the
box given in Fig.~\ref{fig:w1tro-box}, which takes place between a ``prover''
who controls the interfaces on the left-hand side of the box, and a ``verifier''
who controls the interfaces on the right.
%  This primitive is intended to capture
% the minimal functionality sufficient for removing interation through the Fiat-Shamir.
A protocol instantiating
$\wotro^{n,m}$ is secure if for any function $f(\cdot)$, the adversary can't
produce an output of the form $(a,f(a))$ on the verifier's interface.
We ask whether this primitive has a secure non-interactive instantiation in the \crqss{} model.
 \begin{figure}
    \begin{center}
    \begin{tikzpicture}[thick]
        \draw (0,0) node[porte, minimum width=1.5cm, minimum height=1.5cm] (wotro) {$\wotro^{n,m}_{\Gamma}$};
        \draw
            (wotro.west) ++(0, .5) coordinate (gauche1)
            (wotro.west) ++(0, -.5) coordinate (gauche2)
            (wotro.east |- gauche1) coordinate (droite1)
            (wotro.east |- gauche2) coordinate (droite2)
            ;

        \draw
            (gauche1) ++(-1, 0) node[left] {$a \in \bool^n$} edge[->] (gauche1)
            (gauche2) ++(-1, 0) node[left] {$c \in_R \bool^m$} edge[<-] (gauche2)
            (droite1) ++(1, 0) node[right] {$a$} edge[<-] (droite1)
            (droite2) ++(1, 0) node[right] {$c$} edge[<-] (droite2)
            ;
    \end{tikzpicture}
    \end{center}
    \caption{$\wotro^{n,m}$ as a box. The prover on the left puts a chosen $a \in \bool^n$ into it, the box chooses  $c \in_R \bool^m$, and outputs $(a,c)$ to the verifier on the right-hand side.}
    \label{fig:w1tro-box}
\end{figure}
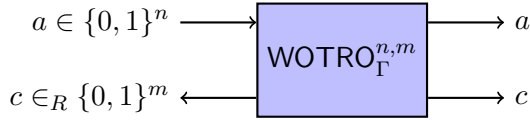
Our main contribution is showing that, despite the evidence to the contrary
presented above, this primitive has no statistically secure implementation in the \crqss{} model. 
Our impossibilities also apply for two-message protocols
 in the plain model (and even in the \crqss{} model) 
 since the \crqss{} could be prepared by the verifier.
 \begin{theorem}[informal version of Theorem~\ref{thm:attaque-simulable}]
   \label{thm:attaque-simulable-informel}
  If $n-m\in\omega(\lg n)$, there is no statistically secure non-interactive
  protocol for $\wotro^{n,m}$ in the \crqss{} model.
\end{theorem}
For any protocol in the \crqss{} model, we construct an (inefficient) attack that will make the verifier accept an output of the form $(a,f(a))$ for a function $f$ chosen at random.
Our attack is a novel use of the operator Chernoff bound of Ahlswede and Winter~\cite{ahlswede-winter}.

What about \wotro{} protocols that provide computational security by relying on a hardness assumption? We show that such a protocol could not be proven secure using reductions that treat the adversary as a black-box (i.e.\ a CPTP map). We call this a quantum fully-black-box (or \fbb{}) reduction.
\begin{theorem}[informal version of Corollary~\ref{cor:bb-imposs-wotro}]
  \label{thm:bb-imposs-wotro-informel}
  If $n-m\in\omega(\lg n)$, there is no protocol for $\wotro^{n,m}$ whose
  security can be established by a quantum \fbb{} reduction to a cryptographic
  game assumption, unless that assumption is false.
\end{theorem}
The statement above and its proof are similar to the impossibility results
of~\cite{bitansky_why_2013,bgw12} in the context of Fiat-Shamir in the CRS
model. We rely on a technique formalized by Wichs in \cite{wichs_barriers_2013}.
We show that the input/output behaviour of our attacker against any \wotro{}
protocol can be \emph{simulated} efficiently by a quantum
algorithm.
This means that no reduction can exist that breaks the security of a cryptographic game assumption 
using only the input/output behaviour of a successful adversary against \wotro{}, 
unless the assumption is false.
 Otherwise, the reduction together with the
simulator for the attack would yield an efficient algorithm for breaking the
assumption.

While \wotro{} implies Fiat-Shamir, the other direction does not hold. Still, we can use the attack from our impossibility of \wotro{} to obtain a similar result ruling out any universal instantiation of Fiat-Shamir in the \crqss{} model. 
\begin{theorem}[informal version of Theorem~\ref{thm:qfs-imposs}]
  For $n-m\in\omega(\lg n)$, there is no $\Sigma_{n,m}$--universal instantiation
  of the Fiat-Shamir transform whose security can be established by quantum \fbb{}
  reduction to a cryptographic game assumption, unless that assumption is false.
\end{theorem}
Interestingly, our impossibility is more general than the classical
one~\cite{bitansky_why_2013}, even when restricted to classical shared
resources. A \crqss{} can capture as a special case asymmetric setups such as
giving the verifier the trapdoor to some primitive the prover uses or
pre-computed randomized oblivious transfers. We obtain this generality ``for
free'' by considering the cryptographic primitive \wotro{} instead of a family
of cryptographic hash functions, as in~\cite{bitansky_why_2013}.

Studying the $\wotro$ primitive instead of Fiat-Shamir directly has another
advantage in that our impossibility result also applies to any
cryptographic task which (black-box) implies $\wotro$. For instance, we
introduce a strengthened variant of Zhandry's quantum
lightning~\cite{zhandry_quantum_2019} that implies \wotro{}. Quantum lightning
(QL) is a primitive that produces a quantum state and an associated
serial number such that no adversary can produce two states
with the same serial number (hence the name ``lightning''). A consequence of
this property is that serial numbers are highly unpredictable. A natural
question is whether some form of metadata can be embedded into quantum lightning
such that changing the value of this metadata requires creating a new lightning
state. This metadata could for example contain ownership information and it
would thus be impossible, even to the emitter of the state, to change the owner
of a state without generating an entirely new state. It could also serve to
encode a denomination for quantum bank notes, such that not even the emitting
bank could change the denomiation of an existing quantum note.

 We introduce a variant of quantum lightning that allows such metadata by
adding a classical input to the state generation procedure. We call this variant
\emph{typed quantum lightning} (\tql{}) which is secure if the serial numbers
remain unpredictable conditioned on the input. We show that this variant implies
\wotro{} and thus inherits the same black-box impossibility.
\begin{theorem}[informal version of Corollary~\ref{cor:lightning}]
  \label{cor:lightning-informel}
  There is no quantum \fbb{} reduction from the security of a \tql{}
  scheme to the security of a cryptographic game assumption when type length $n$
  and serial length $m$ satisfy $n-m\in\omega(\lg n)$, unless that assumption is false.
\end{theorem}

Why would \tql{} be a reasonable assumption? Clearly it is a very powerful primitive, but how much of a leap is it from ``vanilla'' quantum lightning? %We show that the difference appears relatively small.
While we do not have a definitive answer to that question, we can show that \ql{} implies \tql{} with small types. More precisely, we construct in Section~\ref{sec:tql-justif} a \tql{} scheme from regular \ql{} for types of $O(\lg(n))$ bits.

%  Neither min-entropy nor uniqueness provides \emph{independence} of the QL bolts: nothing prevents an adversary from constructing two bolts with serial numbers $s_1,s_2$ that satisfy $s_1\oplus s_2=1^n$ (or any constant string other than $0^n$) while still meeting the requirements of $s_1\neq s_2$ and $H_\infty(s_i)\in\omega(\lg n)$ $(i=1,2)$.

\paragraph{Instantiating \wotro{} from a non-game assumption.}

We show that it is possible to construct a \wotro{} protocol for which security
is based on a cryptographic assumption that does not fit the game formalism. Our
result is based on a new hardness assumption on cryptographic hash functions
called \emph{collision-shelters}. Intuitively, a family of hash functions is a
collision-shelter if no adversary can produce many collisions \emph{in
  superposition}. As such it is an intrinsically quantum definition which cannot
be framed as a game since no challenger can verify that an adversary breaks the
assumption. Using this assumption, we show how to construct a secure
$\wotro^{n,m}$ protocol in the \crqs{} model. We first prove the security of a
construction for $\wotro^{n,n}$ similar to the one based on EPR pairs and
mutually unbiased bases sketched earlier. The proof involves computing bounds on
the optimal probability of distinguishing between states from many mutually
unbiased bases and might be of independent interest. A $\wotro^{n,m}$ protocol
for $m<n$ is obtained by hashing the output with a collision-shelter hash
function.
\begin{theorem}[informal version of Theorem~\ref{thm:wotro-col-shelter}]
  Under the collision-shelter assumption, there are secure instantiations of
  $\wotro^{n,m}$ in the \crqs{} model.
\end{theorem}

\section{Technical Overview}\label{sec:tech-overview}

We call   \emph{Weak One-Time Random
  Oracle}, denoted %by 
  $\wotro^{n,m}$, the following simple non-interactive primitive. To any $a\in\{0,1\}^n$, it provides
  a {challenge} $c\in\{0,1\}^m$ 
  \emph{avoiding} with good probability 
  any function $\sff{c}:\{0,1\}^n\rightarrow \{0,1\}^m$.
  We say that an implementation of $\wotro^{n,m}$
  \emph{avoids} function $\sff{c}$ if no (efficient) dishonest prover
  is able to produce $(a,c)$ such that $c=\sff{c}(a)$.
  An implementation of $\wotro^{n,m}$ is said to be $\kappa$--secure
  if it behaves like a random oracle when the prover 
  is honest and avoids any function $\sff{c}$ 
   with probability at least $\kappa$, when the prover
  is dishonest. 
  It is easy to see that any non-interactive $\kappa$--secure implementation of $\wotro^{n,m}$
  can be used to implement the Fiat-Shamir transform with computational 
  soundness error upper-bounded by $1-\kappa$ (see Section~\ref{sec:wotro-impl-fs}). Any implementation
  of  $\wotro^{n,m}$ that avoids any function $\sff{c}(\cdot)$ would be a powerful 
  cryptographic primitive to remove interaction. An implementation 
  $\pwotro{n,m}=(\prover',\verifier')$ of $\wotro^{n,m}$ in the \crqss{} model is defined by two families
  of efficient POVMs $\prover'=\{\mathcal{P}^a\}_{a}$ and 
  $\verifier'=\{\mathcal{V}^{a,c,v}\}_{a,c,v}$ with $a\in\{0,1\}^n$, $c\in\{0,1\}^m$, and
  $v$ is an auxiliary string announced to $\verifier'$. $\pwotro{n,m}=(\prover',\verifier')$ is
  executed as folllows:
  \begin{enumerate}
  \item Upon input $a\in\{0,1\}^n$, $\prover'$ applies POVM 
  $\mathcal{P}^{a} := \{P^a_{c,v}\}_{c,v}$
  to register $P$ of the \crqss\
  to get classical outcome $(c,v)$. $\prover'$ then announces $(a,c,v)$ to $\verifier'$.
\item $\verifier'$  applies POVM 
  $\mathcal{V}^{a,c,v}:=\{V^{a,c,v}_0,V^{a,c,v}_1\}$
  to register $V$ of the \crqss{} and accepts iff  classical
  outcome $1$ is obtained.
\end{enumerate}
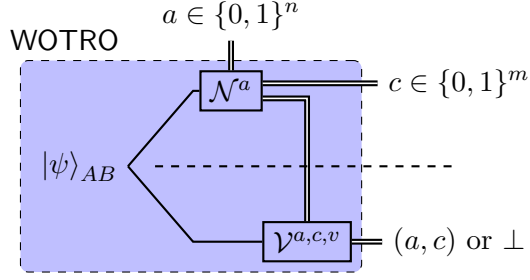
\begin{figure}
  \centering
  \begin{tikzpicture}[thick]
    \draw (0,1) node (epr) {$\ket{\psi}_{AB}$};
    \node[draw] at (2,2) (mp) {$\mathcal{N}^a$};
    \node[draw] at (3,0) (mv) {$\mathcal{V}^{a,c,v}$};

    \path[draw] (epr.east) -- (1.5,2) -- (mp) ;
    \draw[double] (mp.east)+(0,-0.1) -|  (mv);
    \path[draw] (epr.east) -- (1.5,0) -- (mv);
    \node (outv) at (5,0) {$(a,c)$ or $\bot$};
    \draw[double] (mv) -- (outv.west);
    \node (outp) at (5,2.1) {$c\in\bool^m$};
    \draw[double] (mp.east)+(0,0.1) -- (outp.west);

    \draw node[above of=mp] {$a\in\bool^n$} edge[double] (mp);
    % \draw[double] (mp) -- (mv);

    \draw[dashed] (1,1) -- (5,1);

    % pause here
    \begin{pgfonlayer}{background}
      \node[draw,dashed,rounded corners=3pt,fit=(epr)(mp)(mv),fill=blue!25,label=120:$\wotro$]  {};
    \end{pgfonlayer}
  \end{tikzpicture}
  \caption{A $\wotro^{n,m}$ protocol in the \crqss{} model as described in Definition~\ref{def:wotro}. The prover's actions are
    above the dashed line and the verifier's actions are below. The classical
    wire crossing the dashed line represents the classical message sent from the
  prover to the verifier. }
  \label{fig:wotro-crqs}
\end{figure}

An adversary $\adv$ 
  against  $\pwotro{n,m}$ takes no input and applies a POVM 
  $\adv:= \{\mathcal{A}_{a,c,v}\}_{a,c,v}$  to register $P$ of the \crqss{} to 
  obtain $a$ along with the message $(c,v)$. 
  Notice that as defined, $\pwotro{n,m}$
  requires the message transmitted to $\verifier'$ to be classical.
  This can be done without loss of generality as a protocol
  asking $\prover'$ to send a quantum message can be transformed
  into one where $\prover'$ only sends a classical message by 
  adding to the \crqss\ enough EPR pairs for the quantum message
  to be teleported. The security of the original protocol remains untouched
  by this transformation.  
  
\subsection{The impossibility of $\wotro^{n,m}$ in the \crqss{} model.}

For main contribution (Theorem~\ref{thm:attaque-simulable-informel}), we
construct an (inefficient) adversary $\adv^f$ that picks a random function
$f:\bool^n\rightarrow\bool^m$ such that the verifier will always accept the
outcome $(a,f(a))$ in the protocol. $\adv^f$ mounts its attack using the
prover's honest POVM operators $\mathcal{P}^a_{c,v}$ using the following
\emph{attack operators}: $X^f_a:= \sum_v \mathcal{P}^a_{f(a),v}$. The success of
this attack relies on two crucial facts. First, the quantum operator Chernoff
bound of~\citeauthor{ahlswede-winter} allows us to show that the
$\{X^f_a\}_{a\in\bool^n}$ (almost) form a POVM with overwhelming probability
over the choice of $f$. Second, since the attack uses the honest prover
operators, the verifier will accept with the same probability as with the honest
prover.

As for the impossibility of computationally secure \wotro{}
(Theorem~\ref{thm:bb-imposs-wotro-informel}), we use a proof strategy similar to
that of Bitansky {\em et al.} in \cite{bitansky_why_2013,bgw12} when proving
that there exists no black-box reduction from any successful adversary against
the entropy preserving property of a family of hash functions to a cryptographic
game. That is, we show that our adversary against $\wotro^{n,m}$ is
\emph{simulatable} by an efficient quantum circuit. The main difference here is
that our quantum circuit is stateless while it is stateful
in~\cite{bitansky_why_2013,bgw12}. This prevents the security of $\wotro^{n,m}$
to be established by a reduction to any cryptographic game that treats the
adversary as a black-box (according to Definition~\ref{def:quantum-fully-bb}),
as if there was such a reduction the game would also be won using the efficient
simulator (in other words, the game assumption would be false).
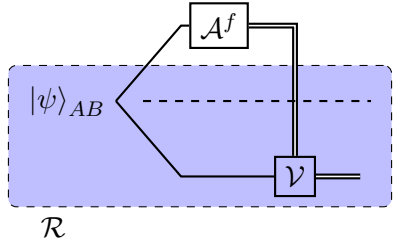
\begin{figure}
  \centering
  \begin{tikzpicture}[thick]
    \draw (0,1) node (epr) {$\ket{\psi}_{AB}$}; \node[draw] at (2,2)
    (mp) {$\adv^f$}; \node[draw] at (3,0) (mv) {$\mathcal{V}$};

    \path[draw] (epr.east) -- (1.5,2) -- (mp) ; \draw[double] (mp.east)
    -| % node[draw,fill=white,pos=0.65,scale=0.75] {$a,c,v$}
    (mv); \path[draw] (epr.east) -- (1.5,0) -- (mv); \node (outv) at
    (4,0) {}; \draw[double] (mv) -- (outv.west);
    % \draw[double] (mp.east)+(0,0.1) -- (outp.west);

    % \draw node[above of=mp] {$a$} edge[double] (mp); \draw[double]
    % (mp) -- (mv);

    \draw[dashed] (1,1) -- (4,1);

    % pause here
    \begin{pgfonlayer}{background}
      \node[draw,dashed,rounded
      corners=3pt,fit=(epr)(mv)(outv),fill=blue!25,label=210:$\mathcal{R}$]
      {};
    \end{pgfonlayer}
    % \node at (-0.2,0) {$\mathcal{R}$};
  \end{tikzpicture}
  \caption{The interface between the \fbb{} reduction $\mathcal{R}$ and the
    \wotro{} adversary $\adv^f$. The reduction simulates the \crqss{} $\ket
    \psi_{AB}$ and the verifier. It provides register $A$ to the adversary
    $\adv^f$ and receives its classical outcome. }
  \label{fig:reduction-interface}
\end{figure}

\subsection{Quantum Black-Box Reductions.}
\label{sec:qbb-discussion}

We should make precise what we mean by quantum black-box reductions. The
classical notion of black-box reductions is uncontroversial; the reduction is an
efficient algorithm $\mach{R}$ having black-box (i.e.\ input/output) access to
an adversary $\adv$ breaking a scheme. In the quantum setting, the reduction
itself can be quantum, i.e.\ be an oracle access quantum circuit
(Definition~\ref{def:oracle-machine}), or it can be classical; and the adversary
\adv{} can also be quantum or classical. Firstly, for any reasonable definition
of fully black-box, $\mach{R}$ should not be able to tell if $\adv$ is quantum
or classical. It has been argued that if $\mach{R}$ is quantum,
$\mach{R}$ can be purified as a unitary circuit, and since unitary circuits are
reversible, $\mach{R}$ should have access to $\adv$ and its inverse $\adv^*$ in
order to preserve this property. While this is sometimes called quantum
black-box, it is closer to the classical \emph{semi-black-box} notion, where the
underlying primitive is treated as a black-box, but where the reduction can
depend on the adversary~\cite{BBF13}. These quantum reductions have
been called \emph{quantum
semi-black-box reductions} in~\cite{cao_gap_2022}. 

%More importantly, reversibility is not unique to quantum
%computing: any classical computation can also be made reversible. So why would
%classical reductions not also have access to an oracle for the inverse of the
%adversary when we consider black-box reductions? 
%We see no reason why quantum reductions should be endowed with more
%power than their classical counterparts when it comes to the type of access
%to the adversary that are allowed by the reduction.

While this type of reduction is very useful (to allow rewinding, for instance), we argue that it cannot truly be considered ``black-box''. Most realistic models of quantum computation (including all current prototypes for quantum computers) include irreversible operations as part of their native gate sets (for control flow if nothing else). Given a quantum algorithm written in, say, openQASM, one would need access to the source code to get a unitary circuit that can be run backwards; it requires ``opening the box'' to the same extent as running a homomorphically-encrypted version of an adversary. 

Notice also that unlike (fully) black-box reductions, semi-black-box reductions
require the adversary to be efficient, otherwise the reduction implemented as a
quantum circuit with $\adv$ and $\adv^*$ gates, would potentially need to feed
exponentially-many auxiliary input wires to these unitaries. (Fully) Black-box
reductions should never be affected by how the adversary is implemented.
Basically, a blackbox reduction remains efficient even when the adversary is not
(when oracle calls to $\adv$ are at unit cost). This is a crucial property of
fully-blackbox reductions.

%  Allowing a reduction $\mach{R}$
% to call $\adv^*$ in addition to $\adv$ requires to purify the adversary, an operation
% applied to the adversary's source code.
%% Phil: la phrase ci-dessus répètais le paragraphe précédent

In view of the above, we adopt the following definition of quantum fully
black-box reductions (\fbb{}): \qpt{} algorithms $\mach{R}$ that have oracle
access to a CPTP map implementing the adversary $\adv$. This is the true quantum
analogue of the kind of black-box reduction in impossibility results such
as~\cite{bitansky_why_2013}. Note that by weakening the notion of black-box, for
example by giving reversible access to the adversary, the set of possible
reductions increases, so black-box impossibility or separation results become
harder to find. In particular, the result of~\cite{bitansky_why_2013} is not
known to, and probably does not, hold for this kind of reductions. Our results of
Section~\ref{sec:impo-fiat-shamir} therefore strictly
generalizes~\cite{bitansky_why_2013} since we consider the quantum variant of
classical \fbb{} reductions. 
We expect that extending our impossibility result to the semi-blackbox setting
would require completely new techniques.
On a positive note, the impossibility that we prove here
could be avoided when the reductions considered are not \fbb{}.
Notice, however, that most security proofs in post-quantum
cryptography proceed by \fbb{} reductions. 
Most relevant to what we are doing here is the
security of the Fiat-Shamir transform in the  QROM as shown in 
\cite{don_security_2019,liu_revisiting_2019}. In these two papers,
the soundness of the Fiat-Shamir transform is established by
\fbb{} reductions. Rewinding the adversary is only required to
extract the witness, and since this is only needed to show that 
the Fiat-Shamir transform is a proof of knowledge in the QROM,
this is irrelevant for the impossibility result we provide here.

\subsection{\siguni\ quantum Fiat-Shamir cannot be \fbb{} reduced to a game.} 
We then show that the black-box impossibility of \wotro{} 
(indirectly) implies that the soundness of any \siguninm\ quantum Fiat-Shamir
transform cannot be established  under the same conditions. 
As our basic impossibility result is about the security of a cryptographic primitive rather 
than a  property of a family of hash functions (as in \cite{bitansky_why_2013,bgw12}),
we follow a different path. First, let us discuss what distinguishes $\wotro^{n,m}$ 
from a \siguninm\ quantum Fiat-Shamir transform in the \crqss\ model.

Consider a Fiat-Shamir transform in the \crqss{} model. The general form of such
a protocol is similar to the generic \wotro{} protocol outlined above. I.e.\ the
pover performs a POVM specified by $a$, the first message in protocol $\Sigma$
and the verifier performs a binary outcome POVM on its part of the \crqss{} as
specified by the prover's message, along with additional checks according to
$\Sigma$ (see Section~\ref{sec:impo-fiat-shamir} for details). 

Although this is providing something very close to $\wotro^{n,m}$ in its inner
workings, it may not need to avoid all functions to be a computationally sound
\siguni\ implementation of the Fiat-Shamir transform. It only needs to avoid
functions $\mathsf{c}:\{0,1\}^n\rightarrow \{0,1\}^m$ such that for some \spr{}
for some language $L$, there exists $x\notin L$ for which upon commitment $a$,
only challenge $\mathsf{c}(a)$ has a third message $z$ such that
$(a,\mathsf{c}(a), z)$ is a valid transcript. We show that this relaxation on
the functions to be avoided by any \siguninm\ $\QFSs$ leads to the same
impossibility result than for $\wotro^{n,m}$.

 The proof follows from the existence of a \spr\
 $\Sigma^f=(\prover, \verifier^f)$ for membership  to the empty language, 
 where $f:\{0,1\}^n\rightarrow\{0,1\}^m$ is a random oracle. 
 Although $\Sigma^f$ only requires $\verifier^f$ to have access 
 to the oracle $f(\cdot)$ to run the protocol honestly, 
 the adversary $\adv^f$ has also access to $f(\cdot)$ to mount its attack
 against the soundness of  $\QFS{\Sigma^f}$. This is essentially the same adversary
 defined  as the one against $\wotro^{n,m}$ described above.  
  Notice that if the soundness of $\QFSs$ was \fbb{}
  reducible to game $\mathcal{G}$ then there would be an efficient
  algorithm $\mathcal{B}^f$, having oracle access to $f(\cdot)$, that wins game $\mathcal{G}$. 
  The strategy  used for $\wotro^{n,m}$ can then be applied.
  A possibly inefficient adversary $\adv^f$ is defined that almost all the time
  breaks the soundness of  $\QFS{\Sigma^f}$.
  We finally show that  both the adversary $\adv^f$ and
  $\verifier^f$   can be simulated by an efficient  stateful simulator. 
  As before, this prevents 
  the soundness of $\QFSs$ to be established by \fbb{} reduction
  to a cryptographic game unless the game is trivial.

\subsection{A quantum assumption allowing for $\wotro^{n,m}$.}

We introduce a strong variant of  collision resistant families of hash functions
allowing for a computationally sound  \siguni\ implementation of the Fiat-Shamir
transform  in the \crqs\ model. We call $\collshn:=\{G^n_s\}_{s}\subset 
\{0,1\}^n \times\{0,1\}^n \rightarrow \{0,1\}^m$ a \emph{collision-shelter}
if,  for %any $\delta>0$ and
any target function $\mathsf{c}:\{0,1\}^n \rightarrow\{0,1\}^m$,
no efficient quantum adversary 
can produce  
any state with inverse-polynomial trace distance to a state of the form
\[\ket{\psi}_{AX} = \sum_{a} \alpha_{a} \ket{a}_A \otimes
\sum_{x:G^n(a,x)=\mathsf{c}(a)}\beta^a_x \ket{x}_{X}\enspace,\]
that contains collisions to $\mathsf{c}(a)$ when $a$ is measured.

 In order to show that collision-shelters are sufficient
 for a sound \siguni\ Fiat-Shamir transform in the \crqs\ model, we 
start  with the weak random oracle implemented using $n$ shared EPR pairs
from the introduction.  We modify the scheme  slightly   
to get an unconditionnally $\frac{1}{4}$--secure\footnote{By $\frac{1}{4}$--secure,
we really mean $\left(\frac{1}{4}-\negl\right)$--secure.} 
implementation 
$\pwotro{n,n}=(\prover',\verifier')$ of $\wotro^{n,n}$
in the \crqs\ model. This forms the basis upon which $\wotro^{n,m}$, with $m<n$,
is constructed  using a  collision-shelter. We prove that 
$\pwotro{n,n}$ is $\frac{1}{4}$--secure using shared maximally entangled pairs
of qutrits  as the \crqs\  to allow the use of a  particular set $\{\theta_a \}_{a\in\{0,1,2\}^n}$
of mutually unbiased bases, introduced by Wootters and Fields~\cite{WF89}. The
set  $\{\theta_a \}_{a}$  is shown to prevent any adversary 
$\adv:= \{\mathcal{A}_{a,c,v}\}_{a,c,v}$ from observing
$\mathcal{A}_{a,\mathsf{c}(a),v}\otimes \mathcal{V}^{a,\mathsf{c}(a),v}_1$
with probability better than $\frac{3}{4}$
when the \crqs\ is measured by  $\prover'$ and $\verifier'$. This
result may be of independent interest and is made possible
as  $\adv$'s success probability is given by an instance of a Weil sum that can
be upper bounded by Deligne's resolution of one of  Weil's conjectures~\cite{d74}. 

A protocol  $\pwotro{n,m}[\collshn]=(\prover'', \verifier'')$ for 
$\wotro^{n,m}$ with $m<n$ can then be constructed using a collision-shelter 
$\collshn$ in the obvious way. Upon input $a\in\{0,1,2\}^n$,
$\prover''$ runs $\prover'$ upon input $a$ to get $(c',v)\in\{0,1\}^n\times\{0,1\}^n$. $\prover''$   announces
$(a,c',v)$ to $\verifier''$.  The challenge 
produced by  $\pwotro{n,m}$ is simply set to $c:=G^n_s(a,c')\in\{0,1\}^m$
for $s$ a  $\crs$. 
$\verifier''$ simply runs $\verifier'$ on $(a,c',v)$ and accepts if $\verifier'$
accepts. It is not difficult to see that if $\collshn$
is a collision-shelter then no efficient adversary $\adv$ can do better against
$\pwotro{n,m}$ than an unconditional adversary against $\pwotro{n,n}$. 
As a result, $\pwotro{n,m}$ 
avoids all functions with probability $\frac{1}{4}$. Negligible soundness error 
can then be achieved  by parallel repetitions.

\subsection{\wotro\ and Quantum Lightning.}

Quantum lightning (\ql), introduced by Zhandry~\cite{zhandry_quantum_2019}, is a
quantum cryptographic task allowing anyone to generate quantum states of which
they can make exactly one copy (called the \emph{uniqueness property}).
The original construction of Zhandry based on an ad hoc assumption was shown insecure by Roberts~\cite{Roberts21}.

Informally, a \ql{} scheme consists of a quantum
algorithm $\thunder$ instructing how to construct bolts $\ket{\lightn}$ and of a
verification algorithm $\verf$ that on input $\ket \lightn$ returns a serial
number $s\in\bool^n$ without disturbing state $\ket\lightn$ such that no
efficient adversary can create two valid states with the same serial number. For
this to hold, there must be uncertainty in the serial number of newly
created bolts: for every \qpt{} adversary \adv{}, $\ket\lightn \leftarrow
\adv(\thunder)$ must satisfy $H_\infty(\verf(\ket\lightn))\in \omega(\lg n)$,
otherwise polynomially many tries would give two bolts with the same serial
number, contradicting uniqueness. Note that an efficient reduction does not
necessarily exist in the other direction: an adversary could for example produce
two valid states with identical serial numbers that each have maximal
min-entropy. Such an adversary appears useless for producing a single lightning
state with low min-entropy in the serial number.

We introduce a variant of quantum lightning where the bolt generation procedure
accepts an input. \emph{Typed quantum lightning} (\tql) is a new primitive
similar to \ql{} where $\thunder$ takes an additional parameter (or type)
$a\in\bool^n$. Intuitively, security asks that when we fix the type $a$, the
resulting scheme still produces unpredictable serial numbers. This is formalized
by requiring that the conditional min-entropy $H_\infty(S\mid A)$ is large. We
show that a \tql{} scheme with type length $n$ and serial number length $m$
implies the existence of a protocol for $\wotro^{n,m}$. The scheme asks the
prover to generate a typed QL state with type $a$ and teleport that state to the
verifier using EPR pairs from a \crqs{}, the verifier accepts if the teleported
state is a valid \tql{} state. A consequence if this scheme is that no \tql{}
scheme satisfying $n-m\in\omega(\lg n)$ can have its security be \fbb{}
reducible to a cryptographic game assumption.

\section{Notations \& Preliminaries}
We use $n\in\naturals$ as the security parameter throughout the paper.
We use $\poly[n]$ to denote a polynomial in $n$.  
A function $f:\N\rightarrow\N$ is said to be \emph{negligible} if 
for all polynomials $p(\cdot)$ and for $n\in\N$ sufficiently large, $f(n)\leq 1/p(n)$.
We denote a negligible function by $\negli{\cdot}$. We use ``\qpt{}'' as a shorthand
for \emph{quantum polynomial time}. We use $\ln(\cdot)$ and $\lg(\cdot)$ to respectively denote the base $e$ and 2 logarithms. To denote a Hilbert space of dimension $2^n$, we write
$\hilbert_n$.

For a set $A$, its cardinality is denoted $|A|$ and its complement $\bar A$. We
write $x\in_R A$ to indicate that $x$ is chosen uniformly at
random from $A$.

We often use the notation $f(\cdot)$ to denote functions as a way to
differentiate them from variables. If $f(\cdot,\cdot)$ is a function of two
arguments, we denote by $f(x,\cdot)$ the function of one argument defined by
restricting the first argument to value $x$. For two sets $A$ and $B$, we denote
the set of functions from $A$ to $B$ as $A\rightarrow B$.
Let $\mathcal{F}^{n,m}$ be the set of functions $\{0,1\}^n\rightarrow\{0,1\}^m$.
We will often view $m$ as a function $m(n)$ of the security parameter $n$. 
To simplify the notation, when $n$ and $m$ are clear in the context,
we will write $\mathcal{F} := \mathcal{F}^{n,m}$.

For a random variable $X$, $\expect{X}$ denotes its expected value
and for $X(r)$ a random variable function of $r$, $\expsub{r}{X}$ denotes
its expected value when $r$ is picked at random.
 Let $\Delta(A,B)= \frac 12\sum_a
  |\Pr[A=a]-\Pr[B=a]|$ denote the statistical distance between the distribution
  of two random variables $A$ and $B$ with the same domain.
For an operator $A\in\complex^{n\times n}$,  $\|A \|_1=\trace{\sqrt{A^*A}}$ denotes its trace norm.%  If 

\subsection{\sprs{} and the Fiat-Shamir Transform}
\label{sec:sprs-fiat-shamir}

Let $R\subseteq \bool^*\times\bool^*$ denote an arbitrary efficiently computable binary
relation such that if $(x,w)\in R$ then $|w|\leq p(|x|)$
for some polynomial $p(\cdot)$. We call $x$ a \emph{public instance} and $\pi$ a
\emph{witness} for $x$. The condition above ensures that the witness of any
public instance can be conveyed efficiently. From the binary relation $R$, we
define the language $L_R =\{x \,|\,(\exists \pi)[(x,\pi)\in R] \}\in \mathbf{NP}$ of
public instances with witnesses for them.
\begin{defi}[\spr~\cite{IvanSig10}]\label{def:sigma}
  A \spr\ $\Sigma=(\prover,\verifier)$ for a binary relation $R$ is a 3-message
  protocol with \emph{conversation alphabet} $\bool$. On public input $x\in
  L_R$ and on private input $\pi$ to $\prover$ such that $(x,\pi)\in R$, the
  protocol structure is as follows:
  \begin{itemize}
  \item The prover sends a message $a=\prover_1(x,\pi)\in\bool^n$ called the
    \emph{commitment}.
  \item The verifier sends a \emph{challenge} $c\in\bool^m$.
  \item The prover sends a \emph{reply} $z=\prover_2(a, x, \pi,c)\in\bool^*$, and the
    verifier outputs $\verifier(x,a,c,z)\in\{{\tt accept},{\tt reject}\}$.
  \end{itemize}
  Moreover, the protocol satisfies the following requirements:
  \begin{description}
  \item[\emph{Random public coins:}] The challenge $c\in \bool^m$ is 
    chosen uniformly at random in $\bool^m$ without any extra processing (i.e.
    no need for private information to generate $c$).
  \item[\emph{Perfect correctness:}] When $x\in L_{R}$, \verifier\ accepts \prover\ with probability $1$.
  \item[\emph{Special soundness:}] When $x\in L_R$, given two accepting conversations for the
    same commitment $(a,c,z)$ and $(a,c',z')$ with  $c\neq c'$, 
    there exists a \ppt\
    algorithm $\sff{E}$ such that $(x, \sff{E}(a,c,z,c',z'))\in R$.
  \end{description}
\end{defi}
We should mention here that \sprs\ are also often used as a synonym
of 3-message public-coins protocols (as in~\cite{KRR17, peikert_noninteractive_2019}, for instance)  
irrespectively of whether the proof system satisfies perfect correctness or
special soundness. However, since we are proving a negative result, there is no loss in generality in adopting the more restrictive definition of~\cite{IvanSig10}. 

By special soundness, if $x\notin L_R$ then for any commitment $a\in
\bool^n$, there is at most one challenge $c\in \bool^m$ such that for
some response $z$, $(a,c,z)$ is an accepting conversation. For some $\spr$ $\Sigma_L$ for a language $L$ and some $x\notin L$, we call the function that maps $a$ to this one challenge $c$ the \emph{bad challenge} function.

In the ROM,
the Fiat-Shamir transform $\FSH{\Sigma}{}=(\prover^{\fsl},\verifier^{\fsl})$ applied 
to  a \spr{} $\Sigma=(\prover,\verifier)$ with first message length $n$ and
challenge length $m$ for a proof of language membership is a 
non-interactive argument where, on public input $x\in L$ and random oracle $H:\bool^n\rightarrow \bool^m$,
\begin{enumerate}
\item $\prover^{\fsl}$ runs $a=\prover(x,w)$ computes $c=H(a)$
  and  $z=\prover_2(a,x,w,c)$, and sends  $(a,c,z)$ to $\verifier^{\fsl}$.
\item $\verifier^\fsl$ rejects if $c\neq H(a)$, otherwise outputs $\verifier(x,a,c,z)$.
\end{enumerate}
In the CRS model, the protocol is the same with the random oracle replaced with a family of  cryptographic hash functions
$\mathcal{H}=\{h_r\}_r$ where $h_r:
\bool^n\rightarrow\bool^m$ is sampled using a CRS.%  On CRS $r$:

\subsection{Black-Box Impossibility Results}\label{bbi}
The following define what is meant by a cryptographic game assumption.

\begin{defi}[\cite{HH09,bitansky_why_2013}]\label{game}
  A \emph{cryptographic game} is a tuple ${\cal G}=(\Gamma,c)$ composed of an interactive Turing machine $\Gamma$ and a constant $c\in[0,1]$. On security parameter $n\in\naturals$, the challenger $\Gamma(1^n)$ interacts with an adversary ${\cal A}_n$ and outputs a bit $b$. The output of this interaction is denoted by $\langle {\cal A}_n \rightleftharpoons \Gamma(1^n)\rangle$. The \emph{advantage} of the family of adversaries $\adv=\{\adv_n\}_{n\in\naturals}$ in game ${\cal G}$ is defined as
  \begin{equation*}
    \advantage{}{\adv,{\cal G}}[(n)] = \Pr[\langle {\cal A}_n \rightleftharpoons \Gamma(1^n)\rangle=1]-c \enspace.
  \end{equation*}
  A cryptographic game ${\cal G}$ is secure if for all \ppt{} adversary $\adv$, the advantage $ \advantage{}{\adv,{\cal G}}[(n)]$ is $\negl[n]$. The communication can be classical or 
  quantum.
\end{defi}

Intuitively, a protocol $\Pi$ for
 \wotro{} has its security reducible to a cryptographic game assumption $\cal G$ if there  
 exists an efficient way to transform any successful adversary $\adv$ against $\Pi$ 
into a challenger winning game $\cal G$. If this transformation works 
only provided the standard input-output behaviour of $\adv$ then we say that the security
of $\Pi$ is \fbb{} reducible to game $\cal G$. 
Quantum black-box reductions are
defined formally in Section~\ref{sect:bbox}.
%In this section, we formalize the notion of black-box reduction in the quantum domain.

In this paper, we show the impossibility of black-box reducing the security of 
cryptographic primitive, called \wotro, to any cryptographic game.
Our proof uses the general technique of simulatable attacks formalized by Wichs~\cite{wichs_barriers_2013} and applied by~\cite{bitansky_why_2013} to the Fiat-Shamir transform. An inefficient adversary $\adv$ against some primitive is \emph{simulatable} if there exists a simulator $\simulator$ such that no efficient algorithm can distinguish between $\adv$ and $\simulator$ from black-box query access. A cryptographic task having a simulatable attack cannot be black-box reduced to a secure cryptographic game since the reduction $\macho{R}{(\cdot)}$ cannot distinguish between the inefficient $\adv$ and the efficient $\simulator$, which means that $\macho{R}{\simulator}$ would yield an efficient algorithm for the game ${\cal G}$ with non-negligible advantage, 
contradicting the assumption.

\section{A Simple Non-Interactive Primitive}\label{construction} 
In this paper, we consider a simple non-interactive cryptographic primitive,
called a \emph{weak one-time random oracle} ($\wotro^{n,m}$) and illustrated in
Fig.~\ref{fig:w1tro-box} where the prover inputs $a\in\bool^n$ into the box and
gets $c\in\bool^m$ as output while the verifier inputs nothing and gets $(a,c)$
as output. An implementation of this primitive is a protocol taking place
between the prover and the verifier. The verifier $\verifier$ is a machine that
takes no input, interacts with the prover in the way prescribed by the protocol,
and either accepts and outputs $(a,c)$ or rejects and outputs $\bot$. In an
honest implementation, the prover is a machine $\prover$ taking as input an $a
\in \bool^n$ and interacts with the verifier as specified by the protocol, in
such a way that the verifier accepts and outputs the same $c$. The strings $a$
and $c$ can then be determined from the transcript of the protocol. We can then
view the whole protocol in the honest case as a conditional distribution
$\Pi(c|a)$ that tells us the probability of getting the challenge $c$ given that
the prover was given $a$ as input. 

In a dishonest implementation, the prover $\dprover$ takes no input at all (it is free to choose $a$) and might behave in a way that will cause the verifier to reject. The protocol is then simply a joint probability distribution $\tilde{\Pi}_{\dprover}(a,c,v)$, representing the distribution one obtains when $\dprover$ runs the protocol with the honest verifier $\verifier$, and where $v \in \bin$ is 1 when the verifier accepts and 0 if he rejects.

We now define correctness and security of an implementation. In a correct implementation of this primitive, $\Pi(c|a)$ will reflect exactly the same distribution over $a$ and $c$ given by the ideal box, namely $c$ will be uniformly distributed and independent of $a$, and the verifier always accepts when the prover is honest:

\begin{defi}[$\epsilon$-correctness]\label{def:correct-impl}
  A protocol $\Pi$ is a \emph{$\epsilon$-correct} implementation of
  $\wotro^{n,m}_{\Gamma}$ if for all $a \in \bool^n$ the conditional
  distribution $\Pi(c|a)$ is $(1-\epsilon)$-close (in statistical distance) to the
  uniform distribution over $c$ and if $\verifier$ 
  accepts with probability at least $1-\negl[n]$ when the
  prover is honest. $\Pi$ is said to be \emph{statistically correct}
  if it is $(1-\negl[n])$--correct.
  \end{defi}

As for our security definition, it will be rather weak (hence the ``weak'' in
the name of the primitive): we will only require that in a secure
implementation, a dishonest prover $\dprover$ cannot steer the choice of $c$
towards a deterministic function of $a$. Rather than require that $c$ be almost
uniform and independent, we will only demand that there be \emph{some}
randomness left in this choice. 

\begin{defi}[$\delta$--avoiding]\label{def:d-avoiding}
    For $0\leq \delta \leq 1$, we say that a tuple of random variables $(A,C,V)$ taking values in $\bool^n \times \bool^m \times \bin$ \emph{$\delta$--avoids} the function $\sff{c}: \bool^n \rightarrow \bool^m$ if 
 \[ \Pr\left[ V = 1 \wedge C = \sff{c}(A) \right] \leq 1- \delta \enspace.  \]
\end{defi}

This then leads to the following definition of security for an implementation of
\wotro{}.

\begin{defi}[$\delta$--security]\label{def:secure-impl}
  A protocol is a statistically (resp. computationally) $\delta$--secure
  implementation of $\wotro^{n,m}_{\Gamma}$ if for all dishonest provers (resp.
  all \qpt{} dishonest provers) $\dprover$, the random variable tuple $(A,C,V)$
  with joint distribution $\tilde{\Pi}_{\dprover}(a,c,v)$ $\delta$--avoids all
  functions $\sff{c}: \bool^n \rightarrow
  \bool^m$. We say that a protocol for \wotro{} is \emph{statistically (resp. computationally)
secure} if it is statistically (resp. computationally) $(1-\negl)$--secure.
\end{defi}

\paragraph{Basic Facts About \wotro.}
Observe that there is a trivial perfectly secure 2--message protocol for \wotro{} where \prover{} sends $a$ and \verifier{} sends a uniformly random $c$. Therefore, we will focus on non-interactive  (or 1--message) implementations of \wotro. 
A secure non-interactive \wotro{} protocol provides enough conditional randomness for sound instantiation of the Fiat-Shamir transform when applied to public-coin special-sound 3--message interactive proofs (\sprs{}). 

In the plain model, there is no secure \wotro\ protocol as the honest prover program defines the output $c$ as a function of $a$ that can never be avoided. In the CR\$ model, there exists a simple statistically $(1-\negli{m-n})$--secure one-message protocol when $m>n$, a statistically $\frac 1e$--secure protocol when $m=n$ and there is no protocol for $m<n$ whose computational security can be black-box reduced to a cryptographic game assumption, as a consequence of~\cite{bitansky_why_2013}. A detailed examination of these facts is provided in Appendix~\ref{sec:basicfactsproofs}.

\paragraph{\wotro{} in the \crqss{} Model.}

Since the object of study is the (im)possibility of the \wotro{} primitive in
the \crqss{} model, we present a general form for a 1--message \wotro{} protocol
in this model. 

\begin{defi}[\wotro{} in the \crqss{} model]\label{def:wotro}
A $\wotro{}^{n,m}$ protocol $\Pi=(\prover,\verifier)$ in the \crqss{} model consists of 
  \begin{itemize}
  \item A \crqss{} $\Psi_{PV}\in \density{\hilbert_{PV}}$ 
  \item A mapping of $a\in\bool^{n}$ to an efficient POVM $\mathcal{N}^a=\{N^a_{y,w}\}_{(y,w)\in\bool^{m\times \ell}}$ on register $P$.
  \item A mapping of $a\in\bool^{n}$, $y\in\bool^m$ and $w\in\bool^\ell$ to an efficient POVM $\mathcal{V}^{a,y,w}=\{V^{a,y,w}_0,V^{a,y,w}_1\}$ on register $V$.
  \end{itemize}
    On input $a\in\bool^n$:
  \begin{enumerate}
  \item $\prover$ applies POVM $\mathcal{N}^a$ on
    register $P$ of $\Psi_{PV}$ to obtain $y$ and an auxiliairy
    verification string $w$ and sends $(a,y,w)$ to the verifier.
  \item $\verifier$ applies POVM $\mathcal{V}^{a,y,w}$ on register $V$ of
    $\Psi_{PV}$, accepts and outputs $(a,y)$ if the result is $1$, and
    rejects and outputs nothing if the result is $0$.
  \end{enumerate}
  The conditional distribution of $y$ given $a$ is
  $\Pi(y|a)=\sum_w\trace{N^a_{y,w}\otimes V_1^{a,y,w}\Psi_{PV}}$ for $y\neq
  \bot$ and $\Pi(\bot|a)=\sum_{y,w}\trace{N^a_{y,w}\otimes
    V_0^{a,y,w}\Psi_{PV}}$. We say that $\Pi$ is \emph{$\delta$--correct} if
  $\frac 1{2^n}\sum_{a\in\bool^n}\sum_{y\neq \bot}\Pi(y|a)\geq \delta$.
\end{defi}
Note that requiring the auxiliary verification string $w$ to be classical is not
a restriction since the \crqss{} can contain EPR pairs for the teleportation of
an arbitrary quantum state from the prover to the verifier.

\subsection{\wotro\ to Implement the Fiat-Shamir Heuristic}
\label{sec:wotro-impl-fs}

Let $R_L$ be a relation for a language $L$ and let $\Sigma_L=(\prover_L,
\verifier_L)$ be a \spr{} for $R_L$ with commitments in $\Gamma^n$ and
challenges in $\Gamma^m$. Consider a secure implementation
$\Pi_\wotro=(\prover_\wotro,\verifier_\wotro)$ of $\wotro^{n,m}_{\Gamma}$. We
construct a non-interactive zero-knowledge proof (argument) system for $L$ by
applying the Fiat-Shamir transform to $\Sigma_L$ using the protocol $\Pi_\wotro$
as the instantiation of the hash function.

\begin{blockquote}
  \rule{\linewidth}{1pt}
  {\bf Protocol} $\Pi_\wotro[\Sigma_L]$ \\
  {\bf Setup: } A $\spr$ $\Sigma_L=(\prover_L,\verifier_L)$ where
  $\prover_L=(\prover_L^1,\prover_L^2)$ with commitments of size $n$ and challenges
  of size $m$ and a protocol $\Pi_\wotro=(\prover_\wotro, \verifier_\wotro)$ for
  $\wotro^{n,m}_\Gamma$.

  \vspace{1em}
  
  {\bf Prover message:} on public input $x\in L$ and private input $w$
  \begin{enumerate}
  \item compute $a\leftarrow \prover_L^1(x,w)$,
  \item compute $c\leftarrow \prover_\wotro(a)$,
  \item compute $z=\prover_L^2(a,x,w,c)$ and
  \item send $z$ to the verifier.
  \end{enumerate}

  {\bf Verification:} on public input $x\in L$ and upon reception of $z$,
  \begin{enumerate}
  \item compute $(a,c)\leftarrow \verifier_\wotro()$ 
  \item if $\verifier_\wotro$ rejected, output ${\tt reject}$ else
    output $\verifier_L(x,a,c,z)$.
  \end{enumerate}
  \rule{\linewidth}{1pt}
\end{blockquote}

\begin{theorem}\label{thm:wotro-impl-fs}
  If $\Sigma_L$ is a $\spr$ for language $L$ and if $\Pi_\wotro$ is a
  statistically (resp. computationally) $(1-\delta)$--secure and correct
  implementation of $\wotro$, then $\Pi_\wotro[\Sigma_L]$ is a statistically
  (resp. computationally) sound (with soundness error $\delta$) and perfectly
  correct non-interactive proof system for language membership in $L$.
\end{theorem}
\begin{proof}
  We first show correctness. By the correctness of $\Pi_\wotro$, it holds that
  the challenge $c\in\Gamma^m$ produced by $\Pi_\wotro$ is uniformly
  distributed. When both parties are honest, the probability that $\verifier_L$
  accepts when $c$ is taken as the output of $\Pi_\wotro$ in protocol
  $\Pi_\wotro[\Sigma_L]$ is the same as the probability that $V_L$ accepts in an
  execution of $\Sigma_L$. Since \sprs{} are perfectly correct by definition,
  this probability is one.

  Now for soundness, again by the definition of $\sprs$, protocol $\Sigma_L$
  satisfies special soundness. That is, for $x\notin L$, for any commitment
  $a\in\Gamma^n$, there exist at most one challenge $c\in\Gamma^m$ that leads to
  an accepting conversation. Let $\sff{c}:\Gamma^n\rightarrow\Gamma^m$ be the
  function that maps commitment $a$ to this unique challenge $c$ that makes
  $\verifier_L$ accept. If $\Pi_\wotro$ is a statistically $(1-\delta)$--secure implementation
  of $\wotro^{n,m}$, then the output of $\Pi_\wotro$ $(1-\delta)$--avoids any
  function for any dishonest $\tilde \prover_\wotro$. The probability that
  $\verifier$ for protocol $\Pi_\wotro[\Sigma_L]$ accepts when $x\notin L$ is
  equal to the probability that $\verifier_\wotro$ accepts output $(A,C)$
  \emph{and} that $\verifier_L$ accepts on input $(x,A,C,Z)$ for some $Z$. By
  special soundness, this probability is at most the probability that $\tilde
  \prover_\wotro$ can make $\verifier_\wotro$ accept the output
  $(A,\sff{c}(A))$. By the statistical $(1-\delta)$--security of $\Pi_\wotro$, this
  probability is at most $\delta$.

  The reasoning for computational soundness is the same, but where we instead
  restrict to \qpt{} adversarial provers $\tilde \prover_\wotro$ against
  $\Pi_\wotro$.\qed
\end{proof}

\subsection{\wotro{} from Non-Local Correlations}
\label{sec:wotro-from-non}

A non-local box (PR box) is a hypothetical device distributed between two
parties such that party $A$ inputs $x\in\bool$ into the device and gets an output $u\in\bool$ and party
$B$ inputs $y\in\bool$ and gets $v\in\bool$. The input/output behaviour of
the PR box is described by
\begin{align}\label{eq:1}
  \Pr[u,v\mid x,y] =
  \left\{
  \begin{array}{ll}
    \frac 12 &\text{ if } u\oplus v=x\wedge y\\
    0 &\text{ otherwise. }
  \end{array}
  \right.
\end{align}

Let ${\cal C}:\bool^n\rightarrow \bool^{N}$ be an error correcting code with
minimum distance $\epsilon n$ (for any distinct $x,x'\in\bool^n$, the Hamming distance
between ${\cal C}(x)$ and ${\cal C}(x')$ is at least $\epsilon n$). Let
$\{h_r:\bool^N\rightarrow\bool^m\}_{r\in {\cal R}}$ be a universal$_2$ family of
hash functions. The $\wotro^{n,m}$ protocol is as follows:
\begin{enumerate}
\item On CR\$ $r$, and using $N$ PR boxes,
\item Prover: on input $a\in\bool^n$, compute codeword $x:={\cal C}(a)$ and
  input $x$ into its interface of the $N$ PR boxes. Let $u\in\bool^N$ be the
  result. Send $(a,x,u)$ to the verifier and use $(a,h_r(u))$ as output.
\item Verifier: On reception of $a,x$, check that $x={\cal C}(a)$. Pick
  $y\in\bool^N$ uniformly at random and input $y$ into its interface of the $N$
  PR boxes. Let $v\in\bool^N$ be the result. Check that $u\oplus v=x\wedge y$.
  If any of the checks failed, output $\bot$, otherwise output $(a,h_r((x\wedge y) \oplus v))$.
\end{enumerate} 
\begin{theorem}
  The above protocol avoids every function $\sff{c}:\bool^n\rightarrow\bool^m$.
\end{theorem}
\begin{proof}
  We begin by describing the most general strategy for an adversary $\adv$
  against the protocol. $\adv$ can input arbitrary values in the PR boxes in any
  order an such that input bits can depend on the CR\$ $r$ and on the boxes'
  outputs to previous inputs. Let $\hat x\in\bool^N$ and $\hat u\in\bool^N$
  denote the input and output bits to the $N$ PR boxes, respectively. \adv{} is
  then free to choose $a$, $x$ and $u$ adaptively based on $\hat x$ and $\hat u$
  and send $(a,x,u)$ to the verifier. Since the verifier checks that $x={\cal
    C}(a)$ and aborts otherwise, we can assume that $x$ is indeed the codeword
  that corresponds to $a$.

  We show that \adv{} has little freedom in the choice of $a$ due to the
  error-correcting code and input/output behaviour of the PR boxes. Since $\cal
  C$ has minimal distance $\epsilon n$, there is at most one codeword $x_0$ such
  that $d(x_0,\hat x)\leq \frac\epsilon2 n$. Let $a_0={\cal C}^{-1}(x_0)$. If
  \adv{} tries to send $(a,x={\cal C}(a),u)$ for any $a\neq a_0$, then the
  verifier will abort with overwhelming probability as the following argument
  shows. Let $(y,v)$ denote the input/output pair of the verifier. Then,
  \begin{align*}
    &\Pr[x\wedge y=u\oplus v]\\
    &= \Pr[ x\wedge y = u\oplus(\hat u\oplus \hat x\wedge y) ]\\
    &= \Pr[ x\wedge y\oplus u = \hat x \wedge y \oplus \hat u]\\
    &= \prod_i \Pr[ x_i\wedge y_i\oplus u_i = \hat x_i \wedge y_i \oplus \hat u_i]\enspace.
  \end{align*}
  Now, consider the set of positions where $\hat x$ and $x$ differ: ${\cal
    S}=\{i: \hat x_i\neq x_i\}$. For any $i\in{\cal S}$,
  \begin{itemize}
  \item When $y_i=0$, the expression becomes $u_i=\hat
    u_i$.
  \item When $y_i=1$, the expression becomes $x_i\oplus u_i=\hat x_i\oplus \hat
    u_i$ and it is satisfied when $u_i\neq \hat u_i$.
  \end{itemize}
  Since $y$ is chosen independently and uniformly at random by the verifier, for
  every $i\in {\cal S}$, the expression $x_i\wedge y_i\oplus u_i = \hat x_i
  \wedge y_i \oplus \hat u_i$ has probability $\frac 12$ of not being satisfied.
  Therefore since $|{\cal S}|\geq \frac \epsilon 2n$ whenever $x\neq x_0$, the verifier rejects with
  probability at least $2^{-\frac\epsilon 2 n}$.

  Finally, since \adv{} is obligated to send $a_0$ and $x_0$ as described
  above and $u$ that satisfies $u\oplus v=x_0\wedge y$ as argued above, the
  output of the verifier satisfies
  \begin{align*}
    \Pr[c=\sff{c}(a_0)]&=\Pr[h_r(u)=\sff{c}(a_0)]=\Pr[h_r(v\oplus x_0\wedge y)=\sff{c}(a_0)]\\
                       &=\Pr[v\oplus x_0\wedge y\in h_r^{-1}(\sff{c}(a_0))]= \frac{|h_r^{-1}(\sff{c}(a_0))|}{2^{-N}}
  \end{align*}
  since $v\oplus x_0\wedge y$ is uniformly distributed.
  On average over the choice of $h_r$, the above expression equals $2^{-m}$
  because the universal$_2$ condition implies $\mathbb{E}_r
  |h_r^{-1}(z)|=2^{N-m}$ for any $z\in\bool^m$.
  \qed
\end{proof}

It is well known that quantum mechanics can approximate the correlations of
``noisy'' PR boxes with success probability of around $85\%$ whereas the best
classical strategies can only achieve $75\%$ success probability. We show in
Section~\ref{sec:nouveaux-trucs} that \wotro{} is impossible in the \crqss{}
model, i.e.\ in a model where PR boxes can be approximated with probability
$85\%$. A natural question is: what is the level of noise at which \wotro{} is
no longer possible? 

We must point out to the reader that our results do not imply that \wotro{} is
impossible using $85\%$ PR boxes. There is a fundamental difference between
(noisy) PR boxes -- which capture classical correlations -- and entangled
states: the latter can be measured coherently with a collapse happening on the
other end. Our impossibility crucially relies on the adversary's ability to
perform a coherent measurement on its register. In other words, if there is a
secure \wotro{} protocol using $85\%$ PR boxes, then this protocol is no longer
secure when the boxes are instantiated using EPR pairs. The question above thus
remains open.

\section{Impossibility of \wotro{} in the \crqss{} Model}
\label{sec:nouveaux-trucs}

In this section, we prove our main result: there exists no protocol for $\wotro$
in the \crqss\ model with statistical security or  
with computational security established by black-box reduction to a cryptographic game, even
a non falsifiable one. Our black-box impossibility result is proven using a similar technique as \cite{bitansky_why_2013,bgw12}. In Section~\ref{grosseattaque}, 
we  define an inefficient adversary that  breaks completely any protocol
implementing $\wotro^{n,m}$ with $n-m\in\omega(n)$ 
in the \crqss\ model. We call this adversary the Chernoff adversary
or the Chernoff attack.
In Section~\ref{sect:bbox}, we define what we mean by the
security of $\wotro$ to be established by quantum black-box reduction
to crypto game.
We generalize to the quantum case this standard way
of proving the security of  cryptographic protocols.   
In Section~\ref{sec:simul-above-attack}, we show how to 
efficiently simulate the attack described in Section~\ref{grosseattaque}.
We then conclude that the security of any protocol for $\wotro^{n,m}$
in the \crqss\ model
cannot be established by a quantum black reduction to crypto
games. As a consequence, feeding the 
reduction with the simulator of the Chernoff adversary instead
of the Chernoff adversary will win the game while
running efficiently. It follows that the game is trivial
if such a reduction existed.

\subsection{The Chernoff Attack Against Any Implementation of \wotro}\label{grosseattaque}
  In this section, we show that there exist (inefficient) attacks against any
 \mbox{1--message} $\wotro^{n,m}$ protocol in the \crqss{} model for $m$ (sufficiently) smaller than $n$. 
  The following definition describes a general strategy for an attack against $\wotro$.
  \begin{defi}\label{def:attack}
    An attack $\adv^f_n$ against a $\wotro{}^{n,m}$ protocol
    (Definition~\ref{def:wotro}) is characterized by a  \emph{target function}
    $f:\bool^n\rightarrow \bool^m$ and a (possibly inefficient) POVM
    $\{P^f_{a,y,w}\}_{(a,y,w)\in\bool^{n\times m\times \ell}}$. The adversary performs this POVM on register $P$ of \crqss{} $\Psi_{PV}$ and sends the result $(A,Y,W)$ to the verifier.
We say that this attack \emph{hits function} $f$ except with probability $\epsilon(\adv^f_n)$ if
    \begin{equation*}
      1-\epsilon(\adv^f_n) = \Pr[Y=f(A) \wedge \text{ \verifier\ \textsc{accepts}}]
      = \sum_{a,w}\trace{(P^f_{a,y,w}\otimes V_1^{a,f(a),w}) \Psi_{PV}}\enspace .
    \end{equation*}
  \end{defi}

  We construct an attack whose success is based on the Chernoff bound for
  operators proven by Ahlswede and Winter in~\cite{ahlswede-winter} and stated below. 
  For  operators $A$ and $B$ and $0\leq\eta\leq 1$, the notation $A\in[(1-\eta)B; (1+\eta)B]$ means that $A\geq (1-\eta)B$ and  $A\leq (1+\eta)B$. 
  \begin{lemma}[Operator Chernoff bound]\label{chernono}
    Let $X_1,\ldots,X_M$ be i.i.d.~random variables taking values in the operators
    ${\cal D}({\cal H})$ % change pour \cal D car je ne sais pas ce que represente \cal B
 on the $D$--dimensional Hilbert space ${\cal H}$ such that
    $0\leq X_j\leq \id$, with $A=\mathbb{E} [X_j]\geq\alpha \id$, and let $0<\eta \leq 1/2$. Then
    \begin{equation}
      \Pr \left[ \frac{1}{M}\sum_{j=1}^M X_j \not\in [(1-\eta)A;(1+\eta)A] \right] \leq 2D \exp\left( -M\frac{\alpha\eta^2}{2\ln 2} \right).
    \end{equation}
  \end{lemma}

Our general attack strategy picks a random $f$ and crafts a measurement on its
part of the \crqss\ such that the measurement outcome $(A,Y,W)$ satisfies: 
\begin{enumerate}
\item $Y=f(A)$ and 
\item $\verifier(A,Y,W)$ accepts with approximately the same
probability as in an honest execution of the protocol. 
\end{enumerate}
Such a measurement is not
efficiently implementable in general. We call this attack the
``Chernoff adversary'' based on the crucial use of the Chernoff bound in
building this measurement. 

  \begin{theorem}[Chernoff adversaries]\label{thm:attaque-simulable}
    Let $n,m\in\naturals$ such that $m<n$. Let $\Pi_\wotro^{n,m}$ be a $\delta$--correct $\wotro{}^{n,m}$ protocol described by \crqss{} $\Psi_{PV}$ and POVM family $\mathcal{N}^a=\{N^a_{y,w}\}_{(y,w)\in\bool^{m\times \ell}}$ and $\mathcal{V}^{a,y,w}=\{V^{a,y,w}_0,V^{a,y,w}_1\}$ for $a\in\bool^n$, $y\in\bool^m$ and $w\in\bool^\ell$. 
Let $\mathcal{F}=\{f:\{0,1\}^n\rightarrow \{0,1\}^m\}$ be the set of boolean functions
from $n$-bit strings to $m$-bit strings.
Let $2^k=\dim P$ be the dimension of the prover's register of $\Psi_{PV}$.
\begin{enumerate}
\item Let $\eta=\sqrt{2\ln 2(n+k)\frac {2^m}{2^n}}$. Then, there exists
$\mathcal{F}^*\subseteq \mathcal{F}$ such that $\Pr_{f\in_R\mathcal{F}}{[f\in \mathcal{F}^*]}\geq 1-\negl[n]$  and for each $f\in\mathcal{F}^*$,
    \begin{equation*}
      \left\{ P_{a,w}^f:=\frac{N^a_{f(a),w}}{(1+\eta)2^{n-m}}
      \right\}_{(a,w) \in \bool^{n\times \ell}}\enspace,
    \end{equation*}
    together with $P_\bot^f=\id-\sum_{a,w}P^f_{a,w}$, form a POVM on the prover's register $P$.
 \item Suppose
    $\{P_{a,w}^f\}_{a,w}$ can be completed
    as a POVM, then let $\adv^f_n$ be the attack where the adversary applies POVM
    $\{P_{a,w}^f\}_{a,w}\cup  \{P_\bot^f\}$ to its register, upon outcome $(A,W)$ sets $Y=f(A)$ or $Y=\bot$ if
    outcome $\bot$ is observed and sends $(A,Y,W)$ to the verifier. The probability of error $\epsilon(\adv^f_n)$
    of this attack satisfies $1- \mathbb{E}_f[\epsilon(\adv^f_n)]\geq
    \delta - \negl[n-m]$ where $\delta$ is the probability that $\verifier$ accepts in an honest execution.
\end{enumerate}
  \end{theorem}
  \begin{proof}
    Let $2^k$ be the dimension of register $P$.
    Consider the subset of measurement operators $N^a_{f(a),w}$ from the honest POVM ${\cal N}^a$ that yield the intended outcome for the cheating prover, i.e. on input $a$ gives outcome $y=f(a)$, and define the operators $X_a^f=\sum_{w\in\bool^\ell} N^a_{f(a),w}$. We have that
    \begin{align*}
      \mathbb{E}_{f} [X^{f}_a] =\mathbb{E}_{f}\sum_{y,w\in\bool^{m\times \ell}} I_{\mathcal{E}^a_y}(f)\cdot  N^a_{y,w} = \sum_{y,w\in\bool^{m\times \ell}}\mathbb{E}_{f} [I_{\mathcal{E}^a_y}(f)]\cdot N^a_{y,w}=  \frac{\id_{P}}{2^m}
    \end{align*}
    where $I_{\mathcal{E}^a_{y}}$ is the indicator function for the event
    $\mathcal{E}^a_{y}=\{f\mid y=f(a)\}$ which has probability
    $\frac 1{2^m}$ for any $y$ and $a$ since every value for $f(a)$ is equally likely.

    Applying the Chernoff bound with $D=2^k$, $M=2^m$, and $\alpha=\frac 1{2^m}$ to the weighted sum over $a$ of the operators $X^f_a$, we have
    \begin{equation*}
      \Pr_{f}\left[ \frac{1}{2^n} \sum_{a\in\bool^n} X_a^f \nleq (1+\eta) \frac{\id}{2^m} \right]
      \leq 2^{k+1} \exp\left( - \frac{1}{2 \ln 2} \cdot \frac{2^n}{2^m}\cdot\eta^2 \right)\enspace .
    \end{equation*}
    This bound becomes negligible in $n$ if we choose $\eta = \sqrt{2\ln 2
      (n+k)\frac {2^m}{2^n}}< \frac 12$. Therefore, except with probability $\negl[n]$,
    \begin{equation}\label{eq:sum-x-bounded}
       \frac{1}{2^n} \sum_{a\in\bool^n} X^f_a= \frac{1}{2^n} \sum_{\substack{{a\in\bool^n}\\ {w\in\bool^\ell}}} N^a_{f(a),w}
      \leq (1+\eta) \frac{\id}{2^m} \enspace.
    \end{equation}
    Define the ensemble of operators $P^f_{a,w}$ by 
    \begin{equation*}
      P^f_{a,w}:=\frac{N^a_{f(a),w}}{(1+\eta)2^{n-m}}\enspace.
    \end{equation*}
    Then, when~\eqref{eq:sum-x-bounded} holds, the set of operators 
    $\{P^f_{a,w}\}_{a,w}$ forms a POVM when completed with $P_\bot^f=\id-\sum_{a,w}P^f_{a,w}$.

    This gives rise to an attack $\adv^f_n$ where the adversary applies the
    above POVM to obtain $(a,w)$ and sets $y=f(a)$ to send $(a,y,w)$ to the verifier.
 The probability of error $\epsilon(\adv^f_n)$ 
  for this attack corresponds to the probability of obtaining outcome ``$\bot$'' or of being
 rejected by the verifier. We have, 
 \begin{equation*}
   1-\epsilon({\cal A}^f_n)= \sum_{(a,w)\in\bool^{n\times \ell}}\trace{\left(P^f_{a,w}\otimes V^{a,f(a),w}_1\right)\Psi_{PV}}\enspace .
 \end{equation*}
 Recall that $\Pi(y|a)=\sum_w\trace{N^a_{y,w}\otimes V_1^{a,y,w}\Psi_{PV}}$
   is the probability that the verifier accepts the outcome $y\neq \bot$ for the
   prover input $a$ in an honest execution and that by the $\delta$--correctness of $\Pi_{\wotro}$, $\frac 1{2^n}\sum_a\sum_{y\neq \bot}\Pi(y|a)\geq \delta$. By~\eqref{eq:sum-x-bounded}, the above probability can, on average over $f$, be upper-bounded as follows:
\begin{align*} \mathbb{E}_f[1-\epsilon({\cal A}^f)]
         &= \frac 1{(1+\eta)2^{n-m}}\cdot \mathbb{E}_f\left[ \sum_{(a,w)\in\bool^{n\times \ell}}\trace{\left( N_{a,f(a),w}\otimes V^{a,f(a),w}_1\right)\Psi_{PV}}\right]\\
         &= \frac 1{(1+\eta)2^{n-m}} \frac 1{2^m}\sum_{a\in\bool^{n}}\sum_{(y,w)\in\bool^{m\times \ell}}\trace{ \left(N_{a,y,w}\otimes V^{a,y,w}_1\right)\Psi_{PV}}\\
         &= \frac 1{(1+\eta)2^{n}} \sum_{a\in\bool^n} \sum_{y\neq \bot} \Pi(y|a)\\
   &= \frac{\delta}{(1+\eta) }\\
   &\geq (1-\eta)\delta\enspace,
 \end{align*}
 which is approximately $\delta$ since $\eta$ is $\negl[n-m]$.
  \qed
  \end{proof}

\subsection{Oracle Access Quantum Circuits}

Establishing the security of $\Pi$ by \emph{black-box reduction} to a
cryptographic game ${\cal G}=(\Gamma,c)$ is defined by a (classical or quantum
but efficient) machine $\macho{M}{\adv}$ such that $\macho{M}{\adv}(1^n)$
produces a quantum circuit (with oracle access) made out of some universal set
of quantum gates together with \emph{oracle access} to the standard interface of
any adversary $\adv=\{\adv_n\}_n$ against $\Pi$ such that if $\adv_n$ breaks
$\Pi$ then $\macho{M}{\adv}(1^n)$ wins game $\mathcal{G}$. Let us first define
this machine $\mach{M}^{\adv}$ producing the circuit that will be called a
\emph{reduction} in the following.

 \begin{defi}[oracle access circuit]\label{def:oracle-machine}
  A \emph{quantum oracle access machine} $\macho{M}{(\cdot)}$ 
  for oracle $\oracle=\{\oracle_n\}_n$ 
  is a polynomial-time Turing machine that,
 on input $1^n$, outputs the description of a quantum circuit over a universal 
 set of quantum gates along with a special quantum gate: 
 $\oracle_{n}: \dens{P}\rightarrow\dens{R}$ with standard 
interface $P$ for the input and $R$ for the output. 
The circuit produced by $\macho{M}{\oracle}(1^n)$ is called an \emph{oracle
access quantum circuit}.
The oracle
calls behave as a CPTP map as the internal register $E$ is not part of the interface.
The $i$--th call to $\oracle_n$ is denoted $\oracle^i_n: \dens{P_i}\rightarrow \dens{R_i}$. 
 We say $\macho{M}{\oracle}$ makes $q(n)$ oracle queries if we can 
  represent the action of circuit ${\cal M}_n^{\oracle}$ produced by
   $\macho{M}{\oracle}(1^n)$
   on initial state $\ket 0$ on all registers of the circuit as the CPTP map 
  \begin{equation*}
    {\cal M}_n^{\oracle_n}\left( \proj {0}\right):={U}^q_n\circ {\cal O}^q_n\circ 
    \dots
    \circ {U}^1_n\circ {\cal O}^1_n \circ
    {U}^0_n\left( \proj {0}\right)
  \end{equation*}
  where ${U}^i_n\in \lxy{R_i \otimes Q_i}{P_{i+1}\otimes Q_{i+1}}$ are unitaries
  made out of the universal set of gates representing the action of circuit
  ${\cal M}_n^{\oracle_n}$ between the calls to $\oracle_n$. Register $Q_i$ is
  the reduction's working register before the action of $U^i_n$.
\end{defi}

As defined above, an oracle access circuit performs each oracle call 
using exactly the same functionality. No information can be kept by the oracle
between calls. We then say that the oracle 
is \emph{stateless}.
In general, an oracle could be allowed to store quantum information between
calls. This information is unavailable through the oracle input-output
interface but
is passed from one call to  the next. These oracle access circuits 
are said to be \emph{stateful}.

\subsection{Quantum Black-Box Reductions}\label{sect:bbox}
It remains to define what we mean exactly by polynomial-time \emph{black-box reductions}.
This notion was introduced by Impagliazzo and Rudich in \cite{IR89}
after observing that most proofs 
establishing the security of a crypto primitive  constructed from
one-way functions consider only  the input-output behaviour 
of the function. In other words, the one-way function is only
used as a black-box to construct the primitive and to prove its security. 
In \cite{IR89},  Impagliazzo-Rudich show that if it is possible to establish
the security of a secret-key agreement based solely on the input-output
behaviour of a one-way function then this security proof also establishes
that $\mathbf{P}\neq \mathbf{NP}$. Reingold, Trevisan, and Vadhan in~\cite{RTV04}
introduce three variants of black-box reductions called fully-BB,
semi-BB, and mildly-BB from the stronger to the weaker
flavour. In~\cite{BBF13}, fully black-box reductions are described
informally as follows:
\begin{quote}
A fully black-box reduction $\mach{R}$ is an efficient algorithm that transforms any (even inefficient) adversary 
$\adv$, breaking any instance $\Pi^f$ of primitive $\mathcal{P}$, into an
algorithm $\macho{R}{\adv,f }$ breaking the instance $f$ of $\mathcal{Q}$. 
Here, the reduction treats both the adversary
as well as the primitive as  black-boxes, and $\Pi^f$ denotes the (black-box) construction out of $f$.
\end{quote}
In our setting and as in~\cite{bitansky_why_2013,bgw12}, we consider proofs 
establishing the security of protocol $\Pi$
(for $\wotro$) by providing an efficient oracle access quantum circuit $\macho{R}{(\cdot)}$
 with the property that for any $\adv$ breaking $\Pi$,  $\macho{R}{\adv}$
 wins game $\mathcal{G}$ (to be more precise, $\macho{R}{\adv,\Gamma}$ wins
 game $\mathcal{G}=(\Gamma,c)$). The adversary $\adv$ against $\Pi$
  is therefore only used through its standard intput-output
  interface in reduction $\macho{R}{\adv}$ (i.e.\ as a quantum channel). This is what we call
  a quantum fully black-box reduction (or \fbb{} reduction), the quantum version of a \emph{fully black-box
  reduction}, also called a \emph{BBB reduction} in~\cite{BBF13}.
Our definition agrees with other works in which fully  black-box reductions are defined in the quantum setting~\cite{ananth_post-quantum_2022,hosoyamada_finding_2020}. 

\begin{defi}[Quantum fully black-box reduction to a crypto game]
  \label{def:quantum-fully-bb}
Consider $\Pi$ a protocol and let $\adv=\{\adv_n\}_n$ be 
an adversary provided through
its standard interface. We say that \emph{the security of $\Pi$ is established by 
quantum fully black-box reduction to crypto game $\mathcal{G}=(\Gamma,c)$}
if there exists an efficient oracle access circuit $\macho{R}{(\cdot)}$,
called \emph{the reduction},
such that when $\adv$ breaks $\Pi$ then 
\[
\Pr{\left[\langle \macho{R}{\adv}(1^n)\rightleftharpoons \Gamma(1^n)\rangle=1 \right]} \geq c+\frac{1}{\poly}\enspace.\]
\end{defi}

\subsection{Efficient Simulation of the Chernoff Attack}\label{sec:simul-above-attack}

We show that no reduction $\macho{R}{\adv}$ can establish the security of a
\wotro{} protocol by quantum black-box reduction to game assumption $\cal G$.
The reason for this state of affair is that the Chernoff attacks described in
Theorem~\ref{thm:attaque-simulable} is \emph{efficiently simulatable}. This
means that there is an efficient algorithm $\simulator_n$ such that no oracle
machine can tell whether it is given oracle access to the inefficient Chernoff
adversary $\adv^f_n$ hitting a random function $f(\cdot)$ or to $\simulator_n$
that does not know anything about $f(\cdot)$.

 The most general attack against a \wotro{} protocol in the \crqss{} model (Definition~\ref{def:attack}) is a POVM on the prover's part of the \crqss{} that produces a classical message which makes the verifier accept the output $c=f(a)$ with high probability. The attack takes no input other than the prover's register
 $P$ of the \crqss{}\ 
and produces its output in 
registers $A\otimes Y\otimes W$.
Let $\mathcal{F}^*$ be the set of functions defined in
Theorem~\ref{thm:attaque-simulable}, i.e.\ such that for $f\in\mathcal{F}^*$, $\sum_{a,w}P^f_{a,w}\leq \id$ so that $P^f_\bot=\id-\sum_{a,w}P^f_{a,w}\geq 0$ and the set of operators $\{P^f_{a,w}\}_{(a,w)}\cup \{P^f_\bot\}$ forms a POVM. 
 For $f\in\mathcal{F}^*$, the adversary $\adv^f_n$ defined in 
 Theorem~\ref{thm:attaque-simulable} 
 can be implemented by the following isometry $\adv^{f}_n\in \lxy{P}{R\otimes E}$, where $R=A\otimes Y\otimes W\approx
 \hilbert_n\otimes \hilbert_{m(n)}\otimes \hilbert_{\ell(n)}$
 and where $E=E'\otimes P\approx
 \hilbert_{p(n)}\otimes\hilbert_{n}$, for $\ell(n),p(n)$ polynomials:
 \begin{multline}\label{chernadviso1}
  \adv^{f}_n : \ket{\psi}_P \mapsto \sum_{\substack{a\in\{0,1\}^n\\ w\in\{0,1\}^{\ell(n)}}}
 \hspace{-0.2cm}
   \ket{a,f(a),w}_{AYW}\otimes \ket{a,f(a),w}_{E'}\otimes \sqrt{P^f_{a,w}}\ket{\psi}_{E''} \\
   + \ket{\bot,\bot,\bot}\otimes\ket{\bot,\bot,\bot}\otimes \sqrt{P^f_{\bot}}\ket{\psi}_{E''}\enspace,
\end{multline}

Let $\adv^{\mathcal {F}^*}_n=\{\adv^f_n\}_{f\in\mathcal{F}^*}$ be the
\emph{family of all Chernoff adversaries} against protocol $\Pi$ implementing
$\wotro^{n,m}$. The standard output interface of any adversary $\adv^f_n$ is
made out of registers $A\otimes Y\otimes W$ while register $E= E'\otimes P$ is
the working register of the adversary. It is easy to verify that any quantum
black-box reduction $\macho{R}{(\cdot)}$ establishing the security of $\Pi$ by
quantum black-box reduction to game $\mathcal{G}$ is such that
$\macho{R}{\adv^f_n}$ wins $\mathcal{G}$ even when $f \in_R \mathcal{F}^*$.
Next, we define what it means for the family of all Chernoff adversaries to be
simulatable.

\begin{defi}[Simulatable Attack for \wotro{}]
  Let $n \in\naturals$, $m(n)\leq n$, $\Pi$ a $\wotro^{n,m}$ protocol, and
  $\adv^{\mathcal{F}^*}_n$ be the family of adversaries defined above. We say
  that $\adv^{\mathcal{F}^*}_n$ is \emph{efficiently $\epsilon(n)$--simulatable}
  if there exists a family of polynomial-time quantum algorithms
  $\simulator=\{\simulator_n\}_n$, called the simulator, such that
  \begin{itemize}
  \item The success probability of ${\cal A}^{f}_n$ is at least $1-\negl[n-m]$ on average over $f\in_R {\cal F}^*$. 
  \item For every (possibly inefficient) oracle acess machine $\macho{M}{(\cdot)}$ making $q(n)=\poly[n]$ queries to its oracle, 
  the CPTP map ${\cal M}^{(\cdot)}_n$ describing the action of  circuit $\macho{M}{(\cdot)}(1^n)$ satisfies
    \begin{equation}
      \label{eq:2}
      \|\mathbb{E}_{f\in\mathcal{F}^*} [{\cal M}_n^{\mathcal{A}^{f}_n}(\proj 0)]- {\cal M}_n^{\simulator_n}(\proj 0_{R})\|_1 \leq \epsilon(n) \enspace.
    \end{equation}
  \end{itemize}
\end{defi}

Next theorem shows that the family adversaries $\adv^{\mathcal{F}}_n$ 
is efficiently simulatable. Unlike the simulator used in~\cite{bitansky_why_2013}
for their family of inefficient adversaries, 
our simulator is \textbf{not} stateful. The full proof is in Appendix~\ref{sec:proof-simulable}.

\begin{theorem}\label{thm:simulation-wotro}
  Let $n,m$ and $\Pi_\wotro^{n,m}$ be as in the statement of Theorem~\ref{thm:attaque-simulable}.  The family of adversaries $\adv^{\mathcal{F}^*}_n$ is efficiently $\negl[n-m]$--simulatable. 
\end{theorem}

Theorem~\ref{thm:attaque-simulable} and Theorem~\ref{thm:simulation-wotro} give a
simulatable attack against any $\wotro^{n,m}$ protocol where $n-m$ is linear in the security parameter $n$. 
We conclude,

\begin{corollary}\label{cor:bb-imposs-wotro}
    Let $\mathcal{G}$ be a cryptographic game assumption and let $\Pi^{n,m}$ be a 
    $\wotro^{n,m}$ protocol with $n-m\in\omega(\lg n)$. For $\delta\geq 1/\poly[n]$, if there is a 
    quantum black-box reduction showing that $\Pi^{n,m}$ $\delta$--avoids all functions assuming game 
    $\mathcal{G}$ then assumption $\mathcal{G}$ is false.
\end{corollary}

\subsubsection{Proof of Theorem~\ref{thm:simulation-wotro}}
\label{sec:proof-thm-simul}

Consider a protocol $\Pi$ for $\wotro^{n,m}$ using  a \crqss\  $\ket{\Psi}_{PV}$
with $\dim(P) = 2^k$.
Let $\mathcal{F}$ be the set
of all functions $f:\{0,1\}^n\rightarrow \{0,1\}^m$. 
We consider the following isometric implementation of the Chernoff adversary $\adv^f$ 
described
in Section~\ref{sec:simul-above-attack}  using internal working quantum register $E=E'\otimes E''$,
\begin{multline}\label{eq:iso-f}
  \adv^{f}_n : \ket{\psi}_P \mapsto \sum_{\substack{a\in\{0,1\}^n\\ w\in\{0,1\}^{\ell(n)}}}
 \hspace{-0.2cm}
   \ket{a,f(a),w}_{AYW}\otimes \ket{a,f(a),w}_{E'}\otimes \sqrt{P^f_{a,w}}\ket{\psi}_{E''} \\
   + \ket{\bot,\bot,\bot}\otimes\ket{\bot,\bot,\bot}\otimes \sqrt{P^f_{\bot}}\ket{\psi}_{E''}\enspace,
\end{multline}
with input register $P$ and output register $A\otimes Y\otimes W$. 
Remember however that Theorem~\ref{thm:attaque-simulable} does not
guarantee that for all $f\in\mathcal{F}$, $\adv^f_n$ is a POVM
(in which case (\ref{eq:iso-f}) is not an isometry).  It only
tells us that there exists $\mathcal{F}^*\subseteq \mathcal{F}$
such that $\forall f\in\mathcal{F}^*$, $\adv^f_n$ implements a valid POVM
(and therefore, (\ref{eq:iso-f}) is indeed an isometry)
and $\Pr{[f\in\mathcal{F}^*]}\geq 1-\negl[n]$. When $f\notin \mathcal{F}^*$,
the implementation of $\adv^f_n$ defined in~(\ref{eq:iso-f}) is not an isometry
but is still a linear map, though not a physically realizable one.
% We consider without loss of generality that ${P^f_{\bot}}$, the outcome
%  corresponding to an
% error, is always a positive operator even when $f\notin \mathcal{F}^*$. 
Whenever $f\notin\mathcal{F}^*$, we set ${P^f_{\bot}}=0$ such that it is always
positive semidefinite for all $f\in\mathcal{F}$.
% That way, we always have that for all $f\in\mathcal{F}$,
% $\sum_{a,w} {P^f_{a,w}} + {P^f_{\bot}}\geq \id_{E''}$.
When $f\in\mathcal{F}^*$,
we  have $\sum_{a,w} {P^f_{a,w}} + {P^f_{\bot}}= \id_{E''}$ 
as $\{P^f_{a,w}\}_{a,w}\cup\{P^f_{\bot}\}$
is a valid POVM. Otherwise, when $f\notin\mathcal{F}^*$, $\sum_{a,w} 
{P^f_{a,w}} + {P^f_{\bot}} \geq \id_{E''}$.

Now, consider the simulator $\simulator=\{\simulator_n\}_n$ where $\simulator_n$ is
defined as follows:
\begin{blockquote}
  {\bf $\simulator_n$: }
  \begin{enumerate}
  \item Pick $a\in_R\{0,1\}^n$, 
  \item Apply the honest POVM $N^a$ to register $P$ to get
    outcome $(c,w)$,
  \item Output $(a,c,w)$.
  \end{enumerate}
\end{blockquote}
This simulator corresponds to the isometry
\begin{equation}
  \label{eq:iso-sims}
\simulator_{n}: \ket{\psi}_P \mapsto 2^{-n/2}\hspace{-0.4cm}
\sum_{\substack{a\in\{0,1\}^n\\y\in\{0,1\}^m \\w\in\{0,1\}^{\ell(n)}}} \hspace{-0.4cm}
\ket{a,y,w}_{ACW}\otimes \ket{a,y,w}_{E'}\otimes \sqrt{N^a_{y,w}}\ket{\psi}_{E''}\enspace,
\end{equation}
with the same input-output interface than any adversary
against $\Pi$. The above simulator is efficiently implementable since it only purifies
the honest prover's measurement. It is not too difficult to show
that if POVM $\mathcal{N}^a=\{N^a_{y,w}\}_{a,y,w}$ can be implemented efficiently
for all $a\in\{0,1\}^n$ then the isometry~(\ref{eq:iso-sims}) is efficient.
Notice that the simulator never produces an error as we assume that 
the honest strategy in $\Pi$ never produces an error. We could have
allowed  a protocol for  $\wotro^{n,m}$ to produce an error with negligible
probability in $n$. This would not cause any problem with what we establish
in the following. 

\newcommand{\esper}[1]{\ensuremath{\underset{#1}{\mathbb{E}}}}
\renewcommand{\ext}{\ensuremath{\mathsf{ext}}}

Before going further, we use the operator Chernoff bound of lemma~\ref{chernono}
to establish a few useful properties of the Chernoff adversaries.
The following is a direct consequence of the operator bound.
\begin{lemma}\label{tech4}
Let $\Pi$ be a protocol for  $\wotro^{n,m}$ with POVM $\mathcal{N}^a=\{N^a_{y,w}\}_{y,w}$
for  prover $\prover_\wotro$.
Consider the Chernoff adversaries $\{\adv^f_n\}_{f\in\mathcal{F}}$
against $\Pi$ as defined in Theorem~\ref{thm:attaque-simulable}
with $\eta=\sqrt{2\ln 2(n+k)\cdot 2^{m-n}}$.
Let $P^f := \{P^f_{a,w}\}_{a,w}$ be the POVM applied by $\adv_n^f$
where $P^f_{a,w}=\frac{N^a_{f(a),w}}{2^{n-m}\left(1+\eta\right)}$.
Then, for any $\frac{1}{2\eta}\geq  t>1$,
\[
\Pr_{f\in \mathcal{F}}{\left[\sum_{a,w} P^f_{a,w} \notin \left[\frac{(1-t\eta)}{1+\eta}\id_P, \frac{(1+t\eta)}{1+\eta}\id_P\right]\right]}
\leq 2^{-n\cdot t^2} \enspace.
\]
\end{lemma}
The proof is rather direct and can be found in Appendix~\ref{sec:proof-simulable}.

Suppose for 
a contradiction that 
$\macho{R}{\adv}$ is a reduction that, for any 
successful adversary $\adv$  against protocol $\Pi$,
produces a circuit
that wins game $\mathcal{G}=(\Gamma,c)$ using $q(n)=\poly$ queries to $\adv$.
We show that oracle access circuits $\macho{R}{\adv^f}(1^n)$ and
$\macho{R}{\simulator}(1^n)$ produce states at 
negligible trace-norm distance when both are evaluated on
$\proj{0}$ and when $\adv^f$ is picked with $f\in_R\mathcal{F}^*$. 
Let $\mathcal{R}^{\adv_n}_n :=  \macho{R}{\adv}(1^n)$ be the oracle-access circuit
produced by the reduction with security parameter $n$ upon oracle $\adv_n$.

The first thing to observe is that picking $f\in_R\mathcal{F}^*$ 
in reduction $\mathcal{R}^{\adv_n^f}$ is essentially the same 
as running the reduction with $f\in_R \mathcal{F}$, even though
in this case  $\mathcal{R}^{\adv_n^f}$ may  not be physically realizable.
\begin{lemma} \label{tech3}
Let $\mathcal{F}$ and $\mathcal{F}^*$ be defined as above
for $n,m\in\N$. For $f\in\mathcal{F}$, 
consider adversary $\adv^f_n$ defined in (\ref{eq:iso-f}). Then,
\[
 \left\|\esper{f\in \mathcal{F}^*}\left [{\cal R}_n^{\mathcal{A}^{f}_n}(\proj 0)\right]- 
\esper{f\in\mathcal{F}}\left[{\cal R}_{n}^{\adv^{f}_n}(\proj 0)\right]\right\|_1
  \leq 
\negl[n]\enspace.
 \]
\end{lemma}
The proof of this lemma can also be found in Appendix~\ref{sec:proof-simulable}. 

As a direct consequence of lemma~\ref{tech3}, we get
\begin{equation}
\begin{split}
\Big\|\Big.{\cal R}_n^{\simulator}(\proj 0) &- \Big.
 \esper{f\in\mathcal{F}^*}\left [{\cal R}_n^{\mathcal{A}^{f}}(\proj 0)\right]\Big\|_1
 \leq \\
 &\hspace{0.7in} \negl[n] + \left\|
 {\cal R}_n^{\simulator}(\proj 0)-\esper{f\in\mathcal{F}}\left [{\cal R}_n^{\mathcal{A}^{f}}(\proj 0)\right]\right\|_1 \enspace.
\end{split}
\end{equation}

To bound the trace-distance between $\mathcal{R}^{\adv^f}_n(\proj 0)$ for
$f\in_R\mathcal{F}$ and $\mathcal{R}^{\simulator_n}_n(\proj 0)$ when
$\mathcal{R}^{(.)}_n$ is an oracle access circuit with $q:=q(n) \in \poly$
queries, we use $q+1$ hybrid reductions where hybrid $i$ acts as $\simulator_n$
on the first $i$ queries and acts as $\adv^f_n$ on the remaining $q-i$ queries.
In the following, we denote by $\mathcal{R}^{\simulator, \adv^f}_{n,j}$ the
oracle-access circuit $\macho{R}{\adv}(1^n)$ where the first $j$ calls are made
to oracle $\simulator_n$ and the last $q-j$ calls are made to $\adv^f_n$. We
therefore have that $\mathcal{R}^{\simulator, \adv^f}_{n,q}$ corresponds to
$\mathcal{R}_n^{\simulator}$ and $\mathcal{R}^{\simulator, \adv^f}_{n,0}$
corresponds to $\mathcal{R}_n^{\adv^f}$, and
 \begin{equation}\label{lebordel}
 \begin{split}
 \left\|
 {\cal R}_n^{\simulator}(\proj 0_{R})- \esper{f\in \mathcal{F}}\left [ {\cal R}_n^{\mathcal{A}^{f}}(\proj 0)\right]\right\|_1
 = \\
\left\|\esper{f\in\mathcal{F}}\left[{\cal R}_{n,q}^{\simulator, \adv^f}(\proj 0) - \right.\right.&
\left. \left.{\cal R}_{n,0}^{\simulator, \adv^f}(\proj 0)\right] \right\|_1\enspace.
 \end{split}
 \end{equation}
 By a standard hybrid argument, the right-hand side of (\ref{lebordel})  is upper bounded 
 as follows:
 \begin{equation}\label{notrebut}
 \begin{split}
 \left\| \esper{f\in \mathcal{F}}\left[
 {\cal R}_{n,q}^{\simulator, \adv^f}(\proj 0)- \right.\right.& \left.\left.
{\cal R}_{n,0}^{\simulator, \adv^f}(\proj 0) \right]\right\|_1
 \leq  \\ & \sum_{j=1}^q \left\| \esper{f\in\mathcal{F}}\left[ 
 {\cal R}_{n,j}^{\simulator, \adv^f}(\proj 0)- {\cal R}_{n,j-1}^{\simulator, \adv^f}(\proj 0)
 \right]  \right\|_1\enspace.
 \end{split}
 \end{equation}
 We now upper bound $ \left\| \esper{f\in\mathcal{F}}\left[ {\cal
       R}_{n,j}^{\simulator, \adv^f}(\proj 0)- {\cal R}_{n,j-1}^{\simulator,
       \adv^f}(\proj 0) \right] \right\|_1$ for $j\in\{1,\ldots,q\}$. Notice
 that the working registers (register $E'$ in (\ref{eq:iso-f})) of any oracle
 call can be measured without modifying the behaviour of the reduction as these
 registers are not under its control, these measurements all commute with the
 operations of the reduction. For any $j\in\{1,\ldots,n\}$, circuits ${\cal
   R}_{n,j}^{\adv^{f},\simulator}$ and $ {\cal R}_{n,j-1}^{\simulator, \adv^f}$
 differ only in the $j$-th query, which is made to $\simulator_n$ in ${\cal
   R}_{n,j}^{\simulator, \adv^f}$ and to $\adv^f_n$ in ${\cal
   R}_{n,j-1}^{\simulator, \adv^f}$. Otherwise, both ${\cal
   R}_{n,j}^{\simulator, \adv^f}$ and ${\cal R}_{n,j-1}^{\simulator, \adv^f}$
 query $\simulator_n$ for all queries prior to the $j$-th and both query
 $\adv^f_n$ for all queries following the $j$-th. Let
 $\mathcal{S}=\{0,1\}^n\times\{0,1\}^m\times \{0,1\}^{\ell(n)}$ be the set of
 possible announcements $(a,y,w)$ for a prover in $\Pi$ except when an error
 occured (i.e. when $a=\bot$ is obtained). For $j\geq 1$, let
 $S^{j-1}=(S^{j-1}_1,\ldots,S^{j-1}_{j-1})\in\mathcal{S}^{j-1}$ be the random
 variable for the outcomes of the $j-1$ first queries to $\simulator_n$ in
 ${\cal R}_{n,j}^{\simulator, \adv^f}(\proj{0})$ and ${\cal
   R}_{n,j-1}^{\simulator, \adv^f}(\proj{0})$, where $S^{j-1}_h$, for
 $h\in\{1,\ldots, j-1\}$, represents the \emph{result} of the $h$--th call.
 Remember that the portion of the adversary's circuit up to but not including
 the $j$--th call is an isometry as it is independent of $f\in\mathcal{F}$. This
 independence of all $j-1$ first outcomes is important in applying the hybrid
 argument. Only querying $\adv^f_n$ can produce the special error outcome
 $(\bot,\bot,\bot)$ and only querying $\adv^f_n$ with $f\notin\mathcal{F}^*$ for
 the $j$--th query can transform the state of the reduction before the $j$--th
 query into a non-physical one, as its trace-norm could exceed $1$. Remember
 that outcome $S^{j-1}_h=(a,y,w)$ corresponds to the outcome when
 \emph{measuring} in the computational basis the internal register $E'$ of the
 $h$--th call to $\adv^f_n$. We say that $S^{j-1}$ \emph{is confused about $a$}
 if $S^{j-1}_h = (a,y,w)$ and $S^{j-1}_{h'}=(a,y',w')$ for some $h\neq h'$ and
 $y \neq y'$. For $s\in \mathcal{S}^{j-1}$, we denote by $Q_{S^{j-1}}(s)$ the
 probability of results $s$ for the $j-1$ first calls (to $\simulator_n$) in
 ${\cal R}_{n,j}^{\simulator, \adv^f}(\proj{0})$. By construction, this also
 corresponds to the probability of $s$ for the $j-1$ first calls in ${\cal
   R}_{n,j-1}^{\simulator, \adv^f}(\proj{0})$.
\newcommand{\ssfa}{\ensuremath{s[\textsf{a}]}}
\newcommand{\error}[1]{\ensuremath{\Delta^{#1}_{\mbox{\tt e}}}}
For $s\in\mathcal{S}$, we let
$\ket{\psi_{j}(s)}$ be the state 
  obtained just prior the $j$-th query in both
  ${\cal R}_{n,j-1}^{\simulator, \adv^{f}}(\proj{0})$
  and
  ${\cal R}_{n,j}^{\simulator, \adv^{f}}(\proj{0})$
  given that registers $E''_1,\ldots,E''_{j-1}$ have each 
  been measured in the computational basis to get $s$.
  In the following, we abuse the notation and 
  write ${\cal R}_{n,j-1}^{\simulator, \adv^{f}}(\proj{\psi_j(s)})$ and
  ${\cal R}_{n,j}^{\simulator, \adv^{f}}(\proj{\psi_j(s)})$
  to denote the result of the each hybrid reductions 
  when $\ket{\psi_{j}(s)}$  is used for the $j$--th query 
  onward. ${\cal R}_{n,j-1}^{\simulator, \adv^{f}}(\proj{\psi_j(s)})$
  will make all its remaining queries to $\adv^f_n$ while
  ${\cal R}_{n,j}^{\simulator, \adv^{f}}(\proj{\psi_j(s)})$
  will query $\simulator_n$ one last time before querying
  $\adv^f_n$.
\newcommand{\pjs}[1]{\ensuremath{\mathbb{S}_{#1}}}

   Let $F$ be a random variable
  uniformly distributed in $\mathcal{F}$. For $s\in\mathcal{S}^{j-1}$, 
  let $\mathbb{S}_s$ be the projector on the subspace producing 
  outcomes $s$ when registers $E''_1,\ldots, E''_{j-1}$ of the $(j-1)$--th
  first calls to $\simulator_n$  are each measured in the computational basis.
  By construction of the simulator, $\{\pjs{s}\}_{s\in\mathcal{S}^{j-1}}$ defines
  a complete Von Neumann measurement
  of register $\bigotimes_{i=1}^{j-1}E''_i$ provided by the $j-1$ first calls to $\simulator_n$.
  We have,
  \begin{align}
  \left\| \esper{f}\right.&\left.\left[
 {\cal R}_{n,j}^{\simulator,\adv^{f}}(\proj 0)- {\cal R}_{n,j-1}^{\simulator,\adv^{f}}(\proj 0)
 \right] \right\|_1 \nonumber \\
 =& \left\| \sum_{f\in\mathcal{F}} \Pr{[F=f]}\left( {\cal R}_{n,j}^{\simulator,\adv^{f}}(\proj 0)- 
 {\cal R}_{n,j-1}^{\simulator,\adv^{f}}(\proj 0)\right)\right\|_1 \nonumber \\
=& \left\|\sum_{\substack{s\in \mathcal{S}^{j-1}\\ f\in\mathcal{F}}}\hspace{-0,1cm}
\Pr{[F=f]}  
 \left( \pjs{s}{\cal R}_{n,j}^{\simulator,\adv^{f}}(\proj 0)\pjs{s}- \pjs{s}{\cal R}_{n,j-1}^{\simulator,\adv^{f}}(\proj 0)\pjs{s} \right)\right\|_1\nonumber \\
= & \left\|\sum_{\substack{s\in \mathcal{S}^{j-1}\\ f\in\mathcal{F}}}\hspace{-0,1cm}
\Pr{[F=f]}  Q_{S^{j-1}}(s)
 \left( {\cal R}_{n,j}^{\simulator,\adv^{f}}({\psi_j(s)})- {\cal R}_{n,j-1}^{\simulator,\adv^{f}}({\psi_j(s)})\right)\right\|_1 \nonumber\\
 \leq & \underbrace{
\sum_{\substack{s\in \mathcal{S}^{j-1}\\ f\in\mathcal{F}}}\hspace{-0,1cm}
\Pr{[F=f]}  Q_{S^{j-1}}(s) \left\|
    {\cal R}_{n,j}^{\simulator,\adv^{f}}({\psi_j(s)})- {\cal R}_{n,j-1}^{\simulator,\adv^{f}}({\psi_j(s)})\right\|_1}_{\text{(D)}} \enspace. \label{acdc} 
\end{align}
We now find a negligible upper bound  for (\hyperref[acdc]{D}) in (\ref{acdc}).
This is where the work  to apply the hybrid argument
to the $j$--th query is done.
Given $s\in\mathcal{S}^{j-1}$, remember that 
$\ket{\psi_{j}(s)}$ is the state obtained (a pure state)
just prior to $j$--th query in both 
 ${\cal R}_{n,j}^{\simulator,\adv^{f}}
 (\proj 0)$ and ${\cal R}_{n,j-1}^{\simulator,\adv^{f}}(\proj 0)$.
Remember also that we denote by 
 ${\cal R}_{n,j}^{\simulator,\adv^{f}}
 (\psi_j(s))$ and ${\cal R}_{n,j-1}^{\simulator,\adv^{f}}(\psi_j(s))$
 the action of the circuit from the $j$--th query onward 
 with initial state $\ket{\psi_j(s)}$. The $j$--th query is made
 to $\simulator_n$ in  ${\cal R}_{n,j}^{\simulator,\adv^{f}}$ 
 and to $\adv^f_n$ in ${\cal R}_{n,j-1}^{\simulator,\adv^{f}}$.

 Let $q^{\adv^f_n}_{a,w}(\psi_{j}(s)):= \trace{(P^f_{a,w}\otimes \id_{Z_j})\proj{\psi_{j}(s)}_{P_jZ_j}}$
 be the \emph{likelihood} of outcome $(a,f(a),w)$ for the $j$--th query
made to $\adv^f_n$ upon  $\ket{\psi_{j}(s)}$. 
 Let $q_{\bot}^{\adv^f_n}(\psi_{j}(s))=
 \trace{P^f_{\bot}\proj{\psi_{j}(s)}}$ 
 be the likelihood
  that the $j$--th query to $\adv^f_n$ produces an error and let $\ket{\psi^{f,\bot}_j(s)}$
  the normalized vector obtained after the $j$--th query has produced an error
  (notice that if the $j$--th query is made to $\simulator_n$ then no error can be produced).
The set $\{q^{\adv^f_n}_{a,w}(\psi^f_{j}(s))\}_{a,w}\cup\{q_{\bot}^{\adv^f_n}(\psi^f_{j}(s))\}$
is not guaranteed to be a probability distribution when $f\notin \mathcal{F}^*$ (this is the reason
why we call these values \emph{likelihoods} instead of probabilities).  

Likewise, 
 $q^{\simulator_n}_{a,y,w}(\psi_{j}(s)):= 2^{-n}\trace{(N^a_{y,w}\otimes \id_{Z_j})\proj{\psi_{j}(s)}}$
 be the probability that $\simulator_n$ picks $a$ uniformly at random and observes $(y,w)$ when applying
 $\Pi$'s honest measurement $\mathcal{N}^a$ on vector $\ket{\psi_{j}(s)}$.
 Notice that by definition of  $\{P^f_{a,w}\}_{a,w}$,
 for all $(a,f(a),w)\in\mathcal{S}$,
 states 
 \[
 \left(\sqrt{P^f_{a,w}}\otimes \id_{Z_j}\right)\ket{\psi_{j}(s)} \text{ and } 
 \left(\sqrt{N^a_{f(a),w}}\otimes \id_{Z_j}\right)\ket{\psi_{j}(s)} 
 \]
 are identical once normalized.  Let
$\ket{\psi^{a,f(a),w}_j(s)}$ be that  state. 
  In the following, we write $(a,z)\in s$ if there exists $w\in\{0,1\}^{\ell(n)}$
  such that $(a,z,w)\in s$. We also write $a\in s$ if there exist $z,w$ such that
  $(a,z,w)\in s$. Let $\delta(s):=\{a\in\{0,1\}^n\,|\,a\in s\}$.
The sum over $f$ in (\hyperref[acdc]{D})
 can now be written as
 \begin{align}
&\frac{1}{\#\mathcal{F}}  \sum_{f\in \mathcal{F}}
 \text{tr}_{E_j}\left((\simulator_n \otimes \id_{Z_j})
 \proj{\psi_{j}(s)}(\simulator_n \otimes \id_{Z_j})^* \right.\nonumber \\
& \hspace{3.9cm} - \left.(\adv^f_n \otimes \id_{Z_j})
 \proj{\psi_{j}(s)}(\adv^f_n \otimes \id_{Z_j})^* \right) \nonumber \\
 &= 
   \sum_{\substack{\substack{a\in\{0,1\}^n\setminus \delta(s)\\y\in\{0,1\}^m \\ w\in\{0,1\}^{\ell(n)} }}} 
   \underbrace{\frac{2^m}{\#\mathcal{F}} \sum_{\substack{f\in\mathcal{F} \\ f(a)=y}}\hspace{-0.1cm}\left|q^{\simulator_n}_{a,y,w}(\psi_{j}(s))- 2^{-m}q^{\adv^f_n}_{a,w}(\psi_{j}(s))\right|}_{\text{(M)}}\label{pppM} \\
& \hspace{5,5cm} \proj{a,y,w}
  \otimes \proj{\psi^{a,f(a), w}_j(s)}\nonumber  \\
 & \hspace{0,5cm}+\underbrace{ \frac{1}{\#\mathcal{F}}\hspace{-0,5cm}\sum_{\substack{f\in\mathcal{F}\\a\in\delta(s) \\ z\in\{0,1\}^m\\ w\in\{0,1\}^{\ell(n)}}}
\hspace{-0,5cm}\left| q^{\simulator_n}_{a,z,w}(\psi_{j}(s))- q^{\adv^f_n}_{a,w}(\psi_{j}(s))\right|}_{\text{(A)}}  \proj{a,z,w}
   \otimes \proj{\psi^{a,f(a),w}_j(s)}\label{qqqA} \\
 &\hspace{0.5cm} + 
   \underbrace{\frac{1}{\#\mathcal{F}}
   \sum_{\substack{f\in \mathcal{F}}} q^{\adv^f_n}_{\bot}(\psi_{j}(s))}_{(\bot)}
 \proj{\bot,\bot,\bot}\otimes \proj{\psi^{f,\bot}_j(s)}
 \enspace.\label{patiel1}
 \end{align}
 In the equation above, (\hyperref[pppM]{M}) is the main difference between
 $\adv^f_n$ and $\simulator_n$, (\hyperref[qqqA]{A}) represents the outcomes in
 which $s$ is confused about $a$ and (\hyperref[patiel1]{$\bot$}) represents the
 adversary's inconclusive outcome. We refer to Appendix~\ref{sec:proof-simulable}
 for the proofs that (\hyperref[qqqA]{A}), (\hyperref[patiel1]{$\bot$}), and the
 main part (\hyperref[pppM]{M}) are all negligible in $n-m$.

 Putting things together using the bounds on
 (\ref{lapartieA}),(\ref{petitapetit}), and (\ref{unpasdeplus}), we conclude
 that
 \begin{equation}
   (\hyperref[acdc]{D}) \leq \negl[n-m] \enspace,
 \end{equation}
 and this negligible upper bound on 
 (\hyperref[acdc]{D}) is independent of $s\in\mathcal{S}^{j-1}$ and $1\leq j\leq q(n)$.
 We conclude from (\ref{acdc}) that for all
 $1\leq j \leq q(n)$,
 \begin{equation}\label{presquerendu}
   \left\| \esper{f}\left[ 
       {\cal R}_{n,j}^{\adv^{f},\simulator}(\proj 0)- {\cal R}_{n,j-1}^{\adv^{f},\simulator}(\proj 0)
     \right]  \right\|_1 \leq \negl[n-m]  \enspace.
 \end{equation}
 Finally, plugging (\ref{presquerendu}) into (\ref{notrebut}) completes the
 proof of Theorem~\ref{thm:simulation-wotro}.

\section{Black-Box Impossibility of Fiat-Shamir in the \crqss{} Model}
\label{sec:impo-fiat-shamir}

We assume the reader is familiar with $\sprs$ and the Fiat-Shamir transform. For more information, we refer to Section~\ref{sec:sprs-fiat-shamir}.

In this section, we consider the natural extension of the Fiat-Shamir transform in the CRQS model where the prover and verifier share an arbitrary entangled state $\ket{\varphi_{n,m}}$, the prover performs some measurement specified by $a$ on its part of the CRQS, sends the result to the verifier who performs its own measurement based on the prover's message. Since a universal instantiation of the Fiat-Shamir is required to transform any $\spr$ into a sound argument, the CRQS $\ket{\varphi_{n,m}}$, as well as the measurement operators of the prover and verifier must be independent of the actual $\spr$ and of the statement $x$. The quantum Fiat-Shamir transform proceeds as follows:
\begin{enumerate}
\item $\prover^{\fsl}$ computes $a=\prover(x,w)$ and performs some measurement $\mathcal{N}^a$ on its part of $\ket{\varphi_{n,m}}$ that yield classical outcomes $(c,v)$. It computes  $z=\prover_2(a,x,w,c)$, and sends  $(a,c,v,z)$ to $\verifier^{\fsl}$.
\item $\verifier^\fsl$ performs a binary-outcome measurement $\mathcal{V}^{a,c,v}$ on its part of $\ket{\varphi_{n,m}}$ and rejects if the outcome is $0$, and otherwise outputs $\verifier(x,a,c,z)$.
\end{enumerate}
We consider without loss of generality that all communication remains classical, since the CRQS could contain polynomially many EPR pairs allowing for the teleportation of quantum states from the prover to the verifier.

 An \emph{abstract} Fiat-Shamir transform that captures all of the above would look like the following. Since we are proving a negative result, we only ask that a universal instantiation of the Fiat-Shamir transform has constant soundness error (instead of $\negl[n]$). 
 \begin{defi}
   The Fiat-Shamir transform is given by  $\Pi_\fsl^{n,m}=(\prover_\fsl,\verifier_\fsl)$ where $\prover_\fsl$ takes as input the commitment $a\in\bool^n$ and outputs a challenge $c\in\bool^m$ and a auxiliary verification information $v$. $\verifier_\fsl$ takes input $(a,c,v)$ and outputs $\tt accept$ or $\tt reject$.  For a \spr{} $\Sigma=(\prover_\Sigma,\verifier_\Sigma)$, the Fiat-Shamir transform applied to $\Sigma$ is the non-interactive protocol $\Pi_\fsl^{n,m}[\Sigma]=(\prover,\verifier)$ defined as
\begin{enumerate}
\item $\prover$ computes $a=\prover_\Sigma^1(x,w)$ and runs $(c,v)\leftarrow \prover_\fsl(a)$. It computes  $z=\prover_\Sigma^2(a,x,w,c)$, and sends  $(a,c,v,z)$ to $\verifier$.
\item $\verifier$ runs $\verifier_\fsl(a,c,v)$ and rejects if $\verifier_\fsl$ rejects, and otherwise outputs $\verifier_\Sigma(x,a,c,z)$.
\end{enumerate}
The Fiat-Shamir transform $\Pi_\fsl^{n,m}$ is $(n,m)$--universal if for any \spr{} $\Sigma$, $\Pi_\fsl^{n,m}[\Sigma]$ is an argument with soundness error bounded above by some constant greater than zero. 
\end{defi}

Note that an instantiation of the Fiat-Shamir transform is also one for $\wotro$ (and vice-versa). More precisely, the \wotro{} protocol implied by Fiat-Shamir is the protocol where $\prover_\wotro$ invokes $\prover_\fsl$, sends $(a,c,v)$ to $\verifier_\wotro$ that outputs $(a,c)$ if $\verifier_\fsl(a,c,v)$ accepts. The main distinction between the two is that a secure protocol for \wotro{} needs to avoid all functions, whereas a universal instantiation of Fiat-Shamir only needs to avoid functions that are ``bad challenges'' functions for some \spr\ for language membership to $L$ upon some public parameter $x\notin L$ .

\subsection{Black-Box Impossibility of Universal Fiat-Shamir}
\label{sec:fs-implique-wotro}

We begin by defining what is a black-box reduction from \fsl{} to a cryptographic game assumption
similarly to how it is done in~\cite{bitansky_why_2013}.
\begin{defi}[Black-Box Reduction for Quantum Fiat-Shamir]\label{def:qfs-redux}
  Let $\mathcal{G}=(\Gamma,c)$ be a cryptographic game assumption and let $\Pi_{\fsl}^{n,m}$ be an instantiation of the Fiat-Shamir transform in the \crqss{} model. A \emph{black-box reduction showing the $(n,m)$--QFS--universality of $\Pi_\fsl^{n,m}$ under the assumption $\mathcal{G}$ in the \crqss{} model} is an oracle-access machine $\mathcal{B}^{(\cdot,\cdot,\cdot)}$ such that the following holds. Let
  \begin{enumerate}
  \item $\Sigma=(\prover,\verifier)$ be a \spr{} for a language $L$ with commitment length $n$ and challenge length $m$ that has perfect completeness and special soundness, and
  \item $\adv$ be a (possibly inefficient) attacker that breaks the computational soundness of the non-interactive proof system $\Pi_\fsl^{n,m}[\Sigma]$ with advantage $1-\negl[n]$.
  \end{enumerate}
  The reduction $\mathcal{B}$ has black-box access to \prover, \verifier{} and \adv{}, runs in time polynomial in the running times of \prover, \verifier{} and \adv{}, and $\mathcal{B}^{\prover,\verifier,\adv}$ has advantage at least $1/\poly[n]$ in game $\mathcal{G}$.
\end{defi}

 As mentioned previously, a \fsl{} protocol is essentially a \wotro{} protocol, albeit satisfying a 
weaker notion of security. 
 In particular, a \wotro{} protocol avoiding \emph{only} the ``bad challenge'' functions
 of  \sprs\ would be enough for \fsl{}. The impossibility to black-box reduce the
 security of \wotro{} to a cryptographic game, as expressed in Corollary~\ref{cor:bb-imposs-wotro},
 does not apply directly to Fiat-Shamir.

 To show black-box impossibility of \fsl{} in the \crqss{} model, we construct a family of $\sprs$ $\{\Sigma^f\}_{f:\bool^n\rightarrow \bool^m}$ such that $\Sigma^f$ has bad challenge function $f(\cdot)$ for any $f$. The verifier $\verifier^f$ in $\Sigma^f$ is not necessarily efficient, but we again exploit the simulation paradigm, where the inefficient adversary is replaced by an efficient indistinguishable simulator, to simulate this verifier in a way that is consistent with the adversarial prover. By definition of the reduction $\mathcal{B}^{(\cdot,\cdot,\cdot)}$, if an adversary $\adv^f$ breaks the soundness of $\Pi_\fsl[\Sigma^f]$, $\mathcal{B}^{\prover,\verifier^f,\adv^f}$ wins game $\mathcal{G}$. By replacing $(\verifier^f,\adv^f)$ with a pair of simulators $(\simulator_\verifier,\simulator_\adv)$ such that no $\poly[n]$--query machine can distinguish between the two pairs, we obtain an efficient algorithm $\mathcal{B}^{\prover,\simulator_\verifier,\simulator_\adv}$ breaking the security of $\mathcal{G}$. We formalize this joint simulation below and then prove the black-box impossibility result using the strategy outlined above and pictured in Fig.~\ref{fig:qfs-redux}.

\begin{defi}[Joint Simulatability]\label{def:sim-sim}
  A family of (possibly inefficient) algorithms $\{(\adv^f,V^f)\}_f$ that have access to the same (possibly inefficient) resource $f:\bool^n\rightarrow \bool^m$ are \emph{jointly simulatable} if there exist two \qpt{} stateful algorithms $\simulator_1$ and $\simulator_2$ that share a \emph{common state} and such that for any $\poly[n]$--query oracle access machine $M^{(\cdot,\cdot)}$,
  \begin{equation*}
    \left| \Pr_f[ M^{(\adv^f,V^f)}=1]- \Pr[ M^{(\simulator_1,\simulator_2)}=1] \right|\leq \negl[n]
    \enspace .
  \end{equation*}
\end{defi}

\begin{figure}
  \centering
  \begin{minipage}{0.3\linewidth}
    \begin{tikzpicture}
      [every node/.style={minimum size=7mm, node distance=1cm},thick]

      \node[draw, minimum width=27mm] (redux) {$\mathcal{B}$} ; \node[draw,
      below of=redux, minimum width=27mm] (challenger) {$\Gamma$} ; \node[draw,
      above of=redux] (verifier) {$\verifier^f$}; \node[draw, left of=verifier]
      (prover) {$\prover$} ; \node[draw, right of=verifier] (adv) {$\adv^f$};

      \draw (redux.north)+(left:1cm) -- (prover); \draw (redux.north) --
      (verifier); \draw (redux.north)+(right:1cm) -- (adv); \draw (redux) --
      (challenger);

      \node[above of=verifier] (g1) {};
      \node[above of=adv] (g2) {};
      \begin{pgfonlayer}{background}
        \node[draw, dashed, rounded corners=3pt, fill=blue!25,
        fit=(verifier) (adv) (g1) (g2)] (sim) {};
      \end{pgfonlayer}
      \node[below] at (sim.north) {$\simulator$};
    \end{tikzpicture}
  \end{minipage}
  \begin{minipage}{0.69\linewidth}
    \caption{Visualization of the proof of Theorem~\ref{thm:qfs-imposs}. The
      black-box reduction $\mathcal{B}^{(\cdot,\cdot,\cdot)}$ wins the game
      $\mathcal{G}=(\Gamma,c)$ if $(\prover,\verifier^f)$ forms a \spr{}
      $\Sigma^f$ and $\adv^f$ breaks the soundness of $\Pi_\fsl[\Sigma^f]$.
      Since $\simulator=(\simulator_{\verifier},\simulator_{\adv})$ jointly
      simulates $\verifier^f$ and $\adv^f$, neither $\mathcal{B}$ nor $\Gamma$
      can distinguish if $\simulator$ or $(\verifier^f,\adv^f)$ is being used.
      Since $\simulator$ is efficient, this means
      $\mathcal{B}^{(\prover,\simulator_{\verifier},\simulator_{\adv})}$ is an
      efficient machine that wins game $\mathcal{G}$.}
    \label{fig:qfs-redux}
  \end{minipage}
\end{figure}
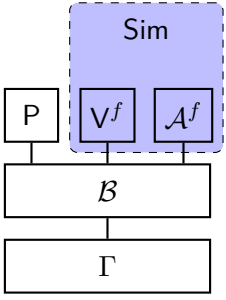

\begin{theorem}\label{thm:qfs-imposs}
Let $\mathcal{G}=(\Gamma,c)$ be a cryptographic game assumption, let $n,m$ be
such that $n-m\in\omega(\lg n)$ and let $\Pi_\fsl^{n,m}$ be a Fiat-Shamir
instantiation in the \crqss{} model. There does not exist a black-box reduction
$\mathcal{B}^{(\cdot,\cdot,\cdot)}$ showing the $\Sigma_{n,m}$--universality of
$\Pi_\fsl^{n,m}$ from the security of game $\mathcal{G}$, unless assumption
$\mathcal{G}$ is false.
\end{theorem}
\begin{proof}
  Assume there exists a black-box reduction $\mathcal{B}^{(\cdot,\cdot,\cdot)}$
  showing the $(n,m)$--universality of $\Pi_\fsl^{n,m}$ from the security of
  game $\mathcal{G}$. We will show that game $\mathcal{G}$ is insecure.
  
  We begin by constructing a family of \sprs{} that has bad challenge function
  $f$ for any function $f\in\mathcal{F}^*$ where $\mathcal{F}^*$ is the set of
  functions for which the operators $P^f_{a,w}$ defined in
  Theorem~\ref{thm:attaque-simulable} form a POVM with $P^f_\bot$.
  The \spr{} $\Sigma^f$ defined below is an interactive proof of language membership for the empty language.  On public input $x$,
    \begin{enumerate}
    \item \prover: does nothing.
    \item $\verifier^f$: interact with a potentially malicious prover in the
      following way.
      \begin{enumerate}
      \item On first message $a\in\bool^n$, pick $c\in_R\bool^m$ uniformly at
        random and send $c$ to the prover.
      \item On response $z$ from the prover, accept iff $c=f(a)$.
      \end{enumerate}
    \end{enumerate}
  This is indeed a \spr{} as it satisfies perfect correctness and special
  soundness.

  Next, we build a dishonest prover that breaks the soundness of the QFS
  transform $\Pi_\fsl^{n,m}[\Sigma^f]$ of this \spr{}. Recall that an
  implementation of Fiat-Shamir is also an implementation of \wotro{}, but that
  might not be secure. Let $\{\adv_\wotro^f\}_f$ be the attack from
  Theorem~\ref{thm:simulation-wotro} against $\Pi_\fsl$ (seen as an
  implementation of \wotro{}), such that there exists a $\negl[n-m]$--simulatable such
  that $\adv^f_\wotro$ produces $(a,f(a),v)$ that $\verifier_\fsl$ accepts with
  probability $1-\negl[n-m]$. Let $\simulator_\wotro$ be the simulator for
  $\{\adv_\wotro^f\}_f$. For a function $f\in\mathcal{F}^*$, define the
  adversarial prover $\mathcal{P}^f$ that attacks protocol
  $\Pi_\fsl^{n,m}[\Sigma^f]$ as follows:
  \begin{enumerate}
  \item Invoke $\adv_\wotro^f$ on register $P$ of the CRQS to obtain $(a,c,v)$.
  \item Send $a,c,v$ and $z=\bot$ to the verifier.
  \item Recall that the verifier for $\Pi_\fsl^{n,m}[\Sigma^f]$ runs
    $\verifier_\fsl$ of $\Pi_\fsl$ with message $(a,c,v)$ on register $V$ of the
    CRQS and then runs $\verifier^f$ of $\Sigma^f$ on input $(a,c,v,z)$.
  \end{enumerate}
  The probability that the verifier accepts in protocol
  $\Pi_\fsl^{n,m}[\Sigma^f]$ is equal to the probability that $\verifier_\fsl$
  accepts and that $c=f(a)$, which by construction of $\adv_\wotro^f$ happens
  with probability at least $1-\negl[n-m]$.

  Plugging \prover, $\verifier^f$ and $\mathcal{P}^f$ into the reduction
  $\mathcal{B}^{(\cdot,\cdot,\cdot)}$ gives an algorithm
  $\mathcal{B}^{\prover,\verifier^f,\mathcal{P}^f}$ that breaks the security of
  game $\mathcal{G}$, and yet that is not efficient. Using the simulator
  $\simulator_\wotro$ for the adversary $\adv^f_\wotro$ allows us to replace the
  inefficient malicious prover $\mathcal{P}^f$ against the QFS transform with an
  indistinguishable efficient simulator, but $\verifier^f$ is still not
  efficiently computable.
  
  We now show how $\mathcal{P}^f$ and $\verifier^f$ can be jointly simulated
  (Definition~\ref{def:sim-sim}) using the stateless simulator $\simulator_\wotro$ for
  $\{\adv_\wotro^f\}_f$. The two stateful algorithms $\simulator_{\mathcal{P}}$
  and $\simulator_\verifier$ are defined as follows
  \begin{enumerate}
  \item Common State: a partial function $f_A:\bool^n\rightarrow\bool^m$ defined
    on an initially empty set $A=\emptyset$.
  \item $\simulator_{\mathcal{P}}$: when invoked on a quantum register $P$, call
    the simulator $\simulator_\wotro$ for the family of adversaries
    $\{\adv^f_\wotro\}_{f\in\mathcal{F}^*}$. Let $(a,c,v)\leftarrow
    \simulator_\wotro$, set $A\leftarrow A\cup \{a\}$ and $f_A(a)=c$, and return
    $(a,c,v,\bot)$. If $\simulator_\wotro$ produces an $a$ that is already in
    $A$, the simulation fails.
  \item $\simulator_\verifier$: when invoked on classical message $(a,c,v,z)$
    and quantum register $V$, run $\verifier_\fsl$ on register $V$ of the CRQS
    with input $(a,c,v)$. If $a\notin A$, pick $x\in_R\bool^m$ uniformly at
    random, set $A\leftarrow A\cup \{a\}$ and $f_A(a)=x$. Output $\tt reject$ if
    $\verifier_\fsl$ rejects or if $c\neq f_A(a)$, otherwise output $\tt
    accept$.
  \end{enumerate}
  \begin{claim}
    The pair of stateful (with common state) algorithms
    $(\simulator_{\mathcal{P}},\simulator_\verifier)$ jointly simulates
    $\{(\mathcal{P}^f,\verifier^f)\}_{f\in\mathcal{F}^*}$.
  \end{claim}
  \begin{proof}
    Let $M^{(\cdot,\cdot)}$ be an oracle-access machine and let $q=\poly[n]$ be
    an upper-bound on the number of queries made by $M$ to either of its
    oracles. We first bound the probability that the simulation fails and then
    condition on the simulation succeeding. Let $\alpha$ denote the random
    variable of the value $a$ produced by $\simulator_\wotro$. Since $\alpha$ is
    uniformly distributed (by the definition of $\simulator_\wotro$ in the proof
    of Theorem~\ref{thm:simulation-wotro}), on any given query, the probability that
    $\simulator_\wotro$ produces $a$ that is already in the set $A$ is
    upper-bounded by
    \begin{equation*}
      \Pr[\alpha\in A]= \sum_{a\in A} \Pr[\alpha=a]\leq q\cdot 2^{-n}\enspace .
    \end{equation*}
    A union bound over the $q$ queries allows us to upper-bound the probability
    that any of the queries returns an $a$ that was already in $A$ by $q^2\cdot
    2^{-n}$ which is $\negl[n]$.

    Conditionned on the event that $\simulator_\wotro$ never produces $a\in A$,
    we show that black-box query access to
    $(\simulator_{\mathcal{P}},\simulator_\verifier)$ is indistinguishable on
    average over $f\in\mathcal{F}^*$ from black-box query access to
    $(\mathcal{P}^f,\verifier^f)$. First, observe that
    $\simulator_{\mathcal{P}}$ behaves exactly as $\mathcal{P}^f$, except that
    it invokes $\simulator_\wotro$ instead of $\adv_\wotro^f$. Therefore the
    BB-indistinguishability of $\simulator_{\mathcal{P}}$ and $\mathcal{P}^f$
    follows from that of $\simulator_\wotro$ and $\adv^f_\wotro$. Second, we
    note that $\simulator_\verifier$ picks each new point of the partial
    function $f_A$ uniformly at random, so that $f_A$ is identically distributed
    to a random function $f$ restricted to $A$. Since a uniformly random
    function $f$ is in $\mathcal{F}^*$ with probability at least $1-\negl$, we
    have that $\simulator_\verifier$ is indistinguishable from $\verifier^f$ on
    average over $f\in\mathcal{F}^*$. Finally, since we condition on the event
    $\alpha \notin A$ at every call of $\simulator_\wotro$, the answers of
    $\simulator_{\mathcal{P}}$ and $\simulator_\verifier$ are always consistent
    with the same function $f$ (i.e.\ the simulation doesn't fail).

    Therefore, the probability that $M^{(\cdot,\cdot)}$ distinguishes
    $(\mathcal{P}^f,\verifier^f)$ from
    $(\simulator_{\mathcal{P}},\simulator_\verifier)$ is at most the probability
    that $\simulator_\wotro$ and $\adv_\wotro^f$ can be distinguished plus the
    probability that the simulation fails, which sum to at most $\negl[n-m]$.
    \qed
  \end{proof}

  We are now ready to conclude the proof. Given the reduction
  $\mathcal{B}^{(\cdot,\cdot,\cdot)}$ we construct an efficient algorithm for
  winning game $\mathcal{G}$ as follows. The machine
  $\mathcal{B}^{(\prover,\simulator_{\verifier},\simulator_{\mathcal{P}})}$
  either:
  \begin{enumerate}
  \item wins game $\mathcal{G}$, or
  \item if it does not, allows to distinguish
    $(\simulator_\verifier,\simulator_{\mathcal{P}})$ from
    $\{(\verifier^f,\mathcal{P}^f)\}_f$.
  \end{enumerate}
  Since we have established the black-box indistinguishability of
  $(\simulator_\verifier,\simulator_{\mathcal{P}})$ and
  $\{(\verifier^f,\mathcal{P}^f)\}_{f\in\mathcal{F}^*}$, we conclude that a
  BB-reduction $\mathcal{B}^{(\cdot,\cdot,\cdot)}$ from the QFS-universality of
  $\Pi_\wotro^{n,m}$ to game $\mathcal{G}$ would allow to win the game.
  \qed
\end{proof}

\section{A Quantum Assumption Allowing for $\wotro^{n,m}$}
\label{sec:impl-wotro-when-m-equal}
 In \cite{barak_lower_2006}, Barak, Lindell, and Vadhan introduce a computational 
 assumption allowing for 
\siguni\ Fiat-Shamir in the \crs\ model. 
It assumes the existence of a family of \emph{entropy preserving}
hash functions. 
In \cite{drv12}, Dodis, Ristenpart, and Vadhan showed that a family 
of entropy preserving hash functions is necessary for a \siguni\ 
implementation of Fiat-Shamir in the \crs\ model. Of course, 
it follows from \cite{barak_lower_2006, bitansky_why_2013} that this assumption 
cannot be black-box reduced to any 
cryptographic game.
In this section, we define a different computational assumption  
allowing for $\wotro^{n,m}$ in the \crqs\ model (and therefore allowing for
\siguni\ Fiat-Shamir). Our assumption is a quantum assumption on 
 hash functions called a \emph{collision-shelter}. 
  We first show in Section~\ref{wotronn} how to construct $\wotro^{n,n}$
with unconditional security  in the \crqs\ model.  In Section~\ref{fshelt},
we define the collision-shelter assumption  and we show 
how to use it to convert $\wotro^{n,n}$ into a computationally 
secure $\wotro^{n,m}$ as long as $m\in\Omega(n)$. 
We conclude in Section~\ref{disccs} by a short discussion about 
some relations and distinctions between  collision-shelters
and collision resistant families of hash functions.

\subsection{Unconditionally Secure $\wotro^{n,n}$ in the $\crqs$ Model}\label{wotronn}

Let us get back to the implementation of  $\wotro^{n,n}$ roughly described in the introduction.
The  result stated in Theorem~\ref{thm:wotro-security-mequaln} requires
the set of MUB to be the one 
introduced by Wootters and Fields in \cite{WF89}.
These bases are for the tensor product of $n$  Hilbert spaces, each
 of odd prime dimension $p$. Let $\Gamma=\{0,\ldots,p-1\}$ denote
 the elements of the finite field $\mathbb{F}_p$ for $p\geq 3$ prime.
 We refer to the Wootters and Fields MUB  for  $\Gamma^n$ as
$\Theta_{\textsf{WF}}^{p,n} = \{{\theta_a}\}_{a\in\Gamma^n}$ where 
$\theta_a = \{\ket{x_a}\}_{x\in\Gamma^n}$ is an orthonormal basis for $\Gamma^n$
that, by virtue of mutual unbiasedness, 
satisfies $|\bracket{x_a}{x'_{a'}}|=p^{-\frac{n}{2}}$ when $a\neq a'$.
The formal 
definition of $\Theta_{\textsf{WF}}^{p,n}$ can be seen in Appendix~\ref{sec:proof-of-m-equals-n}.
The \crqs\ we use to implement 
$\wotro^{n,n}$
is composed of $3n$ $p$--dimensional EPR pairs, each denoted by
$\ket{\epr{}{\Gamma}}_{PV} :=  \frac{1}{\sqrt{p}}\sum_{j\in\Gamma}\ket{jj}_{PV}$. The \crqs\
is then set to 
$\ket{\epr{3n}{\Gamma}}_{PV}:=\ket{\epr{}{\Gamma}}^{\otimes 3n}$.
Henceforth, we denote by $\wotro^{n,n}_\Gamma$ the primitive $\wotro^{n,n}$
where both the  input  and the output  are in $\Gamma^n$.

Before giving our protocol $\pwotro{n,n}=(\prover',\verifier')$ for 
 $\wotro^{n,n}_{\Gamma}$, we first consider a simpler (but insecure)
 version of it where the $\crqs$ is  $\ket{\epr{n}{\Gamma}}_{PV}$ rather
 than  $\ket{\epr{3n}{\Gamma}}_{PV}$. Upon input $a\in\Gamma^n$,
 the simpler scheme asks the prover to measure register $P$ of the \crqs\
 in basis $\theta_a\in \Theta_{\textsf{WF}}^{p,n}$ to obtain outcome 
 $c\in\Gamma^n$. The prover then announces $(a,c)$ to the verifier 
 who verifies that when measuring register $V$ of the \crqs\ the outcome
 $c$ is obtained. If the test is perform with success then the output of the
 primitive is set to $c$. 

This simple protocol cannot be proven secure as it stands. Instead, 
$\pwotro{n,n}$ asks $\prover'$ to 
measure 3 batches 
of EPR pairs $\ket{\epr{n}{\Gamma}}_{PV}$ in the same basis $\theta_a$ to get
outcomes $x_1,x_2,x_3\in\Gamma^n$. The challenge produced by the primitive is 
then  $c=x_3(x_1+x_2)^{-1}$ (where the  operations are done in $\mathbb{F}_{p^n}$). 
This choice for determining $c$  
follows from our proof technique.  $\prover'$ announces $(a,x_1,x_2,x_3)$
that is checked by $\verifier'$ after measuring register $V$ for 
each of the three instances of  $\ket{\epr{n}{\Gamma}}_{PV}$
in basis $\theta_a$. If the test is successful then the output of the primitive is 
set to $c$.
  
\begin{blockquote}
  \rule{\linewidth}{1pt}
  {\bf Protocol} $\pwotro{n,n}$ for $\wotro^{n,n}_\Gamma$\\
  {\bf Setup: }  A CRQ\$ $\ket{\epr{3n}{\Gamma}}_{PV}$. 
  \vspace{1em}

  {\bf Prover: } On input $a\in\Gamma^n$,
  \begin{enumerate}
  \item Measures its part of $\ket{\epr{3n}{\Gamma}}$ in basis
    $\theta_a^{\otimes 3}$, let
    $x=(x_1,x_2,x_3)\in \Gamma^{3n}$ be the result.
  \item If $x_1+ x_{2}=0$, set $c=0$. Otherwise, output $
    c:=x_3(x_1+ x_{2})^{-1}$  and sends $(a,x)$ to verifier.
  \end{enumerate}

  \vspace{1em}

  {\bf Verifier: } Upon reception of $(a,x)$,
  \begin{enumerate}
  \item Measure its part of $\ket{\epr{3n}{\Gamma}}$ in basis $\theta_a^{\otimes 3}$,
    let $x'=(x_1',x_2',x_3')\in \Gamma^{3n}$ be the result. 
  \item Output ${\tt reject}$ if $x\neq x'$ and
  output $(a,c')$ where $c'=x'_3(x'_1+ x'_{2})^{-1}$ otherwise.
  \end{enumerate}

  \rule{\linewidth}{1pt}
\end{blockquote}

Next theorem establishes that  $\pwotro{n,n}$ is $\frac{1}{4}$--secure
against all adversaries. The proof is given in Appendix~\ref{sec:proof-of-m-equals-n}
and may be of independent interest. It consists in showing that the best measurement
to distinguish the state transmitted by a quantum source 
that selects a basis $a\in_R \Gamma^n$ 
at random and sends $\ket{x(a)_a}$ for any set $\{(a,x(a))\}_{a\in\Gamma^n}$ 
cannot be recognized with probability better than $\frac{3}{4}$. Wootters and 
Fields'  MUBs are useful here as this probability is given by a Weil sum that can be
bounded by Deligne's resolution of Weil third conjecture\footnote{Weil's third conjecture is analogue to the Riemann hypothesis over finite fields and is called as such.}~\cite{d74}.

\begin{theorem}\label{thm:wotro-security-mequaln}
  Let $\Gamma=\{0,\ldots,p-1\}$ be the set of
  elements in finite field $\mathbb{F}_p$ for $p\geq 3$ a prime number. Protocol
  $\pwotro{n,n}$, presented above, is a  statistically correct and statistically
  $(\frac 14-\negl)$--secure implementation of $\wotro^{n,n}_{\Gamma}$.
\end{theorem}

We use a  set of mutually unbiased bases (MUBs)
introduced by Wootters and Fields in~\cite{WF89}.
These bases of dimension $p^n$ are for $n$ instances of $p$--level 
quantum mechanical systems with $p\geq 3$ prime. 
The construction is as follows:
\begin{defi}[Mutually Unbiased Bases of~\cite{WF89}]\label{def:mubs}
  Let $p\geq 3$ be  prime. Define the set of mutually unbiased bases
  $\Theta[\mathbb{F}_{p^n}]=\{\theta_a\}_{a\in \mathbb{F}_{p^n}}$ for a Hilbert
  space of dimension $p^n$ where $\theta_a=\{\ket u_a\}_{u\in\mathbb{F}_{p^n}}$
  is composed of vectors $\ket u_a$ expressed in the computational basis as
  \begin{equation}
    \label{eq:mub}
    \ket u_a= {p^{-\frac n2}} \sum_{x\in \mathbb{F}_{p^n}} \exp\left(\frac{2\pi i}p \cdot \trace{a x^2 + ux}\right)\ket x \enspace,
  \end{equation}
  where $\tr:\mathbb{F}_{p^n}\rightarrow \mathbb{F}_p$ denotes the \emph{field trace} $\trace{x} := x+x^p+x^{p^2}+\dots+x^{p^{n-1}}$.
\end{defi}
Notice that Klappenecker and R\"{o}tteler in~\cite{klappenecker_constructions_2004} 
have shown a very similar construction 
for the case $p=2$ (mutually unbiased bases of qubits). Unfortunately, 
our results do not apply to this construction as Weil sums need 
a field of odd characteristics.

\begin{proof}[of Theorem~\ref{thm:wotro-security-mequaln}]
For correctness, observe that if both parties are honest, their
measurement triplets $X$ and $X'$ will be uniformly distributed and perfectly
correlated unless $X_1+X_2=0$. Since $X_1+X_2$ is a random element of $\Gamma^{n}$ due to it
being the result of the measurement of EPR pairs, it holds that this event occurs
with probability at most $|\Gamma|^{-n}$, which is negligible in $n$.

% DEBUT SECURITE
Now onto security. Let $\sff{c}:\Gamma^n\rightarrow\Gamma^n$ be an arbitrary
target function. In order to cheat, i.e. to bias the output challenge towards
$\sff{c}(a)$, a dishonest prover must produce a basis selected by $a$ (the
commitment) and measurement outcome $x_1,x_2,x_3$ such that
\begin{enumerate}
\item $x_3(x_1+ x_{2})^{-1}= \sff{c}(a)$ and
\item \verifier{} obtains the same outcomes $x_1, x_2, x_3$ when he
  measures his part of $\ket{\epr{3n}{\Gamma}}$ in basis $\theta^{\otimes 3}_a$.
\end{enumerate}
We say that $x$ is a \emph{bad} outcome if $x_3(x_1+x_2)^{-1}=
\sff{c}(a)$. Let $\Bg(a)\subseteq \Gamma^{3n}$ denote the set of bad
outcomes for commitment $a$. Note that $|\Bg(a)|=p^{2n}$ for any $a\in\Gamma^n$.

The most general strategy for the prover is to apply a POVM
$\{M_{a,x}\}_{a\in\Gamma^n, x\in \Gamma^{3n}}$ to its part of the EPR pairs to
determine its message to $\verifier$. The probability that \prover{} can bias
the output towards $\sff{c}(a)$ when $\verifier$ accepts is then the probability
that it can produce a commitment (i.e. a basis) such that a bad outcome will be
observed by \verifier{} in that basis. 
\begin{align}
  P_w &= \Pr[X\in \mathcal{B}(A)]\\
      &=  \sum_{a \in \Gamma^n,  x\in \Bg(a)} \trace{\left(
        \proj{x}_a \otimes M_{a,x}\right)
        \cdot \proj{\epr{3n}{\Gamma}}}\label{eq:Pg-m-equal-n}\\
      &= \frac 1{p^{3n}} \sum_{a \in \Gamma^n,x\in \Bg(a)} 
        \trace{ M_{a,x}\proj{x}_a }\enspace. \label{eq:prob-triche1}
\end{align}
To simplify our computations, we have slightly abused notation by writing $\ket
x_a:= \ket{x_1}_a\otimes \ket{x_2}_a\otimes\ket{x_3}_a$ when $x\in\Gamma^{3n}$
and $x_1,x_2,x_3\in\Gamma^n$. Using this notation, for $x,y\in\Gamma^{3n}$ we
have $|\bra{x}_a\ket{y}_b|^2=p^{-3n}$ whenever $a\neq b$.

The optimal cheating strategy for $\prover$ can be framed as the solution to the
following semidefinite program (SDP):
\begin{maxi}
  {\{ M_{a,x} \}}{  \frac 1{p^{3n}} \sum_{a\in\Gamma^n}\sum_{x\in\Bg(a)} \trace{M_{a,x} \proj x_a}}
  {}{}
  \addConstraint{\sum_{a\in\Gamma^n}\sum_{x\in\Bg(a)} M_{a,x}}{\leq \id}\enspace.
\end{maxi}
The dual of this SDP is:
\begin{mini}
  {Z \geq 0}{\frac{1}{p^{3n}}\trace{Z}}
  {}{}
  \addConstraint{\forall a\in \Gamma^n,x\in\Bg(a)\quad}{\proj{x}_a}{\leq Z\enspace.}
\end{mini}
By the duality of semidefinite programming, a feasible solution to the dual
will yield an upper-bound on the optimal solution of the primal. We now show
how to construct a feasible solution that has constant value for $p^{-3n}
\trace{Z}$.

% BORNE PAR FONCTION MONOTONE
Let $S = \sum_{a\in\Gamma^n,x\in\Bg(a)} \proj{x}_a$ and define $f_\alpha(x)=\frac x{\alpha+x}$ for $\alpha\in\reals$. Since $f_\alpha$ is
an operator monotone function (meaning that $A \leq B \Rightarrow f_\alpha(A)
\leq f_\alpha(B)$ for $A,B$ positive semidefinite)\footnote{We can see this by noting that $\frac{x}{\alpha + a} = \alpha\left( \frac{1}{\alpha} - \frac{1}{\alpha + x} \right)$ and that $x \mapsto -x^{-1}$ is operator monotone, see e.g.~\cite[Lemma 2.7]{c09}}, we have that $\frac
1{\alpha+1}\proj{x}_a\leq f_\alpha(S)$ for any $0 < \alpha \leq 1$. The operator  $Z = (\alpha+1)f_\alpha(S)$
 is thus a feasible solution to the dual with associated value
$\frac{\alpha+1}{p^{3n}}\trace{f_\alpha(S)}$.

We now proceed to upper-bound this probability. Since $f_\alpha$ is difficult to deal with directly, we will bound it using Taylor's theorem, yielding powers of $Z$ that will then be easier to compute. To get a good bound, we will use a third degree Taylor bound for $f_\alpha$ centered around $\lambda\in \reals$:
  \begin{equation*}
    f_\alpha(x)\leq \frac{\lambda}{\alpha + \lambda} + \frac{\alpha}{(\alpha+\lambda)^2} (x-\lambda) - \frac{\alpha}{(\alpha+\lambda)^3}(x-\lambda)^2 + \frac{\alpha}{(\alpha+\lambda)^4}(x-\lambda)^3\enspace. 
  \end{equation*}
  Using the Taylor approximation defined above,
  \begin{align}
    \frac {1}{p^{3n}}\trace{Z}&\leq \frac{\alpha+1}{p^{3n}}\trace{f_\alpha(S)} \nonumber\\
       &\leq \frac{\alpha+1}{p^{3n}}
         \bigg(
         \frac{\lambda}{\alpha + \lambda}\trace{\id} + \frac{\alpha}{(\alpha+\lambda)^2} \trace{S-\lambda \id}\nonumber\\
    &\quad\quad\quad\quad\quad\quad- \frac{\alpha}{(\alpha+\lambda)^3}\trace{(S-\lambda \id)^2} + \frac{\alpha}{(\alpha+\lambda)^4}\trace{(S-\lambda \id)^3}
         \bigg)\enspace. \label{eq:taylor-pg}
  \end{align}
  We can rewrite the above traces in the powers of $S-\lambda \id$ in the
  following way.
  \begin{align}
  \left.
  \begin{aligned}
    \trace{\id}&= p^{3n}\enspace,\\
    \trace{S-\lambda \id}&= \trace{S} - \lambda p^{3n}\enspace,\\
    \trace{(S-\lambda \id)^2}&= \trace{S^2}-2\lambda \trace{S}+\lambda^2p^{3n} \enspace,\\
    \trace{(S-\lambda \id)^3}&= \trace{S^3}-3\lambda\trace{S^2}+3\lambda^2 \trace{S}-\lambda^3p^{3n}\enspace.
  \end{aligned}
  \right\}\label{transfo}
  \end{align}

  We refer to Lemmas~\ref{lem:tr-s}, \ref{lem:tr-s2} and~\ref{lem:tr-s3} in
  Appendix~\ref{sec:proof-of-m-equals-n} for the proofs that the following
  relations hold:
  \begin{equation*}
     \trace{S}=p^{3n},\,\trace{S^2}= 2\cdot p^{3n} - p^{2n}
    \text{ and } \tr(S^3)\leq 4p^{3n} + p^{2n}\enspace .
  \end{equation*}
  Choosing to center the Taylor
  approximation around $\lambda=1$ gives the following bounds for~(\ref{transfo}):
\begin{align*}
    \trace{\id}&= p^{3n}\enspace, \\
  \trace{S-\lambda \id}&= 0\enspace,\\
\trace{(S-\lambda \id)^2}&=  2p^{3n} -p^{2n} -2p^{3n}+p^{3n}\\ 
                &= p^{3n}-p^{2n}\enspace, \text{ and} \\
\trace{(S-\lambda \id)^3}&\leq 4p^{3n}+p^{2n}-3( 2p^{3n}-p^{2n}) +3p^{3n}-p^{3n}  \\
                &= 4 p^{2n}\enspace.
  \end{align*}
  Substituting these values into (\ref{eq:taylor-pg}), we get 
  \begin{align}\label{eq:borne-en-ts2-ts3}
    P_w\leq\frac 1{p^{3n}}\trace{Z}
        &\leq \frac{\alpha+1}{p^{3n}}
         \bigg(
          \frac{  p^{3n}}{\alpha + 1} - \frac{\alpha(p^{3n}-p^{2n})}{(\alpha+1)^3} + \frac{\alpha\cdot 4p^{2n}}{(\alpha+1)^4}
         \bigg)\nonumber \enspace.
\end{align}
Looking only at the non-negligible terms, we have
\begin{equation*}
  P_w\leq 1-\frac \alpha{(\alpha+1)^2}+\negl
\end{equation*}
which is minimized at $\alpha=1$ with value $P_w\leq \frac 34+\negl$. Since this probability is the same for all functions $\sff{c}(\cdot)$, it follows that the protocol $(\frac 14-\negl)$--avoids all functions. 

\qed

\end{proof}

\subsection{Collision-Shelters}\label{fshelt}

We are now ready to define a quantum computational assumption that allows for a
secure implementation of $\wotro^{n,m}$ for $m<n$. A collision-shelter for
security parameter $n$, is is a family $\collshn =
\{G^n_s:\Gamma^n\times\Gamma^n \rightarrow \Gamma^{m}\}_{s\in\{0,1\}^{\ell(n)}}$
of hash functions that exhibits a strong quantum flavour of collision
resistance. Intuitively, $\collshn$ is a \emph{collision-shelter} if, for any
function $\mathsf{c}:\Gamma^n\rightarrow \Gamma^m$, no QPT adversary can produce
a state \emph{close} to
\begin{equation}\label{no-shelter}
\ket{\psi_s} = \sum_a \alpha_a \ket{a}_A \otimes \sum_{c:G^n_s(a,c)=\mathsf{c}(a)}\gamma^a_c
\ket{c}_C\otimes\ket{\varphi(a,c)}_{W'}  \enspace, 
\end{equation} 
for $s\in_R\{0,1\}^{\ell(n)}$ and in average over  
outcome $a$ when register $A$ is measured in the 
computational basis, 
$ \sum_{c:G^n_s(a,c)=\mathsf{c}(a)}\gamma^a_c\ket{c}_C\otimes\ket{\varphi(a,c)}_{W'}$
\emph{contains} collisions in superposition. 
Notice that no such state can be produced efficiently 
when the number of possible $a$ is in $O(\lg{n})$ 
and $G^n_s$ is collision resistant, as the generation of $2$ such
states would provide a collision for $G^n_s$ with good probability.

\begin{defi}[$\delta$--Colliding States]
Let $\mathsf{c}:\Gamma^n\rightarrow \Gamma^m$
be  arbitrary and let $G^n_s \in\Gamma^{n}\times\Gamma^n\rightarrow \Gamma^m$. Let
\[\ket{\psi} =\sum_a \alpha_a \ket{a}_A \otimes \sum_{c:G^n_s(a,c)=\mathsf{c}(a)}
\gamma^a_c \ket{c}_C \otimes \ket{\varphi(a,c)}_{W'} 
\]
 be a state hitting function $\mathsf{c}(\cdot)$.  
 Let ${c}^*(a)$ be such that $|\gamma^{a}_{c^*(a)}|^2=
\max_c{\{|\gamma^a_c|^2\}}$ for every $a\in\{0,1\}^n$ and let
$\ket{\tilde{\psi}^*}=\sum_a \alpha_a \ket{a}_A\otimes 
\gamma^a_{{c}^*(a)}\ket{{c}^*(a)}_X\otimes\ket{\varphi(a,c^*(a))}_W$
be the corresponding sub-normalized state obtained from $\ket{\psi}$. 
If $\|\ket{\tilde{\psi}^*}\|^2<1-\delta$ then $\ket{\psi}$
is said to be  \emph{$\delta$--colliding to $\mathsf{c}(\cdot)$ under $G^n_s$}. 
\end{defi}
A collision-shelter is a family of hash functions 
(efficiently samplable and efficiently evaluable)  
that prevents any QPT adversary 
from generating a 
$\delta$--colliding
state hitting any function $\mathsf{c}(\cdot)$. 
\begin{defi}[Collision Shelter]\label{fshelter}
  The efficiently samplable and efficiently evaluable hash function family
  $\collshn = \{G^n_s:\Gamma^n\times\Gamma^n \rightarrow
  \Gamma^{m}\}_{s\in\{0,1\}^{\ell(n)}}$ is a \emph{collision-shelter} if, for
  all $\delta>0$, all functions $\mathsf{c}:\Gamma^n\rightarrow \Gamma^m$, and
  all QPT adversaries $\adv=\{\adv_n\}$, the probability over
  $s\in_R\{0,1\}^{\ell(n)}$ that $\ket{\psi}_{ACW'}\leftarrow \adv_n(s)$ is
  $\delta$--colliding to $\mathsf{c}(\cdot)$ under $\collshn$ is negligible in
  $n$. The \emph{collision-shelter assumption} simply posits the existence of a
  collision-shelter $\collshn$ for $m\leq (1-\alpha)n$ with $0<\alpha<1$.
\end{defi}

We now consider the obvious implementation
of $\wotro^{n,m}_\Gamma$  using $\pwotro{n,n}$ and a function-shelter
 $\collshn$ that simply sets the  
challenge  $\hat{c}\in\Gamma^m$ 
as $\hat{c}=G^n_s(a,c)$ where $c\in\Gamma^n$ is the challenge
produced in $\pwotro{n,n}$ and $s$ is the \crrs. 
Let us denote this implementation of  $\wotro^{n,m}_\Gamma$ 
by  $\pwotro{n,m}[\collshn]$. %=(\prover'',\verifier'')$.
The following theorem
is an easy consequence of Definition~\ref{fshelter} and Theorem~\ref{thm:wotro-security-mequaln}.
\begin{theorem}\label{thm:wotro-col-shelter}
Assuming that $\collshn=\{G^n_s\}_s$ is a collision-shelter with $\Gamma$
a set of elements in finite field $\mathbb{F}_p$ with $p\geq 3$  prime. Then,
  $\pwotro{n,m}[\collshn]$ is a computationally
  $\left(\frac 14- o(1)\right)$--secure implementation of $\wotro^{n,m}_{\Gamma}$.
\end{theorem}
\begin{proof}
  Let $\adv_n$ be a purified adversary against $\pwotro{n,m}[\collshn]$, i.e.\
  whose actions are described by a unitary transform up to the point where it
  needs to send a classical message to the verifier.
Let $s\in\{0,1\}^{\ell(n)}$ be the \crrs\ of an execution of 
$\pwotro{n,m}[\collshn]$ where the adversary $\adv_n$ hits function
$\mathsf{c}:\{0,1\}^n \rightarrow \{0,1\}^m$ with probability $p_{\textnormal{hit}}$.
It follows that  $\ket{{\psi}^{\adv_n}_s}$ -- the joint state of $\adv_n$ and
$\verifier$  generated from the \crqs\ by
$\adv_n$ after running $\pwotro{n,n}$ in
$\pwotro{n,m}[\collshn]$ (but before measuring) -- 
can be assumed w.l.g. to be of the following form
\[
\ket{{\psi}^{\adv_n}_s}_{ACXWV} =  \sqrt{p_{\textnormal{hit}}}\ket{\psi_s}_{ACXWV} +
\sqrt{1-p_{\textnormal{hit}}} \ket{\boxtimes}_{ACXWV}\enspace,
\]
 where $\ket{\boxtimes}_{ACXWV}$ results either in \verifier's rejection or acceptance
 without hitting target $\mathsf{c}(\cdot)$ and
 $\ket{\psi_s}_{AXWV}$ hits function $\mathsf{c}(\cdot)$,   
 \begin{align*}
\ket{{\psi}_s}_{ACXWV}  &= \sum_{a\in\{0,1\}^n} \alpha_a \ket{a}_{A} 
\hspace{-0.15in}  \sum_{x: G_s(a,x_3(x_1+x_2)^{-1})=\mathsf{c}(a)}\hspace{-0,3in} \beta^a_x 
\ket{x_3(x_1+x_2)^{-1}}_C \ket{x}_X  \ket{\varphi(a,x)}_W
\ket{\textnormal{v}(a,x)}_V  \\
&= \sum_{a\in\{0,1\}^n} \alpha_a \ket{a}_{A} \hspace{-0.15in} 
\sum_{\stackrel{\scriptstyle{c: G_s(a,c)=\mathsf{c}(a)}}{{x:x_3(x_1+x_2)^{-1}=c}}}
\hspace{-0,15in} \beta^a_x \ket{c}_C
\ket{{x}}_X  \ket{\varphi(a,x)}_W
\ket{\textnormal{v}(a,x)}_V  \\
&= \sum_{a\in\{0,1\}^n} \alpha_a \ket{a}_{A} \hspace{-0.15in} 
\sum_{c: G_s(a,c)=\mathsf{c}(a)}
 \gamma^a_c \ket{c}_C \hspace{-0.15in} 
\sum_{x:x_3(x_1+x_2)^{-1}=c} \hspace{-0.2in}  \hat\beta^a_{c,x}
%\ket{{x}}_X  
\ket{\hat{\varphi}(a,x)}_{XWV} %\ket{\textnormal{v}(a,x)}_V
 \enspace.
\end{align*} 
Registers $A$ and $X$ contain the final results $a\in\Gamma^n$ 
and $x\in\Gamma^{3n}$ for the execution of $\pwotro{n,n}$
in $\pwotro{n,m}[\collshn]$ and $C$ contains the 
final challenge $c=x_3(x_1+x_2)^{-1}$,  $W$ is  $\adv_n$'s working register,
and $V$ is the register for the  state of the verifier $\ket{\textnormal{v}(a,x)}$ 
when $(a,x)$ is announced. We assume that when $\verifier$
measures $\ket{\textnormal{v}(a,x)}$,  it always accepts $(a,x)$ and therefore
sets the final challenge as $c=x_3(x_1+x_2)^{-1}$. 
Notice that if $\adv_n$ produces state $\ket{{\psi}^{\adv_n}_s}$ efficiently
then it can also produce $\ket{\psi_s}$ efficiently  as long as
$p_{\text{hit}}$ is polynomial since projecting $\ket{{\psi}^{\adv_n}_s}$
into the subspace of states hitting $\mathsf{c}(\cdot)$ can be implemented
efficiently.
Let
\[
\ket{\tilde{\psi}^*_s}_{ACXWV} = \sum_{a} \alpha_a \ket{a}_A\otimes \gamma^a_{\mathsf{c}^*(a)} 
\ket{\mathsf{c}^*(a)}_C \otimes \hspace{-0.1in}
\sum_{x:x_3(x_1+x_2)^{-1}=\mathsf{c}^*(a)}\hspace{-0.1in} \beta^a_{\mathsf{c}^*(a),x}
\ket{\hat{\varphi}(a,x)}_{XWV}
\]
be defined so that $\textsf{c}^*(a)\in\Gamma^{n}$
 is such that $|\gamma^a_{\textsf{c}^*(a)}|^2$ is maximum for each $a\in\Gamma^n$.
 By definition~\ref{fshelter},
 $\collshn=\{G^n_s\}_s$  being a collision-shelter implies  that 
\begin{equation}\label{inprod2}
\left|\bracket{{\psi}_s}{\tilde{\psi}^*_s}\right|^2 \geq 1-o(1) \enspace.
\end{equation}
$\adv_n$'s strategy is almost  the same as the one where $\ket{\psi_s}$ is
replaced by 
$\ket{\tilde{\psi}^*_s}$. 
By theorem~\ref{thm:wotro-security-mequaln}, the (subnormalized) state 
\[
\ket{\hat{\psi}^{\adv_n}_s}  = \sqrt{p_{\textnormal{hit}}} \ket{\tilde{\psi}^*_s} +
\sqrt{1-p_{\textnormal{hit}}}\ket{\boxtimes}
\]
allows to hit function $\mathsf{c}^*(\cdot)$ in $\pwotro{n,n}$
with probability $p_{\textnormal{hit}}\leq \frac{3}{4}+\negl$. From (\ref{inprod2}),
we therefore have that $\ket{{\psi}^{\adv_n}_s}$ hits target function $\mathsf{c}(\cdot)$ 
with probability no larger than $\frac{3}{4}+o(1)+\negl$. The result follows.
\qed
\end{proof}

\subsection{Is the Collision-Shelter Assumption Realistic?}\label{disccs}

While  
$h:\Gamma^{\ell(n)}\times\Gamma^n \rightarrow \Gamma^m$ is \emph{entropy-preserving} 
if no efficient adversary can, given the first argument $\mathbf{s}\in_R \Gamma^{\ell(n)}$ 
picked uniformly at
random, find $x\in\Gamma^n$ 
such that $h(\mathbf{s}, x)$ has almost no entropy (when $\mathbf{s}$ 
has been forgotten), collision-shelters prevent
efficient quantum adversaries from preparing a state 
with entropy in the second argument when the output of the hash function
applied to both arguments is fixed to a function of its first argument.
Why would it be possible for collision-shelters to exist?

Suppose that for all $a\in\Gamma^n$, the hash function $G^n_s(a,\cdot)$ is
collision-resistant against quantum adversaries. Let $\mathsf{c}(a)\in \Gamma^m$
be arbitrary. It follows that for any $a \in \Gamma^n$, no efficient quantum
adversary can produce a state of the form
$\ket{\psi_a}=\sum_{x:G^n_s(a,x)=\mathsf{c}(a)} \beta^a_x \ket{x}_X
\otimes\ket{\varphi(a,x)}$ where
$|\bracket{\psi_a}{\tilde{\psi}^*_a}|^2<1-\delta$ since two states of that form
would allow to find a collision with non-negligible probability. This, of
course, does not imply that $\{G^n_s(\cdot, \cdot)\}_s$ is a collision-shelter
as $G^n_s(a,\cdot)=G^n_s(a',\cdot)=:h_s(\cdot)$ for all $a,a'\in \Gamma^n$ is
such that $G^n_s(a,\cdot)$ is collision-resistant when $h_s(\cdot)$ is
collision-resistant but the following easy-to-generate state is
$o(1)$--colliding to $\mathsf{c}(a)=a_1\ldots a_m$ when $m<n$. We start with a
uniform superposition over $x$ over $\ket x\ket{h_s(x)}$, which ``fixes'' the
first $m$ values of $a$, and introduce a superposition over the $n-m$ remaining
values $a_{m+1}\dots a_n$:
\begin{align*}  p^{-n/2} \sum_x \ket{x}_X \otimes\ket{h_s(x)} &\mapsto 
  p^{-n+\frac{m}{2}} \sum_a \sum_{x:h_s(x)=a_1\ldots a_m}\ket{x}_X \otimes \ket{a}_A \\
  &= p^{-n/2}\sum_a \ket{a}_A \otimes  p^{\frac{-n+m}{2}}\sum_{x: G^n_s(a,x)=\mathsf{c}(a)}
  \ket{x}_X \enspace.
\end{align*}

Such an attack seems  difficult to conduct 
when $\{G^n_s(a,\cdot)\}_a$ is  a set of collision resistant hash functions
 that \emph{appear independent} of each other as far as collisions are concerned. 
What it means exactly for  hash functions in $\{G^n_s(a,\cdot)\}_a$ to \emph{appear independent}
is unclear. It is easy to see that a random 
oracle 
$\mathcal{O}^{n,m}_\Gamma:\Gamma^n \times\Gamma^n \rightarrow \Gamma^m$
acts as a collision-shelter, in fact, it is a much stronger assumption as
$\pwotro{n,m}[\mathcal{O}^{n,m}_\Gamma]$ is already statistically secure.

\begin{theorem}\label{th:romshelter}
Let $\mathcal{O}^{n,m}_\Gamma:\Gamma^n \times\Gamma^n \rightarrow \Gamma^m$
be a random oracle accepting quantum queries.
Then, $\mathcal{O}^{n,m}_\Gamma$ is a collision-shelter and moreover,
$\pwotro{n,m}[\mathcal{O}^{n,m}_\Gamma]$ is $(1-\negl)$--secure.
\end{theorem}
\begin{proof}
As before, we set $p :=|\Gamma|$.
Let $\mathsf{c}:\Gamma^n\rightarrow\Gamma^m$ be an arbitrary 
target function and let 
$L_{\mathsf{c}} :=\{(a,c)\in \Gamma^n\times \Gamma^n \,|\, \mathcal{O}^{n,m}_\Gamma(a,c) = \mathsf{c}(a)\}$. 
In order to succeed, an adversary must announce $(a,c)\in L_{\mathsf{c}}$.
Suppose for a contradiction that $\adv_n$ hits $\mathsf{c}(\cdot)$ with non-negligible
probability $\delta(n):=\frac{1}{q(n)}$ for a positive polynomial $q(n)$.
Remember that $\adv_n$  wins  in  $\pwotro{n,m}[\mathcal{O}^{n,m}_\Gamma]$ 
when it can produce 
$(a,c)\in L_{\mathsf{c}}$. It corresponds to searching 
one element of $L_{\mathsf{c}}$ in a random database
containing $N=p^{2n}$ elements. Let $t=|L_{\mathsf{c}}|$ 
be the number of good elements for the adversary.
It is straightforward to see that 
 except with negligible probability, 
$t\leq N/p^{(1-\epsilon)m}$.
In \cite{BdW1998Lower}, it is shown that searching in a database 
of size $N$ for a solution when there are no more than $t$ solutions
using $T$ queries has a probability of error $p_e$ that satisfies
\begin{equation}\label{pelb}
p_e \geq \exp{\left(-4bT^2/(N-t) - 8T\sqrt{tN/(N-t)^2}\right)} \enspace,
\end{equation} 
for $b>0$ some fixed constant.
Since $e^{-x}\geq 1-x$, we then get  that
for $T\in \poly$, (\ref{pelb}) satisfies
\[
p_e \geq 1-\negl[n] \enspace.
\]
\qed
\end{proof}
The random oracle makes it difficult to produce a single 
$(a,c)\in L_{\mathsf{c}}$ while a collisions-shelter posits that
it is difficult to generate a collision on $\textsf{c}(a)$ for many $a\in\Gamma^n$,
which at least  requires finding $(a,c)\in L_{\mathsf{c}}$.

Notice that while any secure universal Fiat-Shamir transform in the \crs\
model requires the existence of an entropy-preserving family of hash functions\cite{drv12},
this does not  seem to be the case for collisions-shelters with respect to 
\wotro\ in the \crqss\ model.

\section{Black-Box Impossibility of a Flavour of Quantum Lightning}
\label{sec:black-box-imposs-lightning}
In this section, we show that a secure \wotro\ can be constructed 
from a quantum lightning scheme that satisfies a
slightly stronger security notion. 
Quantum lightning was introduced by Zhandry in \cite{zhandry_quantum_2019} as a primitive
allowing for  publicly verifiable quantum money schemes and provable randomness among others.

\subsection{Typed Quantum Lightning.}

Quantum lightning provides some fresh randomness that even an adversarial procedure cannot bias towards a certain value.  We present a strenghtened version of this property that requires that this randomness remains in the presence of an input to the lightning generation procedure. This notion is sufficiently strong to provide a secure \wotro{} protocol. 

\begin{defi}\label{def:typed-quant-lightn}
  A \emph{typed quantum lightning scheme} is a tuple of $\qpt$ algorithms $(\tqlsetup,\tqlgen,\tqlver)$ where
  \begin{itemize}
  \item $\tqlsetup(1^n)$ produces a storm $\thunder$.
  \item $\tqlgen(\thunder, a)$ takes an additional parameter  $a\in\bool^n$, and produces a lightning state $\ket{\lightn_a}$.
  \item $\tqlver(\thunder,\ket\lightn)$ returns the type $a$, a serial number $s$
    or $\bot$ if the state is not valid, and a leftover quantum register.
  \end{itemize}
  Correctness is defined similarly to regular \ql{}: serial numbers are deterministic for honestly generated bolts and verification does not noticeably affect the bolt.
  The security properties of a \tql{} scheme are as follows:
  For any $\qpt$ adversary $\adv$ that on input $\thunder$ produces a type
  $A\in\bool^n$ and a state $\ket{\lightn}$, if we let $\rho_{QSA'}=\tqlver(\thunder,\ket{\lightn})$, then
    \begin{equation*}
      \Pr\left[H_\infty(S\mid A \wedge (S \neq \bot)\wedge (A=A'))\leq \lg p(n)\right] \leq \negl
      \enspace.
    \end{equation*}
\end{defi}

Based on Definition~\ref{def:typed-quant-lightn}, typed quantum lightning
provides randomness in the serial number conditionned on the type. 
It is the ability of the adversary to choose the type $a$ that makes this primitive stronger
than regular \ql{}.
 A natural
\wotro{} protocol in the \crs{}+\crqs{} model based on this new primitive is presented
below.
\begin{blockquote}
  \rule{\linewidth}{1pt}
  {\bf Protocol} $\Pi_{\mathsf{WRO}}^{\tql}$ for $\wotro^{n,m}$\\
  {\bf Setup: }  A \crs{} containing $\thunder\leftarrow\tqlsetup(1^n)$ for a
  \tql{} scheme with $n$--bit types and $m$--bit serial numbers. A
  \crqs{} containing $\ket{\epr{}{}}^{\otimes q}$ where $q$ is the qubit size of a \tql{} state.
  \begin{enumerate}
  \item On input $a\in\bool^n$, \prover{} calls $\ket{\lightn_a}\leftarrow
    \tqlgen(\thunder,a)$, sets $\rho_{QSA}=\tqlver(\thunder,\ket{\lightn_a})$,
    teleports register $Q$ to \verifier{} using the EPR pairs and sends $(A,S)$ to
    \verifier{}.
  \item Upon reception of $(a,s,\rho_Q)$, \verifier{} calls $\sigma_{Q'S'A'}\leftarrow
    \tqlver(\thunder,\rho_Q)$ and tests that $A'=a$ and $S'=s$. \verifier{} aborts if
    the tests failed, otherwise \verifier{} sets $c=s$ and outputs $(a,c)$.
  \end{enumerate}
  \rule{\linewidth}{1pt}
\end{blockquote}

\begin{theorem}\label{thm:lightning}
  The above protocol is a secure instantiation of $\wotro^{n,m}$.
\end{theorem}
The proof is a direct consequence of the security of the \tql{} primitive.
\begin{corollary}\label{cor:lightning}
  There is no black-box reduction from the security of a \tql{} scheme with type
  length $n$ and serial length $m$ satisfying $n-m\in\omega(\lg n)$ to the
  security of a cryptographic game assumption, unless the assumption is false.
\end{corollary}

\subsection{Justification for the \tql{} assumption.}
\label{sec:tql-justif}

Why is typed quantum lightning a realistic assumption? 
It turns out that the \tql{} primitive can be built from ``vanilla'' \ql{} for
types of length $O(\lg n)$.  We present a construction of a \tql{}
scheme for $\lg p(n)$ bits types for any polynomial $p(\cdot)$ from an
arbitrary (regular) \ql{} scheme. 

\begin{blockquote}
  \rule{\linewidth}{1pt}
  {\bf Prerequisite: }
  A \ql{} scheme $(\qlsetup,\qlgen,\qlver)$. A family of \mbox{$n\cdot p(n)$}--wise
independent hash
functions ${\cal H}\subset\bool^n\rightarrow \bool^{\lg p(n)}$. 
  \begin{itemize}
  \item $\tqlsetup(1^n)$: Let $\thunder\leftarrow\qlsetup(1^n)$ and $h\sample
    {\cal H}$, output $\thunder'=(\thunder,h)$
  \item $\tqlgen(\thunder',a):$ Parse $\thunder'$ as $(\thunder,h)$. Do $\ket\lightn \leftarrow \qlgen(\thunder)$ until
    $s=\qlver(\thunder,\ket\lightn)$ satisfies $h(s)=a$ and output
    $\ket\lightn$.
  \item $\tqlver(\thunder',\ket\lightn)$: Parse $\thunder'$ as $(\thunder,h)$. Compute $\rho_{SQ}\leftarrow
    \qlver(\thunder,\ket\lightn)$ and set $A=h(S)$. Output $\rho_{ASQ}$.
  \end{itemize}

  \rule{\linewidth}{1pt}
\end{blockquote}
  \begin{theorem}\label{thm:tql}
    $(\tqlsetup,\tqlgen,\tqlver)$ is a
    \tql{} scheme of $\lg p(n)$ bits types. 
  \end{theorem}
  \begin{proof}
    Correctness follows from that of the underlying QL scheme: a state produced
    by $\tqlgen$ will be recognized as a valid state by $\tqlver$ if $\qlgen$
    produces valid states.

    The expected running time of $\tqlgen$ is exponential in $\lg p(n)$ and
    thus polynomial in $n$. Since $h$ is sampled from a family of
    $n\cdot p(n)$--pairwise independent hash functions, the probability that $\tqlgen$ does not produce an
    output after $n\cdot p(n)$ steps is at most
    \begin{equation*}
      \Pr[h(s_1)\neq a \wedge\dots\wedge h(s_{n\cdot p(n)})\neq a] = \left( 1-\frac 1{p(n)} \right)^{n\cdot p(n)}\leq e^{-n}
    \end{equation*}

    %%% version min-entropie
    For security (Definition~\ref{def:typed-quant-lightn}), let $\cal A$ be an attacker
    against the min-entropy of the \tql{} scheme, i.e.\ $\cal A$ produces
    with inverse polynomial probability a state $\ket\lightn$ such that $\rho_{ASQ}\leftarrow\tqlver(\thunder,\ket\lightn)$
    has logarithmic  min-entropy in $S$ conditioned on $A$:
    \begin{equation}
      \label{eq:6}
      \Pr\Big[H_\infty(S\mid A \wedge (S \neq \bot))\leq \lg n^r\Big] \geq \frac 1{n^k}
    \end{equation}
    for some $r,k>0$. We construct an adversary $\cal B$ against the uniqueness
    of the original lightning scheme from this $\cal A$. The strategy of $\cal
    B$ is as follows: call $\adv(1^n)$ twice to obtain $\ket{\lightn_1}$ and
    $\ket{\lightn_2}$, if $\qlver(\ket{\lightn_1})=\qlver(\ket{\lightn_2})$ halt and output
    $\ket{\lightn_1}$ and $\ket{\lightn_2}$, otherwise repeat.
    We now show that this strategy will produce a collision for the underlying
    \ql{} scheme with an expected polynomial number of calls to $\adv$.

    Let $\bar a$ be such that $H_\infty(S\mid (A=\bar a) \wedge (S \neq
    \bot))\leq \lg n^r$ and such that $\Pr[A=\bar a]\geq\frac 1{q(n)}$ for some
    polynomial $q(\cdot)$ when $\rho_{ASQ}$ is obtained from
    $\tqlver(\adv(\thunder))$. Note that since $a$ is $\lg p(n)$ in length,
    such an $\bar a$ must exist for~\eqref{eq:6} to hold (otherwise all $a$
    that have low conditional min-entropy have negligible probability of being
    produced by $\adv$). Then for each pair of invocations of $\adv$, the
    following holds with probability at least $\frac 1{n^k}$:
    \begin{align*}
      &\Pr[\qlver(\sket{\lightn_1})=\qlver(\sket{\lightn_2})]\\
      &\geq \frac 1{q(n)^2}\Pr[\qlver(\sket{\lightn_1})=\qlver(\sket{\lightn_2})\mid A_1=\bar a \wedge A_2=\bar a]\\
      &\geq \frac 1{q(n)^2}2^{-H_2(\qlver(\rho)\mid A=\bar a)}\\
      &\geq \frac 1{q(n)^2}2^{-H_\infty(\qlver(\rho)\mid A=\bar a)}\\
      &\geq \frac 1{q(n)^2} \frac 1{n^r}
    \end{align*}
    where $H_2$ denotes the collision entropy and is upper-bounded by the
    min-entropy $H_\infty$. The probability that $\cal B$ halts and succeeds is
    therefore at least $(q(n)^2n^{r\cdot k})^{-1}$. \qed
  \end{proof}

\printbibliography{}

\appendix

\section{Technical Lemmas for Theorem~\ref{thm:simulation-wotro}}
\label{sec:proof-simulable}

\begin{proof}[Proof of Lemma~\ref{tech4}]
Applying the Chernoff bound with $D=2^k$, $M=2^n$, $\alpha=2^{-m}$,
and $X_a^f = \sum_{w} N^a_{f(a),w}$ results in 
\begin{align*}
\Pr_{f\in \mathcal{F}}{\left[2^{-n}\sum_{a} X_a^f  \notin \left[{(1-t\eta)}\frac{\id_P}{2^{m}}, {(1+t\eta)}\frac{\id_P}{2^{m}}\right]\right]}
& \leq 2^{k+1} \exp{\left(-\frac{1}{2\ln 2}2^{n-m}t^2 \eta^2\right)} \\
&=  2^{k+1} \exp{\left(-(n+ k) t^2 \right)} \\
&\leq 2^{-n\cdot t^2} \enspace.
\end{align*}
The result then follows easily assuming $\frac 1{2^n} \sum_{a} X_a^f \leq (1+t\eta)\frac{\id_P}{2^{m}}$,
\begin{align*}
\sum_{a,w} P^f_{a,w} &= \frac{\sum_{a,w} N^a_{f(a),w}}{2^{n-m}(1+\eta)} \\
&=  \frac{2^{-n}\sum_{a} X^f_a }{2^{-m}(1+\eta)} \\
&\leq  \frac{(1+t\eta)\frac{\id_P}{2^m}}{2^{-m}(1+\eta)} \\
&= \frac{(1+t\eta)\id_P}{1+\eta} \enspace.
\end{align*}
On the other hand, assuming  $\frac 1{2^n} \sum_{a} X_a^f \geq (1-t\eta)\frac{\id_P}{2^{m}}$,
\begin{align*}
\sum_{a,w} P^f_{a,w} &= \frac{\sum_{a,w} N^a_{f(a),w}}{2^{n-m}(1+\eta)} \\
&=  \frac{2^{-n}\sum_{a} X^f_a }{2^{-m}(1+\eta)} \\
&\geq  \frac{(1-t\eta)\frac{\id_P}{2^m}}{2^{-m}(1+\eta)} \\
&= \frac{(1-t\eta)\id_P}{1+\eta} \enspace.
\end{align*}
\qed
\end{proof}

\begin{proof}[Proof of Lemma~\ref{tech3}]
As we did before, we set the dimension of the \crqss\ on $\prover$'s
side to be $2^k$ (which we should write $k(n)$ rather than $k$).
Let 
\[\Delta:=\left\|\esper{f\in \mathcal{F}^*}\left [{\cal R}_n^{\mathcal{A}^{f}_n}(\proj 0)\right]- 
\esper{f\in\mathcal{F}}\left[{\cal R}_{n}^{\adv^{f}_n}(\proj 0)\right]\right\|_1\enspace.
\]
Let $\rho_{\mathcal{F}^*}=\frac{1}{\#\mathcal{F}^*}\sum_{f\in\mathcal{F}^*}
\proj{f}$ and $\rho_{\mathcal{F}}=\frac{1}{\#\mathcal{F}}\sum_{f\in\mathcal{F}}
\proj{f}$. It is easy to see that 
\begin{align*}  
\left\| \rho_{\mathcal{F}^*} - \rho_{\mathcal{F}} \right\|_1 &=
\sum_{f\in\mathcal{F}^*} \left( \frac{1}{|\mathcal{F}^*|} - \frac{1}{|\mathcal{F}|}\right)
+ \sum_{f\in \mathcal{F}-\mathcal{F}^*} \frac{1}{|\mathcal{F}|} \\
&= {1} - {\Pr{[f\in\mathcal{F}^*]}}
+ \Pr{\left[f\notin \mathcal{F}^*\right]} \\
&\leq \negl[n] \enspace,
\end{align*}
after applying lemma~\ref{thm:attaque-simulable}.
Then, we have
\begin{align}
\Delta 
&\leq \left\|   \rho_{\mathcal{F}^*} - \rho_{\mathcal{F}}   \right\|_1  + 
\frac{1}{\#\mathcal{F}}\sum_{f\in\mathcal{F}-\mathcal{F}^*}\left\|
{\cal R}_n^{\mathcal{A}^{f}_n}(\proj 0) \right\|_1 \nonumber \\
&\leq \negl[n] + \frac{1}{\#\mathcal{F}}\sum_{f\in\mathcal{F}-\mathcal{F}^*}\left\|
{\cal R}_n^{\mathcal{A}^{f}_n}(\proj 0) \right\|_1\label{uneformesimple} \enspace.
\end{align}
Bounding this sum is not as straightforward as it might look at first glance: while we know that the sum has very few terms, we have no guarantee that $\mathcal{A}_n^f$ is a physically realizable map when $f \in \mathcal{F} - \mathcal{F}^*$, and hence we cannot trivially bound these norms by 1. Instead, let us consider reduction $\mathcal{R}_n^{\adv^f_n}:=
  U_q (\adv^f_n\otimes \id) U_{q-1} \ldots(\adv^f_n\otimes \id)U_1(\adv^f_n\otimes \id)U_0$
  where all $U_j, j\in\{0,\ldots,q \}$ are unitaries and the $\id$'s are acting on the wires
  that are not part of $\adv^f_n$'s standard interface. 
Let $\ket{\varphi^f_0} := U_0 \ket{0}$ be a normalized state and let
$\ket{\varphi^f_j} := U_j (\adv^f_n\otimes \id)\ket{\varphi^f_{j-1}}$  for $1<j\leq q$, 
not necessarily of norm 1 when $f\notin \mathcal{F}^*$.
Using lemma~\ref{tech4} with $t :=2^{\frac{n-m}{4}}$ (and 
$\eta=\sqrt{2\ln 2(n+k)\cdot 2^{m-n}}$ as in the statement of lemma~\ref{tech4}), we have that 
\begin{equation}\label{eqqr}
\Pr_{f\in_R \mathcal{F}}{\left[ 
\sum_{a,w} P^f_{a,w} \notin \left[\frac{(1- t\eta)\id_P}{1+\eta},
       \frac{(1+t \eta)}{1+\eta}\id_P \right]\right]} \leq 
2^{-n\sqrt{2^{n-m}}} \enspace.
\end{equation}
For $\ket{\varphi}$ a state of norm 1 and 
for $f\in\mathcal{F}$ such that  $\sum_{a,w} P^f_{a,w} \leq  \left(\frac{1+ t
    \eta}{1+\eta}\right)\id_P\leq (1+t\eta)\id_P$,
$\|(\adv^f_n\otimes \id)\proj{\varphi}(\adv^f_n \otimes \id)^*\|_1\leq 1+t\eta$. Starting with a normalized state $\ket{\varphi}$, after $q(n)$ queries 
to $\adv^f_n$, the square of the norm of the resulting vector is upper bounded by
$(1+ t\eta)^{q(n)}$. Notice that when  $t\eta = p(n) 2^{-\beta n}$ for $p(n)$ a
polynomial and $\beta>0$,
\[
(1+ t\eta)^{q(n)} = \left(1+\frac{p(n)}{2^{\beta n}}\right)^{q(n)} = 
\left(1+\frac{p(n)}{2^{\beta n}}\right)^{\frac{2^{\beta n}q(n)p(n)}{p(n)2^{\beta n}}}
\approx \exp{\left(\frac{q(n)p(n)}{2^{\beta n}}\right)} \enspace,
\]
since 
$\lim_{N\rightarrow \infty}(1+\frac{1}{N})^N=e$. In other words, when $t\leq 2^{\frac{n-m}{4}}$, we have
$t\eta = p(n) 2^{-\beta n}$ and\footnote{Using the fact that
$1+2^{-x+1}>\exp{(2^{-x})}$ for all $x\geq 0$.},
\begin{equation}\label{llp1} 
\lim_{n\rightarrow \infty}
\left\|
{\cal R}_n^{\mathcal{A}^{f}_n}(\proj 0) \right\|_1 =
\exp{\left(\frac{q(n)p(n)}{2^{\beta n}}\right)} \leq
1+\negl[n]\enspace.
\end{equation}
Therefore, for $t\leq 2^{\frac{n-m}{4}}$ and $f$  such that
$\sum_{a,w} P^f_{a,w} \leq  \left(1+ t \eta\right)\id_P$, we have that
 ${\cal R}_n^{\mathcal{A}^{f}_n}(\proj 0)$ \emph{essentially preserves} norms
 like an isometry: $\left\|{\cal R}_n^{\mathcal{A}^{f}_n}(\proj 0) \right\|_1=1$.
For handling the other case ($t>2^{\frac {n-m}4}$), we first define  
 \[
 \mathcal{F}_t := \left\{f\in \mathcal{F}\,\Biggm| \, 
 \frac{(1-t\eta)\id_P}{1+\eta} \leq  \sum_{a,w} P^f_{a,w} %\in \left[\frac{(1+t \eta)}{1+\eta}\id_P,
 \leq 
  \frac{(1+t\eta)\id_P}{1+\eta} %\right]
  \right\} \enspace,
 \]
 and notice that by construction, for every $f\in\mathcal{F}$,
 \begin{equation}\label{legrosmax}
 \sum_{a,w} P^f_{a,w} = \frac{\sum_{a,w} N^a_{f(a),w}}{2^{n-m}(1+\eta)} \leq 
  \frac{\sum_{a} \id_P}{2^{n-m}(1+\eta)} =\frac{2^{n}\id_P}{2^{n-m}(1+\eta)}
  \leq \frac{2^m \id_P}{1+\eta}\enspace.
 \end{equation}
Let $t :={2^{\frac{n-m}{4}}}$ and note  that $t<\frac 1{2\eta}$ so
that the Chernoff bound expressed in lemma~\ref{tech4} can be used. 
We consider two cases for $f\in\mathcal{F}-\mathcal{F}^*$: either $f$ is in
$\mathcal{F}_t - \mathcal{F}^*$, or $f$ is outside $\mathcal{F}_t\cup \mathcal{F}^*$. 
We have,
\begin{align} 
\frac{1}{\#\mathcal{F}}\sum_{f\in \mathcal{F}-\mathcal{F}^*}
  \left\| {\cal R}_n^{\mathcal{A}^{f}_n}(\proj 0) \right\|_1 
&=  \frac{1}{\#\mathcal{F}}\left(\sum_{f\in\mathcal{F}_t-\mathcal{F}^*}
  \left\| {\cal R}_n^{\mathcal{A}^{f}_n}(\proj 0) \right\|_1 \right.\nonumber\\                                                    
&\quad\qquad\quad\quad\quad\quad\quad\quad+ \left.\sum_{f\notin\mathcal{F}_t\cup\mathcal{F}^*}\left\| {\cal R}_n^{\mathcal{A}^{f}_n}(\proj 0) \right\|_1 \right)
\nonumber \\
&\leq \Pr{\left[F\in\mathcal{F}_t-\mathcal{F}^*\right]}
  \cdot(1+t\eta )^{q(n)} \nonumber\\                                                    
&\quad\quad\quad\quad\quad+    \Pr{\left[F\notin\mathcal{F}_t\cup\mathcal{F}^*\right]}\cdot 2^{mq(n)} \label{coton}\\
&\leq \negl[n](1+\negl[n])+ 2^{-n\sqrt{2^{n-m}}+mq(n)} \label{metal}\\
&\leq \negl[n] \label{norm1r}\enspace,
\end{align}
as long as $n>m$,
where (\ref{coton}) follows from (\ref{legrosmax}) and (\ref{metal}) follows
from (\ref{llp1}) and the Chernoff bound, as stated in (\ref{eqqr}).

Finally, using (\ref{norm1r}) in (\ref{uneformesimple}) proves the statement.
\qed
\end{proof}

\begin{lemma}\label{lem:a-negl}
  (\hyperref[qqqA]{A}) is negligible.
  \end{lemma}
  \begin{proof}
  Since $\simulator_n$ picks $a\in \{0,1\}^n$ uniformly at random, we get
  \begin{align}
    \sum_{\substack{a\in\delta(s) \\ z\in\{0,1\}^m\\ w\in\{0,1\}^{\ell(n)}}}
    q^{\simulator_n}_{a,z,w}(\psi_{j}(s))
    &= \sum_{\substack{a\in\delta(s) \\z\in\{0,1\}^m\\ w\in\{0,1\}^{\ell(n)}}}
    2^{-n}\trace{N^a_{z,w}\proj{\psi_j(s)}}  \nonumber \\
    &=  \sum_{\substack{a\in\delta(s)}}
      2^{-n}\trace{\sum_{z,w}N^a_{z,w}\proj{\psi_j(s)}} \nonumber \\
      % &\leq \negl\enspace. \label{partieaa} \nonumber \\
&\leq  \sum_{\substack{a\in\delta(s)}}
2^{-n} \nonumber \\
&\leq q(n)2^{-n} \nonumber \\
&\leq \negl \enspace.\label{partieaa}
 \end{align}
A similar argument can be applied to 
$q_{a,w}^{\adv^f_n}(\psi_j(s))$ although $\{P^f_{a,w}\}_{a,w}$ is a collection 
of positive operators that do not form 
a valid POVM when $f\notin \mathcal{F}^*$, as $\sum_{a,w} P^f_{a,w}\nleq \id_P$ in this
case.
We have,
\begin{align}
\frac{1}{\#\mathcal{F}}\sum_{\substack{f\in\mathcal{F} \\ a\in\delta(s) \\ z\in\{0,1\}^m\\ w\in\{0,1\}^{\ell(n)}}}
q_{a,w}^{\adv^f_n}(\psi_j(s)) &=\frac{1}{\#\mathcal{F}} \hspace{-0,4cm}
\sum_{\substack{a\in\delta(s) \\ z\in\{0,1\}^m\\ w\in\{0,1\}^{\ell(n)}}}
\sum_{\substack{f\in\mathcal{F} \\ f(a)=z}} \trace{P^f_{a,w}\proj{\psi_j(s)}}\nonumber \\
&=\frac{1}{\#\mathcal{F}}  \hspace{-0,4cm}
\sum_{\substack{a\in\delta(s) \\ z\in\{0,1\}^m\\ w\in\{0,1\}^{\ell(n)}}}
\sum_{\substack{f\in\mathcal{F} \\ f(a)=z}} 
\trace{\frac{N^a_{z,w}\proj{\psi_j(s)}}{2^{n-m}+ \sqrt{2\ln 2(n+k)2^{n-m}}}}\nonumber\\
&= \frac{1}{\#\mathcal{F}}  \hspace{-0,1cm}
\sum_{\substack{a\in\delta(s)}}
\sum_{\substack{f\in\mathcal{F} \\ f(a)=z}} 
\trace{\frac{\sum_{z, w}N^a_{z,w}\proj{\psi_j(s)}}{2^{n-m}+ \sqrt{2\ln 2(n+k)2^{n-m}}}}\nonumber \\
&\leq \frac{1}{\#\mathcal{F}}  \hspace{-0,1cm}
\sum_{\substack{a\in\delta(s)}}
\sum_{\substack{f\in\mathcal{F} \\ f(a)=z}} 
\frac{1}{2^{n-m}+ \sqrt{2 \ln 2(n-m+k)2^{n-m}}} \nonumber \\
&\leq
\left(\sum_{\substack{f\in\mathcal{F} \\ f(a)=z}} \frac{1}{\#\mathcal{F}} \right)
\frac{q(n)}{2^{n-m}+ \sqrt{2\ln 2(n-m+k)2^{n-m}}} \nonumber \\
&\leq \negl[n-m]\enspace. \label{partieab}
\end{align}
 We now conclude,
 \begin{equation}\label{lapartieA}
(\hyperref[qqqA]{A}) \leq (\ref{partieaa}) + (\ref{partieab}) \leq \negl[n-m]\enspace.
\end{equation}
\qed
\end{proof}

 \begin{lemma}\label{lem:bot-negl}
   (\hyperref[patiel1]{$\bot$}) is negligible.
\end{lemma}
\begin{proof}
  This corresponds to the likelihood (and not \emph{the probability}, as
  $\{P^f_{a,w}\}_{a,w}$ is not a valid POVM when $f\notin \mathcal{F}^*$) of an
  error when $\adv^f_n$ is queried once on $\ket{\psi_j(s)}$, a normalized state
  vector. Observe that when $f\in \mathcal{F}_{t}$, $\sum_{a,w} P^f_{a,w}\geq
  (1-t\eta)\id_P$ and therefore $0\leq P^f_{\bot}\leq t\eta\,\id_P$. In other
  words, the probability to get an error when $f\in\mathcal{F}_t$ is upper
  bounded by $t\eta$. On the other hand, when $f\notin \mathcal{F}_{t}$, the
  only thing we can say from our construction is that $0\leq P^f_{\bot}\leq
  \id_P$. In the following and as before, we set $t:=2^{\frac{n-m}{4}}$ and
  $\eta:=\sqrt{2\ln 2(n+k)2^{m-n}}$. Using the operator Chernoff bound expressed
  in lemma~\ref{tech4}, we have
  \begin{align}
    \text{(\hyperref[patiel1]{$\bot$})}
    &\leq \frac{1}{\#\mathcal{F}}\left(\sum_{f\in\mathcal{F}_t} q_{\bot}^{\adv^f_n}(\psi_j(s))
      +\sum_{f\notin\mathcal{F}_t} q_{\bot}^{\adv^f_n}(\psi_j(s))\right)\nonumber \\
    &= \frac{1}{\#\mathcal{F}}\left(\sum_{f\in\mathcal{F}_t}  \trace{P^f_{\bot}\proj{\psi_j(s)}}
      +\sum_{f\notin\mathcal{F}_t}  \trace{P^f_{\bot}\proj{\psi_j(s)}}\right)\nonumber \\
    &\leq t\eta + \Pr{\left[f\notin \mathcal{F}_t \right]} \nonumber \\
    &\leq \sqrt{2\ln 2(n+k)2^{\frac{-n+m}{2}}} + 2^{-n\sqrt{2^{n-m}}} \nonumber \\
    &\leq \negl[n-m]\enspace. \label{petitapetit}
  \end{align}
  \qed
\end{proof}
 
 \begin{lemma}\label{lem:m-negl}
   (\hyperref[pppM]{M}) is negligible.
 \end{lemma}
 \begin{proof}
   This the main part to show that the Chernoff adversary $\adv_n^{\mathcal{F}}$
   is simulatable. This is where we use the fact that $\simulator_n$
   \emph{simulates} the adversary $\adv_n^{f}$ whenever $f\in_R\mathcal{F}$ (in
   real life this is for $f\in_R\mathcal{F}^*$). It follows essentially the same
   steps as in \cite{bitansky_why_2013, bgw12} adapted to deal with our Chernoff
   adversary. For all $f\in \mathcal{F}$ such that $f(a)=y$, we have
   \begin{align}
     (\hyperref[pppM]{M}) &=\sum_{\substack{a\in\{0,1\}^n\setminus \delta(s) \\ w\in\{0,1\}^{\ell(n)} \\y\in\{0,1\}^m}} 
     \frac{2^m}{\#\mathcal{F}} \sum_{\substack{f\in\mathcal{F} \\ f(a)=y}}\hspace{-0.1cm}\left|q^{\simulator_n}_{a,y,w}(\psi_{j}(s))- 2^{-m}q^{\adv^f_n}_{a,w}(\psi_{j}(s))\right|\nonumber \\
                          &=\sum_{\substack{a\in\{0,1\}^n\setminus \delta(s) \\ w\in\{0,1\}^{\ell(n)} \\y\in\{0,1\}^m}} 
     \left(\sum_{\substack{f\in\mathcal{F} \\ f(a)=y}}\hspace{-0,1cm}\frac{2^m}{\#\mathcal{F}}\right)\hspace{-0.1cm}
     \left|\trace{\frac{N^a_{y,w}}{2^n}- \frac{2^{-m}N^a_{y,w}}{\frac {2^n}{2^m} +\sqrt{2\ln 2(n+k)\frac {2^n}{2^m}}}\proj{\psi_j(s)}}\right|\nonumber \\
                          &= \sum_{{a\in \bool^n\setminus\delta(s)}} 
                            \left|\frac{1}{2^n}-\frac{ 2^{-m}}{\frac {2^n}{2^m} + \sqrt{2\ln 2(n+k)\frac {2^n}{2^m}}} \right|
                            \sum_{y,w}\trace{N^a_{y,w}\proj{\psi_j(s)}} \nonumber \\
  &\leq  \sum_{\substack{a\in \bool^n\setminus \delta(s)}} \left|\frac{1}{2^n}-\frac{ 2^{-m}}{\frac {2^n}{2^m} + \sqrt{2\ln 2(n+k)\frac {2^n}{2^m}}} \right|  \nonumber \\
 &\leq \left|1-\frac{2^{-m+n}}{\frac {2^n}{2^m} + \sqrt{2\ln 2(n+k)\frac {2^n}{2^m}}}\right|
 \nonumber \\
 &\leq \sqrt{2\ln 2(n+k)2^{m-n}} \nonumber \\
 &\leq \negl[n-m]  \enspace.\label{unpasdeplus}
   \end{align}
   \qed
\end{proof}

\section{Technical Lemmas for Theorem \ref{thm:wotro-security-mequaln}}
\label{sec:proof-of-m-equals-n}

We now proceed to compute the $\trace{S}$, $\trace{S^2}$ and $\trace{S^3}$
values used in the proof of Theorem~\ref{thm:wotro-security-mequaln}.

\begin{lemma}\label{lem:tr-s}
  $\trace{S}=p^{3n}$.
\end{lemma}
\begin{proof}
  Since $|\Bg(a)|=p^{2n}$,
  \begin{equation*}
\trace{S}= \sum_{a\in\Gamma^n}\sum_{x\in\Bg(a)}\trace{\proj x_a} 
   = p^{3n}\enspace .
 \end{equation*}
 \qed
\end{proof}

\begin{lemma}\label{lem:tr-s2}
  $\trace{S^2} = 2p^{3n}-p^{2n}$.
\end{lemma}
\begin{proof}
  \begin{align*}
    \trace{S^2} &= \sum_{a,b\in \Gamma^n}\sum_{x\in\Bg(a),y\in\Bg(b)} \tr(\proj{x}_a \proj{y}_{b})\\
                &= \sum_{a\in \Gamma^n}\left( \sum_{x\in\Bg(a)} 1 + 
                  \sum_{b \neq a} \sum_{x\in\Bg(a),y\in\Bg(b)}|
                  \bra{x}_a\ket{y}_{b} |^2 \right)\\
                &= \sum_{a\in \Gamma^n}\left( |\Bg(a)| + 
                  p^{-3n} \sum_{b \neq a}  |\Bg(a)|\cdot|\Bg(b)| 
                  \right)\\
                &= \sum_{a\in \Gamma^n}\left( p^{2n} + 
                  p^{-3n} \sum_{b \neq a}  p^{4n}
                  \right)\\
                &= \sum_{a\in \Gamma^n}\left( p^{2n} + 
                  p^{n} \sum_{b \neq a}  1
                  \right)\\
                &= p^n\left( p^{2n} + 
                  p^{n} (p^n-1)
                  \right)\\
                &= 2p^{3n}-p^{2n}
  \end{align*}
   \qed
\end{proof}

Upper-bounding $\trace{S^3}$ will require a little more machinery. We introduce a
theorem of Deligne~\cite{d74} and some of its corollaries before proceeding with
the proof.

% traduction EN: https://arxiv.org/abs/1807.10810
% PDF FR: https://publications.ias.edu/sites/default/files/Number23.pdf
\begin{theorem}[\cite{d74}, Theorem 8.4]
  Let $Q$ be a polynomial of $n$ variables $x_1,\ldots,x_n$ and of degree $d$ on $\mathbb{F}_q$, let $Q_d$ be the homogeneous part of degree $d$ of $Q$ and let $\psi:\fq \rightarrow\complex^*$ 
 be an additive non-trivial character on $\fq$.  Assume that
    \begin{enumerate}
    \item $d$ is coprime with $p$, the characteristic of $\mathbb{F}_q$, and
    \item the hypersurface $H_0$ of $\mathbb{P}_{\mathbb{F}_q}^{n-1}$ defined by $Q_d$ is smooth,
    \end{enumerate}
    then
    \[ \left| \sum_{x_1,\dots,x_n \in \mathbb{F}_q} \psi\left( Q(x_1,\dots,x_n) \right) \right| \leqslant (d-1)^n q^{n/2}\enspace. \]
    \label{thm:deligne}
\end{theorem}
In the above, the second condition boils down to ensuring that there is no point at which the $\frac{\partial Q}{\partial x_i}$ all vanish simultaneously. Here is a version that is closer to what we will need:
\begin{corollary}
  Let $m \leq k$, $A$ a $k \times m$ matrix with rank $m$ in $\mathbb{F}_q$, and
  let $C$ be a $k \times k$ matrix in $\mathbb{F}_q$. Then, if $A^{\intercal} C A$
  is non-singular,
    \[ \left| \sum_{\vec{v},\vec{x} = A \vec{v}}  \psi\left( \vec{x}^{\intercal} C \vec{x} \right) \right| \leqslant q^{m/2}\enspace. \]
    \label{cor:deligne1}
\end{corollary}
In other words, we take the sum over all $(x_1,\dots,x_k)$ that satisfy a system
of $k-m$ independent linear equations.
\begin{proof}
  Let $Q = \vec{x}^{\intercal} C \vec{x} = \vec{v}^{\intercal} A^{\intercal} C A \vec{v}$, and
  observe that
    \[ \frac{\partial Q}{\partial v_i} = e_i^{\intercal} A^{\intercal} C A \vec{v} + \vec{v}^{\intercal} A^{\intercal} C A e_i = 2e_i^{\intercal} A^{\intercal} C A \vec{v}\enspace. \]
     Condition 2 of Theorem~\ref{thm:deligne} is thus equivalent to
    \[ A^{\intercal} C A \vec{v} = 0 \Leftrightarrow \vec{v}=0\enspace, \]
    which amounts to saying that  $A^{\intercal} C A$ is non-singular.
    \qed
\end{proof}

Here is now a version that is more directly relevant to our case. 
\begin{corollary}
    Let  $m \leq k$ and let $B \in \mathbb{F}^{(k-m) \times k}_q$ and $C \in
  \mathbb{F}_q^{k \times k}$ be full rank matrices. Then,
    \[ \left| \sum_{\vec{x}:B\vec{x} = 0}  \psi\left( \vec{x}^{\intercal} C \vec{x} \right) \right| \leqslant q^{m/2}\enspace. \]
    \label{cor:deligne2}
\end{corollary}
\begin{proof}
  Let $B^c \in \mathbb{F}^{m \times k}_q$ such that $M := \begin{bmatrix} B\\
    B^c \end{bmatrix} \in \mathbb{F}_q^{k\times k}$ has full rank. Then
  condition $B \vec{x} = 0$ is equivalent to $\vec{x} = M^{-1} \begin{bmatrix}
    0\\ \vec{v} \end{bmatrix}$ for some $\vec{v} \in \mathbb{F}_q^{m}$. We can
  thus define $P := \begin{bmatrix} 0 \\ \id \end{bmatrix}$ and apply
  corollary~\ref{cor:deligne1} with $A = M^{-1} P$, while observing that $P^\intercal
  {M^{-1}}^{\intercal} C M^{-1} P$ has full rank, since ${M^{-1}}^{\intercal} C M^{-1}$
  also has full rank.\qed
  \end{proof}

  \begin{lemma}\label{lem:tr-s3}
    $\trace{S^3}\leq 4p^{3n}+ p^{2n}$
  \end{lemma}               
  \begin{proof}
    Let's first write out the expression of interest:    
    \begin{align}
      \trace{S^3} &= \sum_{a,b,c\in\Gamma^n}\sum_{x\in\Bg(a)}\sum_{y\in\Bg(b)}\sum_{z\in\Bg(c)}
                    \trace{\proj x_a \proj y_b \proj z_c}\nonumber\\
                  &= \sum_{a=b=c}
                    \sum_{x} 1 + 3\sum_{a\neq b}
      \sum_{x,y} |\bra x_a\ket y_b|^2
      + \sum_{a\neq b\neq c} \sum_{x,y,z} \bra {x}_a\ket {y}_b\bra
      {y}_b\ket {z}_c \bra {z}_c \ket {x}_a \label{eq:s3-third-term} % \\
    \end{align} 
    where the middle term groups the three cases $a\neq b$, $a\neq c$ and $b\neq
    c$ that all have the same value. We know how to upper-bound the first two
    sums using the same techniques as Lemma~\ref{lem:tr-s2}. Most of the proof
    is dedicated to finding an upper-bound to the third term.

    Recall our construction of mutually unbiased bases $\theta_a$ presented in
    Definition~\ref{def:mubs}. For $r\in\mathbb{F}_{p^n}$ and
    $a\in\mathbb{F}_{p^n}$:
    \begin{equation*}
      \ket r_a= {p^{-\frac n2}} \sum_{u\in \mathbb{F}_{p^n}} \exp\left(\frac{2\pi i}p \cdot \tr(a u^2 + ru)\right)\ket u\enspace.
    \end{equation*}
    Extending this basis to $3$ systems through $\theta_a^{\otimes 3}$ yields
    vectors of the form 
    \begin{align*}
      \ket{{x}}_{a} &= p^{-3n/2} \sum_{{u} \in \mathbb{F}_{p^n}^3} \exp\left({\frac{2\pi i}{p}\trace{a {u}^{\intercal} {u} + {x}^{\intercal} {u}}}\right) \ket{{u}}\enspace,
    \end{align*}
    where $x^\intercal$ denotes the transpose of $x\in
    \mathbb{F}_{p^n}^3\simeq \Gamma^{3n}$. Here, we slightly abuse notation by
    writing $\ket{{x}}_a$ for a vector in basis $\theta_a^{\otimes 3}$.

    The inner product of two such vectors is given by the expression
    \[ \bra{{y}}_b\ket{{x}}_a = p^{-3n} \sum_{{u}\in
        \mathbb{F}_{p^n}^3} \exp\left( \frac{2\pi i}{p}\trace{(a-b) {u}^{\intercal}
          {u} + ({x} - {y})^{\intercal} {u}} \right) \enspace.\]
    Combining the three inner products in the expression of
    interest~\eqref{eq:s3-third-term}, we have
    \begin{equation*}
      \bra {x}_a\ket {y}_b\bra {y}_b\ket {z}_c \bra {z}_c \ket {x}_a =
      p^{-9n} \sum_{u,v,w\in
        \mathbb{F}_{p^n}^3} \exp\left(\frac{2\pi i}{p}\trace{
        \begin{array}{l}
          (a-b) {u}^{\intercal}
          {u} + ({x} - {y})^{\intercal} {u}\\
          +(b-c) {v}^{\intercal}
          {v} + ({y} - {z})^{\intercal} {v}\\
          +(c-a) {w}^{\intercal}
          {w} + ({z} - {x})^{\intercal} {w}
        \end{array}
        }
       \right)
    \end{equation*}
    We introduce some notation that will allow us to present the above
    expression in a more compact, albeit more complicated form.
    Let $\sff{c}:\mathbb{F}_{p^n}\rightarrow\mathbb{F}_{p^m}$ and for $a\in\mathbb{F}_{p^n}$, define
    \begin{equation}
      \label{eq:bai}
      B_a = 
      \begin{bmatrix} 1 & 0 & \sff{c}(a)\\
        0 & 1 & \sff{c}(a) \end{bmatrix} \in \mathbb{F}_{p^n}^{2\times 3}
    \end{equation}
    such that for $x_1,x_2\in\mathbb{F}_{p^n}$, the expression
    \begin{equation}\label{eq:baiu}
      [x_1,x_2]\cdot {B_a}^\intercal = [x_1, x_{2}, \sff{c}(a)(x_1+x_{2})]^\intercal\in \mathbb{F}_{p^n}^{3}
    \end{equation}
    is a sequence of measurement outcomes that leads to the bad outcome 
    $\sff{c}(a)$ in the protocol.
    
    For $a,b,c\in \mathbb{F}_{p^n}$, write
    \begin{equation*}
      B_{a,b,c}:=
      \begin{bmatrix} -B_a &     0 & B_{a}\\
        0     & B_{b} & -B_{b}\\
        B_{c}   & -B_{c} & 0\end{bmatrix}\in\mathbb{F}_{p^n}^{6\times 9}\enspace,
      \end{equation*}
      and
    \begin{equation}
      \label{eq:ce}
       C_{a,b,c} := 
      \begin{bmatrix} 
        (c-a)\id_{\mathbb{F}_{p^n}^{3\times 3}} & 0 & 0\\
        0 & (b-c)\id_{\mathbb{F}_{p^n}^{3\times 3}} & 0\\
        0 & 0 & (a-b)\id_{\mathbb{F}_{p^n}^{3\times 3}}
      \end{bmatrix}\in \mathbb{F}_{p^n}^{9\times 9}\enspace.
    \end{equation}
    The previous operators are defined such that
    \begin{align*}
      &\sum_{\substack{x\in\Bg(a) \\y\in\Bg(b)\\ z\in\Bg(c)}}
      \bra {x}_a\ket {y}_b\bra {y}_b\ket {z}_c \bra {z}_c \ket {x}_a \\&\qquad=
      p^{-9n}  \sum_{\varrho\in\mathbb{F}_{p^n}^6}\sum_{\xi\in\mathbb{F}_{p^n}^9} \exp\left( \frac {2\pi i}p
        \trace{\xi^{\intercal}C_{a,b,c}\xi + \varrho^{\intercal}B_{a,b,c}\xi} \right)
    \end{align*}
    with the goal of bounding above the right-hand side using
    Corollary~\ref{cor:deligne2}. The construction of $B_{a,b,c}$ appears more
    complex than necessary because we want it to have a large rank.

    Equipped with the above, we are now ready to upper-bound the third term
    in~\eqref{eq:s3-third-term} with Corollary~\ref{cor:deligne2}.
    \begin{align}
      &\sum_{a\neq b\neq c} \sum_{\substack{x\in\Bg(a)\\ y\in\Bg(b)\\ z\in\Bg(c)}} \bra {x}_a\ket {y}_b\bra
      {y}_b\ket {z}_c \bra {z}_c \ket {x}_a \nonumber\\
      &= p^{-9n}\sum_{a\neq b\neq c} \sum_{\varrho\in\mathbb{F}_{p^n}^{6}} \sum_{\xi\in\mathbb{F}_{p^n}^{9n}}\exp\left({\frac{2\pi i}p \trace{\xi^\intercal \cdot C_{a,b,c}\cdot\xi +\varrho^\intercal\cdot B_{a,b,c}\cdot\xi}}\right)\nonumber\\
      &=  p^{-9n}\sum_{a\neq b\neq c} \sum_{\varrho\in\mathbb{F}_{p^n}^{6}} 
        \sum_{\substack{\xi\in\mathbb{F}_{p^n}^{9n} \\ B_{a,b,c}\cdot \xi=0}}
      \exp\left({\frac{2\pi i}p \trace{\xi^\intercal \cdot C_{a,b,c}\cdot\xi}}\right)\label{eq:sum-roots} \\
      &\leq  p^{-9n}\sum_{a\neq b\neq c}
        \sum_{\varrho\in\mathbb{F}_{p^n}^{6}} p^{2n}\label{eq:application-deligne} \\
      &=  p^{-9n}\sum_{a\neq b\neq c} p^{6n} p^{2n}\nonumber\\
      &= p^{-n}(p^n)(p^n-1)(p^n-2)\nonumber
        \enspace.
    \end{align}
    Equality~\eqref{eq:sum-roots} above follows from the observation that once
    $\xi$ is fixed, if $B_{a,b,c}\cdot\xi$ is non-zero then the sum over
    $\varrho$ will span all $p$th roots of unity in equal proportions which sums
    to 0. % source: https://en.wikipedia.org/wiki/Root_of_unity#Summation
    In more details, letting $\alpha=\xi^\intercal \cdot
    C_{a,b,c}\cdot\xi\in\mathbb{F}_{p^n}$ and
    $0\neq v=B_{a,b,c}\cdot\xi\in\mathbb{F}_{p^n}^6$,
    \begin{align*}
      &\sum_{\varrho\in\mathbb{F}_{p^n}^6}\exp\left({\frac{2\pi i}p \trace{\alpha
            +\varrho^\intercal\cdot v}}\right)
      = p^{5n}\sum_{\beta\in\mathbb{F}_{p^n}}\exp\left({\frac{2\pi i}p \trace{\alpha
            +\beta}}\right)
      % il y a p^{5n} solutions a l'equation \rho^T * v = \beta (5 variables
      % libres et 1 liee)
      \\&\qquad= p^{6n-1}\sum_{\gamma\in\mathbb{F}_{p}}\exp\left({\frac{2\pi i}p \gamma}\right)=0\enspace.
      % il y a p^{n-1} solutions a \tr(\beta)=\gamma
      % source: https://en.wikipedia.org/wiki/Field_trace#Finite_fields
    \end{align*}
    Inequality~\eqref{eq:application-deligne} follows from
    Corollary~\ref{cor:deligne2} by observing that $\rank{B_{a,b,c}}\geq 4$. To
    see this, note that by removing columns 3, 6 and 9 from $B_{a,b,c}$ (those
    corresponding to $\sff{c}(a),\sff{c}(b)$ or $\sff{c}(c)$), we are left with the matrix
    \begin{equation*}
      \begin{bmatrix} \id &  0 & -\id\\
        0 & \id & -\id\\
        \id & -\id & 0\end{bmatrix} \enspace. 
    \end{equation*}
    Taking linear combinations of the above we can obtain
    \begin{equation*}
      \begin{bmatrix} \id & 0 & -\id\\
        0 & \id & -\id\\
        0 & 0 & 0 \end{bmatrix}
    \end{equation*}
    and hence $B_{a,b,c}$ has rank at least that of the above matrix, which is
    equal to $4$ since each of the identities act on $\mathbb{F}_{p^n}^2$.

    We can now complete the proof by taking the expected value over $g$.
    Continuing from~\eqref{eq:s3-third-term},
    \begin{align*}
     \trace{S^3} 
      &= \sum_{a=b=c}
        \sum_{x\in\Bg(a)} 1 + 3\sum_{a\neq b}
        \sum_{\substack{x\in\Bg(a) \\y\in\Bg(b)}} |\bra x_a\ket y_b|^2
        + \sum_{a\neq b\neq c} \sum_{\substack{x\in\Bg(a) \\y\in\Bg(b)\\ z\in\Bg(c)}} \bra {x}_a\ket {y}_b\bra
        {y}_b\ket {z}_c \bra {z}_c \ket {x}_a \\
      &\leq p^{3n} + 3p^{2n}(p^n-1) + (p^n-1)(p^n-2)\leq 4p^{3n}+ p^{2n}\enspace.
    \end{align*}
    \qed

  \end{proof}

\section{Basic Properties of \wotro{}}
\label{sec:basicfactsproofs}

\begin{proposition}
    The 2-message protocol in which $\prover$ sends $a\in\Gamma^n$ directly to $\verifier$, 
    and $\verifier$ then chooses $c\in_R \Gamma^m$ at random, sends it to $\prover$ and always accepts is a 
    correct and $\delta$--secure implementation of $\wotro^{n,m}_{\Gamma}$ for $\delta = 1-\frac{1}{|\Gamma|^m}$ and for any alphabet $\Gamma$ and $n,m \geqslant 1$. 
\end{proposition}
\begin{proof}
  Let $\Pi(c|a)$ denote the conditional distribution of the protocol output. Indeed, correctness is obvious as $a$ and $c$ are correctly distributed with 
    $\Pi(c|a) = \frac{1}{\# \Gamma^m}$. For security, let $A$ be the random variable produced by $\dprover$ and $C$ be the random variable produced by $\verifier$, and let $\sff{c}: \Gamma^n \rightarrow \Gamma^m$ be some function. Then,
    \begin{align*}
        \Pr\left[ V = 1 \wedge C = \sff{c}(A) \right] &= \Pr\left[ C = \sff{c}(A) \right]\\
        &= \frac{1}{|\Gamma|^m} \enspace.
    \end{align*}
    \qed
\end{proof}

\begin{proposition}%\label{prop:no-wotro-adaptive-c}
  Let $\Gamma$ be an arbitrary finite alphabet, let $m,n\geq 1$ and let $0<
  \delta\leq 1$. There is no correct and $\delta$--secure 1-message
  implementation of $\wotro^{n,m}_{\Gamma}$ in the bare model. Moreover, there
  is no such $\delta$--secure non-interactive implementation of
  $\wotro^{n,m}_{\Gamma}$ common random string (resp. random oracle) model if
  the function $\sff{c}$ from Definition~\ref{def:secure-impl} can depend on the
  CR\$ $r$ (resp. the random oracle \oracle).
\end{proposition}
\begin{proof}
  Consider the message sent from the prover \prover{} to the verifier
  \verifier{}. Without loss of generality, it is of the form $(a,c,w)$ where $a$
  is \prover{}'s input, $c$ is the joint output and $w$ is additional
  information for \verifier{} to decide whether to accept or reject. Let
  $\Pi=(\prover,\verifier)$ (resp. $\Pi^r=(\prover^r,\verifier^r)$ and
  $\Pi^\oracle=(\prover^\oracle,\verifier^\oracle)$) be a correct implementation
  of \wotro{} in the bare model (resp. CR\$ model and ROM). Define the first
  message of the prover in each model by
  \begin{align*}
    P(a,s)&:= (a, c(a,s), w(a,s))\tag{bare}\\
    P^r(a,s)&:= (a, c^r(a,s,r), w^r(a,s,r))\tag{CR\$}\\
    P^{\oracle}(a,s)&:= (a, c^\oracle(a,s,v), w^\oracle(a,s,v))\tag{ROM}
  \end{align*}
where $s$ is the random tape of the prover, $r$ is the value of the CR\$ and
$v=(\oracle(a_1), \oracle(a_2), \dots, \oracle(a_{\kappa(n)}))$ where
$a_1,\dots,a_{\kappa(n)}\in\Gamma^n$ are chosen using $s$ for some upper bound
$\kappa(n)$ on the number of oracle queries performed by $\prover{}$ in
$\Pi^\oracle$.

Since the protocol is correct, it must hold that
\begin{equation}
  \Pr[V(P(A,S))=1]=\Pr[V^r(P^r(A,S))=1]= \Pr[V^\oracle(P^\oracle(A,S))=1]=1
\end{equation}
where the probability is taken over the values of $A$ and $S$. Then for each $a$
with non zero probability, there exist a value $s(a)$, $s^r(a)$ and
$s^\oracle(a)$ such that
\begin{equation}
  V(P(a,s(a)))= V^r(P^r(a,s^r(a)))= V^\oracle(P^\oracle(a,s^\oracle(a)))=1 \enspace.
\end{equation}
Define malicious prover $\tilde \prover{}$ (resp. $\tilde \prover{}^r$ and
$\tilde \prover{}^\oracle$) that on input $a$ uses random tape value $s(a)$ (resp.
$s^r(a)$ and $s^\oracle(a)$). Then the protocol $\Pi$ (resp. $\Pi^r$ and
$\Pi^\oracle$) does not avoid the functions $\sff{c}(a):= c(a,s(a))$ (resp.
$\sff{c}^r(a):= c^r(a,s^r(a),r)$ and $\sff{c}^\oracle(a):=
c^\oracle(a,s^\oracle(a),v)$\ ).
\qed
\end{proof}

\begin{proposition}\label{prop:crs-for-big-challenges}
  Let $m>n$. The protocol for $\wotro{}_{\Gamma}^{n,m}$ in the \crrs{} model where both
  parties output the CR\$ $r\in\Gamma^m$ for any $a\in\Gamma^n$ and \verifier{}
  always accepts is correct and $\delta$--secure, for $\delta = 1-|\Gamma|^{n-m}$.
\end{proposition}
\begin{proof}
    Correctness is obvious, and security is easy to prove as well: suppose that $\dprover$ wants to steer the output of the protocol towards some function $\sff{c}$. He must then look at the \crrs\ $r$, and announce an $a$ such that $\sff{c}(a)=r$. Hence, $r$ must happen to be in the image of $\sff{c}$. However, since $\sff{c}$ is a function from $\Gamma^n$ to $\Gamma^m$ and $m>n$, there are at most $|\Gamma|^n$ strings in the image of $\sff{c}$, and the probability that a uniformly chosen $r$ falls into that set is at most $|\Gamma|^{n-m}$.\qed
\end{proof}

\begin{proposition}%\label{prop:no-rom-for-small-challenges}
    Let $\Gamma$ be an arbitrary finite alphabet of size $q\geq 2$. Then, for any $m,n$ with $m\leq n$, there exists no $\exp{(-q^{n-m})}$--secure implementation of $\wotro_{\Gamma}^{n,m}$ in the \rom.
\end{proposition}
\begin{proof}
    We will show that a cheating prover that is unbounded in time can search for an $a$ that will satisfy $\verifier$. Consider a dishonest prover $\dprover$ who uses the following strategy: run the honest prover $\prover$ on all possible inputs $a$ in lexicographic order, and declare victory if it ever outputs $\sff{c}(a)$. We will also assume that function $\sff{c}(\cdot)$ is chosen uniformly at random, and show that the expected winning probability of the cheating prover is at least $1 - \exp(-q^{n-m})$. We have the following:
    \begin{align*}
        \Pr_{\mathcal{O}, \sff{c}}[\text{$\dprover$ loses}] &= \Pr_{\mathcal{O},\sff{c}}\left[ \text{$\dprover$ loses at step $a=0$} \wedge \text{$\dprover$ loses at $a=1$} \wedge \ldots \right]\\
        &= \prod_{a \in \Gamma^n} \Pr_{\mathcal{O},\sff{c}}\left[ \text{$\dprover$ loses at step $a$} \middle| \text{$\dprover$ loses at all steps before $a$}  \right]\\
        &= \prod_{a \in \Gamma^n} \Pr_{\mathcal{O},\sff{c}}\left[ \text{$\prover$ does not output $c(a)$ on input $a$} \middle| \text{$\dprover$ loses at all steps before $a$}  \right]\\
        &= \prod_{a \in \Gamma^n} \frac{q^m - 1}{q^m}\\
        &= \left( 1 - \frac{1}{q^m} \right)^{q^n}\\
        &= \left[ \left( 1 - \frac{1}{q^m} \right)^{q^m} \right]^{q^{n-m}}\\
        &< \exp(-q^{n-m}).
    \end{align*}
    since $\sff{c}(a)$ is chosen uniformly at random for each $a$. Hence, $\dprover$'s winning probability is at least $1 - \exp(-q^{n-m})$ as advertised, and there must exist a choice of function $\sff{c}(\cdot)$ that achieves this bound. \qed
\end{proof}

\begin{proposition}\label{prop:rom-for-big-challenges}
  The protocol for $\wotro{}_{\Gamma}^{n,m}$ in the \rom{} model where both
  parties output the $\mathcal{O}(a)$ for any $a\in\Gamma^n$ and \verifier{}
  always accepts is correct and statistically $\delta$--secure, for $\delta =
  1-|\Gamma|^{n-m}$.
\end{proposition}
  The proof is identical to that of
  Proposition~\ref{prop:crs-for-big-challenges} by considering $r=\mathcal{O}(a)$.

\begin{proposition}%\label{prop:rom-vs-polytime-provers}
  The protocol described in Proposition~\ref{prop:rom-for-big-challenges} is
  $1-\negl$--secure in the \rom{} against polynomial-time provers as long as $m$ is
  at least linear in $n$.
\end{proposition}
\begin{proof}
    Let $\ell(n)$ be a polynomial which bounds the number of oracle queries that $\dprover$ can make. Furthermore, without loss of generality we will assume that $\dprover$ never makes the same oracle call twice. Then, given any function $\sff{c}:\Gamma^n \rightarrow \Gamma^m$, in order to cheat successfully, $\dprover$ must be able to find an $a$ such that $\mathcal{O}(a) = \sff{c}(a)$. 
    
    Now, let $A_1,\cdots,A_{\ell(n)}$ be random variables taking values in
    $\Gamma^n$ where $A_i$ represents the $i$th query to the oracle (if
    $\dprover$ makes fewer than $\ell(n)$ queries, let $A_i$ be any string that
    was not queried so far). These random variables are functions of the oracle
    $\mathcal{O}$, in that they can depend on the results of previous queries. We then
    have by the union bound that
    \begin{align*}
        \Pr_\mathcal{O}\left[ \text{$\dprover$ wins} \right] &\leqslant \Pr\left[ \mathcal{O}(A_1) = \sff{c}(A_1) \vee \mathcal{O}(A_2) = \sff{c}(A_2) \vee \ldots \vee \mathcal{O}(A_{\ell(n)}) = \sff{c}(A_{\ell(n)}) \right]\\
            &\leqslant \sum_{i=1}^{\ell(n)} \Pr\left[ \mathcal{O}(A_i) = \sff{c}(A_i) \right]\\
            &= \ell(n) q^{-m}\\
            &\leqslant \negl.
    \end{align*}
    \qed
\end{proof}
\begin{proposition}%\label{prop:no-wotro-crs-equal-size}
There are one-message implementations  of $\wotro^{n,n}_{\Gamma}$ 
 arbitrarily close to be $\frac{1}{\EE}$--avoiding against unbounded provers in the \crrs\ model.
 \end{proposition}
 \begin{proof}
 Let $\ell(n)$ be the length of the \crrs\ (i.e. $r\in_R \Gamma^{\ell(n)}$) upper bounded by some polynomial.
 Let $\prover^r:\Gamma^n\rightarrow\Gamma^m\times\Gamma^*$ denote  
 \prover's message to \verifier\ upon \crrs\ $r$ and input $a\in \Gamma^n$. For $a\in \Gamma^n$ and 
 \crrs\ $r\in \Gamma^{\ell(n)}$, we have $\prover^r(a)=(c(r,a) ,v(r,a))$ which defines
 announcement $(a,c(r,a),v(r,a))$ to \verifier. 
 The verifier's algorithm $\verifier^r:\Gamma^{n}\times\Gamma^n\times\Gamma^*\rightarrow
\{0,1\}$ upon \crrs\ $r$ accepts $(\alpha,\beta,\gamma)$ when $\verifier^r(\alpha,\beta,\gamma)=1$. 
The prover's algorithm can be considered deterministic given $r$,
all randomness being provided by $r$. 
For $\{1,\ldots, {p^n}\}= \Gamma^n$ an enumeration of all  elements
in $\Gamma^n$,
let
\[ C^r := c(r,1) \| c(r, {2}) \| c(r, 3) \| \ldots \| c(r,{p^n}) 
\]
be the sequence of all challenges announced by $\prover$ upon \crrs\ $r$,
one for each possible input $a\in \Gamma^n$. Let $\mathcal{C}:= \{C^r\}_{r\in \Gamma^{\ell(n)}}$.
For $\omega\in (\Gamma^n)^{p^n}$, we define
\[ H_\omega := \left\{ C \in \mathcal{C} \,|\, \left(\exists  j\in [p^n]\right)\left[C_j = \omega_j\right]\right\}
\]
 as the set of sequences containing challenges  {\em hitting} $\omega$
 somewhere. If $\Pi$ is $\delta$--avoiding then for all $\omega\in (\Gamma^n)^{p^n}$,
 $|H_\omega| \leq  \delta \cdot p^{\ell(n)}$.
 
 We define $\Pi$ and then show it is $\frac{3}{4}$--avoiding 
 using a \crrs\ $r\in (\Gamma^n)^2$.  $\Pi$ is simply defined from 
 $r=r_1 \| r_2\in (\Gamma^n)^2$ as
 \[  C^r =\underbrace{r_1,r_1,\ldots, r_1}_{\text{$\frac{p^n}{2}$ times}}, \underbrace{r_2,r_2,\ldots, r_2}_{\text{$\frac{p^n}{2}$ times}} \enspace.
 \]
 We denote the elements of $\Gamma^n$ by $\{1,2,\ldots, p^n\}$.
Let $\omega^*\in (\Gamma^n)^{p^n}$ be defined as
 \[  \omega^* := 1,2,3,\ldots, \frac{p^n}{2}, 1,2,3,\ldots, \frac{p^n}{2} \enspace. 
 \]
 It is not difficult to see that $\omega^*$ maximizes the probability to be hit by $C^R$. 
 We have,
 \begin{align*}
\prob{C^R\in H_{\omega^*}} &= \prob{\left(R_1 \leq \frac{p^n}{2}\right) \vee \left(R_2 \leq \frac{p^n}{2}\right) } \\
 &= 1- \prob{ \left(R_1 > \frac{p^n}{2}\right) \wedge \left(R_2 > \frac{p^n}{2}\right) } \\
 &= 1-\frac{1}{4}=  \frac{3}{4}\enspace.
 \end{align*}
 By considering longer \crrs\ $r=r_1,r_2,\ldots, r_{\ell(n)}$ where $r_i\in \Gamma^n$, it is possible to get arbitrarily close 
 to a correct $\frac{1}{\EE}$--avoiding scheme with 
 \[C^r= \underbrace{r_1,r_1,\ldots, r_1}_{\text{$\frac{p^n}{\ell(n)}$ times}}, 
 \underbrace{r_2,r_2,\ldots, r_2}_{\text{$\frac{p^n}{\ell(n)}$ times}}, \ldots, 
 \underbrace{r_{\ell(n)},r_{\ell(n)},\ldots, r_{\ell(n)}}_{\text{$\frac{p^n}{\ell(n)}$ times}} \enspace .\]
\qed
\end{proof}

\end{document}